\documentclass[acmsmall,nonacm]{acmart}

\usepackage{amsmath}
\usepackage{amsfonts}
\usepackage{amsthm}
\usepackage{booktabs}
\usepackage{graphicx}
\usepackage[capitalize]{cleveref}
\usepackage{url}
\usepackage{xcolor}
\usepackage{multirow}
\usepackage{tikz}
\usepackage{mathtools}
\usepackage{array}
\usepackage{ragged2e}
\usetikzlibrary{arrows.meta,positioning,shapes.geometric,fit,backgrounds,calc}
\pgfdeclarelayer{farback}\pgfsetlayers{farback,background,main}
\newtheorem{theorem}{Theorem}
\newtheorem{proposition}{Proposition}
\newtheorem{definition}{Definition}
\newtheorem{lemma}{Lemma}
\newtheorem{corollary}{Corollary}
\newtheorem{remark}{Remark}

\acmJournal{TEAC}
\acmVolume{0}
\acmNumber{0}
\acmArticle{0}
\acmYear{2026}
\acmMonth{0}
\setcopyright{rightsretained}

\begin{document}

\title{Credibility Trilemma in Polymatroidal Service Markets}

\author{Lauri Lov\'{e}n}
\orcid{0000-0001-9475-4839}
\email{lauri.loven@oulu.fi}
\affiliation{%
  \institution{Future Computing Group, University of Oulu}
  \city{Oulu}
  \country{Finland}}

\author{Sujit Gujar}
\affiliation{%
  \institution{Machine Learning Lab, International Institute of Information Technology}
  \city{Hyderabad}
  \country{India}}

\author{Kalle Timperi}
\affiliation{%
  \institution{Perception Engineering, University of Oulu}
  \city{Oulu}
  \country{Finland}}

\author{Hassan Mehmood}
\affiliation{%
  \institution{Future Computing Group, University of Oulu}
  \city{Oulu}
  \country{Finland}}

\author{Praveen Kumar Donta}
\affiliation{%
  \institution{Department of Computer and Systems Sciences, Stockholm University}
  \city{Stockholm}
  \country{Sweden}}

\author{Sasu Tarkoma}
\affiliation{%
  \institution{Department of Computer Science, University of Helsinki}
  \city{Helsinki}
  \country{Finland}}

\author{Schahram Dustdar}
\affiliation{%
  \institution{Distributed Systems Group, TU Wien}
  \city{Vienna}
  \country{Austria}}
\affiliation{%
  \institution{ICREA}
  \city{Barcelona}
  \country{Spain}}

\renewcommand{\shortauthors}{Lov\'{e}n et al.}

\begin{abstract}
Mechanism-mediated service markets with polymatroidal feasibility admit efficient, dominant-strategy incentive-compatible (DSIC) allocation, but these guarantees implicitly assume truthful \emph{execution} by the marketplace operator. Modelling the operator as a strategic player, we establish a credibility trilemma: for single-parameter agents on a non-modular polymatroid carrying a competitive prior profile, no static sealed-bid mechanism is simultaneously revenue-optimal, DSIC for agents, and credible for the operator. We introduce the Cost of Non-Credibility (CoNC) as a price-of-anarchy-style welfare-loss measure and prove a per-triple lower bound scaling with the polymatroid's conditional non-modularity gap; on released instances where capacity binds, its deviation extracts under 0.02\% of welfare and phantom insertion 0.9--2.2\%, neither destroying welfare. Two resolutions follow: public broadcast or deferred-revelation commitment, and administrative domain separation under settlement separation and four side conditions. On the supply side, an integrator paid a share of its slice's payments can profitably understate capacity, which settlement separation does not prevent and a known-cost envelope rule does. An instance-level grounding over the edge-pricing market of Amin et al.\ proves the trilemma there and shows that settlement separation closes its price-perturbation channel. Marketplace neutrality is thus a first-order design constraint on polymatroidal service markets.
\end{abstract}



\maketitle

\section{Introduction}
\label{sec:introduction}

Polymatroidal service markets (network slice brokering~\cite{sciancalepore2019slicebrokering,afolabi2018slicing}, multi-tenant cloud auctions with shared interconnect bandwidth, edge-priced capacity sharing on series--parallel networks~\cite{amin2026market}, and edge-cloud service-dependency markets~\cite{loven2026realtime}) admit dominant-strategy incentive-compatible (DSIC)\footnote{A mechanism is DSIC if truthful reporting is a dominant strategy for every agent regardless of others' reports.} allocation under gross-substitutes (GS) valuations: when the joint-capacity structure forms a polymatroid and valuations satisfy GS, a Walrasian equilibrium exists (for concave valuations on divisible goods, or on an integer polymatroid for indivisible slices), welfare maximisation over single-parameter valuations is polynomial-time solvable, and Vickrey--Clarke--Groves (VCG) or polymatroid clinching yields a DSIC, individually rational mechanism. These guarantees implicitly assume that the mechanism's \emph{executor}, the marketplace operator, implements the prescribed allocation and payments faithfully. When the operator is a strategic player with private information about supply, allocation discretion and custody of payment flows, the agent-side incentive analysis does not imply faithful execution: an operator that also integrates may misreport capacity, any operator may inject fictitious bids, mis-charge agents or selectively reveal information, and the DSIC welfare guarantees collapse without contradicting any individual agent's view. This paper addresses that \emph{credibility gap} for polymatroidal service markets; its central result is a single impossibility, the \emph{credibility trilemma} (\cref{thm:trilemma}):

\begin{quote}
\emph{On a non-modular polymatroid carrying a competitive prior profile, no static sealed-bid mechanism for single-parameter agents is simultaneously revenue-optimal, DSIC for agents, and credible for the operator; any two of the three are simultaneously achievable (the revenue-optimal credible pair under symmetric regular priors on a symmetric region), but not all three.}
\end{quote}

\noindent The rest of the paper establishes this claim, measures the resulting credibility gap and shows how to escape it; \cref{sec:integrator-credibility} treats the supply-side channel the trilemma takes as given. \cref{fig:trilemma} summarises the impossibility and its three pairwise witnesses.

\begin{figure}[t]
\centering
\resizebox{0.95\linewidth}{!}{%
\begin{tikzpicture}[scale=2.8]
  \coordinate (R) at (0,0);
  \coordinate (D) at (3,0);
  \coordinate (K) at (1.5,2.598);
  \draw[thick] (R)--(D)--(K)--cycle;
  \node[below left]  at (R) {\textbf{Revenue-optimality}};
  \node[below right] at (D) {\textbf{Agent DSIC}};
  \node[above]       at (K) {\textbf{Operator credibility}};
  \coordinate (RD) at ($(R)!0.5!(D)$);
  \coordinate (RK) at ($(R)!0.5!(K)$);
  \coordinate (DK) at ($(D)!0.5!(K)$);
  \fill[black!70!blue] (RD) circle (2.3pt);
  \node[below, font=\small, text width=3.2cm, align=center] at ($(RD)-(0,0.16)$) {Rev.-optimal $+$ DSIC\\[-1pt]\scriptsize Myerson / ironed-VV greedy\\[-1pt]\scriptsize\itshape (not credible)};
  \fill[black!65!red] (RK) circle (2.3pt);
  \node[left, font=\small, text width=2.9cm, align=center] at ($(RK)-(0.18,0)$) {Rev.-opt.\ $+$ credible\\[-1pt]\scriptsize First-price (pay-as-bid)\\[-1pt]\scriptsize\itshape (not DSIC)};
  \fill[black!60!green] (DK) circle (2.3pt);
  \node[right, font=\small, text width=2.9cm, align=center] at ($(DK)+(0.18,0)$) {DSIC $+$ credible\\[-1pt]\scriptsize Ascending clinching\\[-1pt]\scriptsize\itshape (not rev.-optimal)};
  \coordinate (ctr) at (1.5,0.866);
  \fill[red!8] (ctr) circle (0.52);
  \node[font=\small, text=red!60!black] at (ctr) {all three};
  \node[font=\scriptsize, text=red!60!black, align=center] at ($(ctr)-(0,0.26)$) {impossible\\(\cref{thm:trilemma})};
  \draw[thick, blue!60!black, dashed] (RD) to[bend left=15] (RK);
  \draw[thick, blue!60!black, dashed] (RD) to[bend right=15] (DK);
  \draw[thick, blue!60!black, dashed] (RK) to[bend left=15] (DK);
\end{tikzpicture}}
\caption{The credibility trilemma (\cref{thm:trilemma}). On a non-modular polymatroidal service market carrying a competitive prior profile, with single-parameter agents, any two of the three properties are simultaneously achievable (filled circles, each with a witness mechanism; the revenue-optimal credible pair on the symmetric class named below), but all three are not (shaded interior): the revenue-optimal DSIC mechanism (Myerson / ironed-virtual-value greedy) is not credible; the ascending clinching auction is credible and DSIC but not revenue-optimal; the first-price (pay-as-bid) auction with the Myerson reserve is credible and, for symmetric regular priors on a region symmetric in the agents, revenue-optimal in Bayes--Nash equilibrium, but not DSIC (\cref{rem:trilemma-pairwise-witnesses}; it is open whether a static credible revenue-optimal mechanism exists for all regular priors). Dashed curves indicate the achievable Pareto frontier. The two resolutions recover credibility by leaving the static sealed-bid class on which the impossibility rests (\cref{thm:commitment}) or by removing the operator's revenue motive (\cref{prop:domain-separation}).}
\Description{A triangle whose three corners are labelled revenue-optimality, agent DSIC and operator credibility. A filled dot sits at the midpoint of each edge, marking a witness mechanism that achieves the two adjacent properties but not the third: Myerson or the ironed-virtual-value greedy on the revenue-optimality--DSIC edge, the first-price pay-as-bid auction on the revenue-optimality--credibility edge, and the ascending clinching auction on the DSIC--credibility edge. The shaded interior of the triangle is labelled ``all three impossible''. Dashed curves join the three dots to indicate the achievable frontier.}
\label{fig:trilemma}
\end{figure}

As a running example we use a personal AI agent (PAA) navigating a smart building~\cite{saleh2025follow}, chosen for concreteness only: the results are stated and proved for the abstract polymatroidal structure and transfer to every realisation listed in the first paragraph. The agent runs a lightweight model on the user's phone but needs real-time services such as context inference, task planning and knowledge retrieval, composed from shared building resources spanning its sensors, edge infrastructure and a cloud backend. These services form a dependency chain (e.g., sensor data $\to$ local inference $\to$ reasoning) with interactive latency constraints: a useful response must arrive, for example, within $200$~ms~\cite{satya2017emergence,jiang2019lowlatency}. As the user moves between environments, the PAA discovers available edge services and composes ad-hoc service chains for the current context. An \emph{integrator} bundles resources from multiple providers into a \emph{service slice}, a capacity-constrained bundle that aggregates a sub-DAG of underlying services into one item priceable in the marketplace, and a marketplace operator offers such slices to agents (\cref{fig:paa-architecture}). The PAA case belongs to the emerging real-time AI service economy~\cite{deng2025agenticservicescomputing,derouiche2025agentic,10.1145/3773274.3777421}; the same credibility pressures arise in the other realisations wherever shared multi-tier capacity binds the allocation. Recent work~\cite{loven2026realtime} shows that when the service-dependency graph has tree or series--parallel structure, DSIC mechanisms achieve the welfare-maximising allocation in polynomial time. On a polymatroidal region the VCG allocation, the Edmonds greedy in bid order, is the exact per-epoch welfare maximiser and makes truthful reporting a dominant strategy (\cref{prop:efficient-mechanism}).

Akbarpour and Li~\cite{akbarpour2020credible} proved that no static sealed-bid single-item auction is simultaneously revenue-optimal, strategyproof, and \emph{credible}, i.e., proof against auctioneer deviations undetectable to any single bidder; the single-item case has since been resolved by Ferreira and Weinberg~\cite{ferreira2020credible} and the matroid case by Ganesh and Zhang~\cite{ec2025matroid}, both via commitment. This paper starts from polymatroidal feasible regions, which generalise matroids.

Because the welfare guarantee of the VCG allocation rests on truthful reports, and truthful reporting is a dominant strategy only against the mechanism as prescribed, faithful execution is the channel through which that guarantee reaches agents: if the operator deviates, agents lose their reason to bid truthfully, and the allocation degrades. We argue that truthful mechanism execution is itself a first-order mechanism-design problem, whose solution determines whether the theoretical guarantees hold in practice. Our contributions are:

\begin{enumerate}
    \item \emph{(The trilemma, and its quantitative counterpart.)} We formalise the marketplace operator as a strategic player and adapt the credible mechanism design framework to polymatroidal service markets, proving the credibility trilemma stated above. To measure how large the resulting credibility gap is, we introduce the Cost of Non-Credibility (CoNC), a welfare-loss measure analogous to the price-of-anarchy, and prove a per-triple lower bound scaling with the conditional non-modularity gap, the capacity two agents contest once the agents ranked above them are served.

    \item \emph{(Three structural results: commitment, separation, supply-side credibility.)} We show that commitment devices restore credibility (\cref{thm:commitment}); that administrative domain separation under settlement separation (\cref{prop:domain-separation}) provides an alternative credibility guarantee through revenue-channel separation, with a knife-edge at any positive ownership stake; and that a non-operator integrator paid any share of the payments on its slice profitably withholds capacity with positive probability, which neither settlement separation nor commitment prevents and an envelope rule does when its cost is known (\cref{prop:supply-credibility}).

    \item \emph{(Independent grounding.)} \cref{app:appendix-z} proves the trilemma as an instance over the edge-pricing market of Amin, Jaillet, Pulyassary, and Wu~\cite{amin2026market} (ACM TEAC 14(1), Art.~2) and shows that settlement separation closes the price-perturbation channel there; the mediator's deviation surplus there supplies a CoNC-type revenue-increment reading, while the $\gamma_{ij|S}$-structural identification behind the CoNC lower bound does not transfer to their primitives.

    \item \emph{(Computation on the companion's instances.)} On the certified tree and fan instances of~\cite{loven2026realtime}, with one task per agent, we compute the perturbation of \cref{lem:perturbation} exactly under VCG and the revenue-optimal mechanism, recover its payment increment at every realisation, test its forged form against the clinching transcript, and size it: the perturbation raises the operator's revenue by about $0.003\%$ to $0.02\%$ of expected welfare, and phantom insertion by about $0.9\%$ to $2.2\%$, both without welfare loss. The broadcast verifiers reject every transcript mutation the proof of \cref{thm:commitment} names and admit the phantom once the roster check (S3) is removed, and removing settlement separation (C0), the no-stake condition (C1) or the no-affiliation condition (C2) of \cref{prop:domain-separation} makes the cases its proof assigns to them profitable.
\end{enumerate}

A continuous deployability surface in stake and escrow is outside the paper's scope; the scope conditions of the formal results are collected in \cref{rem:scope-consolidated}.

\begin{figure}[!t]
\centering
\begin{tikzpicture}[
  >=Latex, font=\footnotesize,
  tier/.style={draw, rounded corners=2pt, minimum width=15mm, minimum height=8mm, align=center, inner sep=2pt, fill=white},
  agent/.style={draw, circle, minimum size=11mm, inner sep=1pt, align=center, fill=white},
  op/.style={draw, double, rounded corners=2pt, minimum width=21mm, minimum height=12mm, align=center, fill=gray!12},
  econ/.style={->, thick},
  data/.style={->, gray, dashed, semithick}
]
  \node[tier] (dev) {Sensor\\\scriptsize bldg.\ sensing};
  \node[tier, right=4mm of dev] (edge) {Edge\\\scriptsize bldg.\ infer.};
  \node[tier, right=4mm of edge] (cloud) {Cloud\\\scriptsize FM reason.};
  \node[op, right=20mm of cloud] (mop) {Marketplace\\operator $\mathcal{O}$};
  \node[agent, minimum size=13mm, right=26mm of mop] (paa) {agent $i$\\(PAA)\\{\scriptsize user's phone}};
  \draw[data] (dev) -- (edge);
  \draw[data] (edge) -- (cloud);
  \draw[econ] (cloud.east) -- node[below, font=\scriptsize] {slice cap.\ $\bar C_k$} (mop.west);
  \draw[econ] (paa.north west) to[bend right=20] node[above] {bid $b_i$, payment $p_i$} (mop.north east);
  \draw[econ] (mop.south east) to[bend right=20] node[below] {allocated slice $x_i$} (paa.south west);
  \begin{pgfonlayer}{background}
    \node[draw, dashed, rounded corners=3pt, fit=(dev)(edge)(cloud), inner sep=6pt] (intg) {};
  \end{pgfonlayer}
  \node[font=\scriptsize\itshape, below=3pt of intg] (intglbl) {integrator $k$: service slice $=$ sub-DAG bundle};
  \begin{pgfonlayer}{farback}
    \node[draw=orange!75!black, dashed, very thick, rounded corners=5pt, fill=orange!9,
      fit=(intg)(mop)(intglbl), inner sep=10pt,
      label={[font=\scriptsize\itshape, text=orange!45!black]above:integrator and operator may be one entity $\Rightarrow$ credibility gap}] (gap) {};
  \end{pgfonlayer}
\end{tikzpicture}
\caption{Running PAA scenario and the credibility gap. Agent~$i$, a personal AI agent running on the user's phone, submits a bid $b_i$ and pays $p_i$ to the \emph{marketplace operator}~$\mathcal{O}$, which returns the allocated \emph{service slice} $x_i$ (solid economic flows; the agent's quasi-linear utility is $v_i x_i - p_i$). An \emph{integrator}~$k$ bundles a building-side sensor--edge--cloud sub-DAG (grey dashed service-dependency arrows; the shared, capacity-constrained resources) into a single service slice of capacity $\bar C_k$, offered to~$\mathcal{O}$. The credibility problem is most acute when the integrator and operator are the same legal entity (orange enclosure; e.g., the building's edge provider): the entity that runs the auction also controls and prices the supply. The user's phone is only the agent's endpoint, outside the priced slice.}
\Description{A diagram showing a personal AI agent (PAA) on a user's phone as an endpoint outside the priced slice, connected via bid and payment arrows to a marketplace operator, which is fed slice capacity by an integrator that bundles a building-side sensor--edge--cloud sub-DAG (dashed integrator box); the operator returns the allocated slice to the PAA.}
\label{fig:paa-architecture}
\end{figure}

The polymatroidal feasibility region builds on our earlier work~\cite{loven2026realtime}, a service-economy framework on tree and series--parallel dependency DAGs whose fork-join counting the routed reading used here matches on trees and series chains. \cref{prop:polymatroidal-structure,prop:efficient-mechanism,prop:encapsulation-p3} in \cref{sec:background} restate its structural results on the routed reading, the first two supported by the refereed literature and the third proved in full, so the results below do not rest on~\cite{loven2026realtime}; it is cited for context, for the fork-join reading reconciled in \cref{sec:bg-dags}, for the inner exposures of \cref{sec:integrator-credibility}, and for the released instances and simulator on which \cref{sec:evaluation} computes.

\cref{sec:background} sets out the background. \cref{sec:credible-mechanisms} presents the theoretical results: the trilemma, the CoNC lower bound, the two credibility resolutions and integrator credibility. \cref{sec:evaluation} computes the trilemma's deviation, its cost and the resolutions' closure conditions on the companion's certified instances. \cref{sec:discussion} covers related work, practical implications, limitations and future directions, and \cref{sec:conclusion} concludes.
\section{Background}
\label{sec:background}

\subsection{Polymatroidal Markets and Gross Substitutes}
\label{sec:bg-dags}

A \emph{polymatroidal service market} is a market whose feasible allocation region is a polymatroid and whose agents have gross-substitutes (GS) valuations. The polymatroid encodes the joint capacity structure of a shared underlying resource graph; the GS condition encodes substitutability across the agent-facing items. The credibility analysis operates on this abstract structure, so it transfers to each realisation listed in \cref{sec:introduction}; the service-dependency DAG setting of~\cite{loven2026realtime} is the running formalism.

In the running formalism, agents generate latency-sensitive tasks that require multi-resource service compositions across a sensor--edge--cloud continuum. Each service composition is captured by a \emph{service-dependency DAG} $G_{\mathrm{res}}$ whose nodes (the set $\mathcal{R}$) represent service types and whose edges represent dependencies. Agents consume leaf services, and the leaves form the ground set $E$ of the polymatroids below; internal nodes impose capacity constraints.

\paragraph*{Definitions.} A \emph{polymatroid} on ground set $E$ is defined by a rank function $f: 2^E \to \mathbb{R}_{\ge 0}$ that is monotone non-decreasing, submodular ($f(A \cup B) + f(A \cap B) \le f(A) + f(B)$), and satisfies $f(\emptyset) = 0$; the associated polymatroid is $\{x \in \mathbb{R}^E_{\ge 0} : x(S) \le f(S) \; \forall S \subseteq E\}$. A matroid is the special case where allocations are binary ($x_i \in \{0,1\}$). Valuations satisfy the \emph{gross substitutes} (GS) condition~\cite{kelso1982job,bikhchandani1997competitive} if raising the price of one item never decreases demand for other items.

Lov\'{e}n et al.~\cite{loven2026realtime} establish the structural results that P1--P3 below restate: the mechanism result P2 and the fork-join counterparts of P1 and P3 (its Propositions~1 and~3). We restate P1 and P3 on the routed reading reconciled below, cite the refereed sources on which P1 and P2 rest, and prove P3 inline.

\begin{proposition}[Polymatroidal structure, routed reading (after P1)]
\label{prop:polymatroidal-structure}
When $G_{\mathrm{res}}$ is a rooted tree or series--parallel (SP) network with positive internal-node capacities $\{C_r\}_{r\in\mathcal{R}}$, read as a routed network (each allocated unit travels one leaf-to-root path), the feasible allocation region $\mathcal{X}_{\mathrm{res}} = \{x \in \mathbb{R}^E_{\ge 0} : x(S) \le f(S) \;\forall S\subseteq E\}$ is a polymatroid, where the rank function $f(S)$ is the max-flow capacity from the leaves in $S$ to the root subject to the internal-node capacity constraints.
\end{proposition}

\noindent\emph{Sources.} With node capacities absorbed by the standard node-splitting reduction, the max-flow rank $f$ from a leaf set to the root is normalised, monotone and submodular, the polymatroid induced by a capacitated network~\cite[\S44.6g]{schrijver2003combinatorial}, and an allocation $x\ge 0$ admits a simultaneous feasible flow if and only if $x(S)\le f(S)$ for every $S\subseteq E$~\cite{fujishige2005submodular}, which identifies $\mathcal{X}_{\mathrm{res}}$ with the polymatroid of $f$. On this routed reading the tree/SP hypothesis is not consumed by the polymatroid property, which holds for arbitrary capacitated $G_{\mathrm{res}}$; it is retained because the rank oracle is then linear-time (\cref{prop:efficient-mechanism}(ii)) and \cref{prop:encapsulation-p3} reduces to it. On the companion's fork-join reading the hypothesis is consumed (next paragraph).

\paragraph*{Routed and fork-join readings.} The companion~\cite{loven2026realtime} charges each throughput token to every service node above its leaf (fork-join composition, the service-feasibility region of its Definition~1), so its region is the leaf-block region $\{x\ge 0 : x(L_v)\le C_v \text{ for every service node } v\}$, with $L_v$ the leaves below $v$, which its Proposition~1 shows to be a polymatroid when the leaf blocks are laminar, as on rooted trees and SP networks, and not in general otherwise. A routed unit loads a subset of the nodes the fork-join count charges, so the fork-join region lies inside the routed one; the two coincide on rooted trees and series chains, and at a parallel stage the routed region is larger, as the companion notes with the same example: two parallel capacity-$1$ nodes feeding one service of capacity $2$ admit one unit under fork-join counting and two when routed. Accordingly, for a cluster whose exported units all load the same internal nodes, the companion exports their least capacity (its Lemma~2, the scalar interface of its Proposition~3), and the cluster's max-flow equals that bottleneck if and only if a least-capacity loaded node lies on every internal route, since every entry-to-exit node cut consists of loaded nodes, each of capacity at least the bottleneck, and a cut of two or more positive-capacity nodes exceeds it. The companion's Proposition~3, the fork-join counterpart of P3, rests on four conditions: (i)~the exported scalar is the least capacity among the internal service nodes that carry the export, each loaded exactly once by every admitted external unit, (ii)~the exported units are interchangeable in the internal schedule, (iii)~no external feasibility constraint couples different integrators beyond the quotient graph, and (iv)~each integrator presents one entry node and one exit node~\cite{loven2026realtime}. The results of \cref{sec:credible-mechanisms} use only that the region is the polymatroid of $f$, so they hold on either reading wherever it yields one; we proceed on the routed reading, which is also that of \cref{app:appendix-z}.

\begin{proposition}[Efficient mechanism design (after P2)]
\label{prop:efficient-mechanism}
If $\mathcal{X}_{\mathrm{res}}$ is polymatroidal and agents' valuations satisfy the gross substitutes (GS) condition, then (i)~a Walrasian equilibrium exists, for divisible goods when valuations are concave and, when slices are indivisible, whenever $f$ is integer-valued; (ii)~for single-parameter valuations, welfare maximisation is polynomial-time solvable; (iii)~the efficient allocation is implementable via a DSIC, individually rational mechanism: VCG~\cite{ausubel2005vickrey} applies in full generality, and \emph{in the single-type-per-task, unit-demand case} (where each task class maps to a unique admissible slice type and each agent demands at most one unit) the ascending clinching auction implements the same allocation and the same payments in an ascending format whenever the servable 0/1 points form a matroid, as when $f$ is integer-valued~\cite{ausubel2004ascending,bikhchandani2011matroid}. All guarantees hold within each decision epoch; cross-epoch strategic dynamics lie outside their scope.
\end{proposition}

\noindent\emph{Sources.} Claim~(i) is Walrasian-equilibrium existence under gross substitutes: for divisible goods with concave valuations by convex duality (linear-programming duality for single-parameter linear valuations), and for indivisible slices by discrete convex analysis, since GS valuations over distinct items are M$^{\natural}$-concave~\cite{fujishige2003note} and the integer points of an integer polymatroid form an M$^{\natural}$-convex supply set~\cite{murota2003discrete,fujishige2005submodular}, the companion's route~\cite{loven2026realtime}; on box supply it is the result of Kelso and Crawford~\cite{kelso1982job}, and Gul and Stacchetti~\cite{gul1999grosssubstitutes} show GS necessary as well as sufficient on the relevant valuation domain. Claim~(ii), for single-parameter valuations, is Edmonds' greedy algorithm for linear optimisation over a polymatroid~\cite{edmonds1970submodular}, $x_i^{*}=f(S_i\cup\{i\})-f(S_i)$ in decreasing-bid order, at cost $O(n\log n+n\,T_f)$ for rank-oracle cost $T_f$, with $T_f$ linear in $|\mathcal{R}|$ on the tree and SP cases of \cref{prop:polymatroidal-structure}. Claim~(iii) is the VCG theorem~\cite{vickrey1961counterspeculation,clarke1971multipart,groves1973incentives,ausubel2005vickrey}; in the single-type-per-task, unit-demand case Ausubel's ascending clinching auction~\cite{ausubel2004ascending} terminates at the Walrasian equilibrium of~(i), its polymatroid extension~\cite{goel2015polyhedral} covers divisible goods, and on a matroid the ascending clinching prices coincide with the VCG payments~\cite{bikhchandani2011matroid}.

\begin{proposition}[Encapsulation, routed reading (after P3)]
\label{prop:encapsulation-p3}
For arbitrary DAGs $G_{\mathrm{res}}$, cross-domain \emph{integrators} can partition $G_{\mathrm{res}}$ into clusters, each exposing a single composite service (slice) with capacity equal to the max-flow of its sub-DAG. If the resulting quotient graph $G'$ is tree or series--parallel and the encapsulation conditions E1--E4 hold (scalar capacity summary, cut-exact max-flow, no cross-slice complementarity, confluent clusters; full assumption-applicability table in \cref{app:structural-assumptions}), the agent-facing feasible region is polymatroidal and coincides with the set of allocations $G_{\mathrm{res}}$ can deliver when each unit is routed along one internal path.
\end{proposition}

\begin{proof}
Let the vertex set of $G_{\mathrm{res}}$ admit a partition $V = V_1 \sqcup \cdots \sqcup V_K$, where each cluster $V_k$ is the responsibility of integrator $k$ and induces a sub-DAG $G_k=(V_k, E_k)$ with internal capacities inherited from $G_{\mathrm{res}}$. Inter-cluster edges retain their original capacities. Define the slice capacity
\[
\bar C_k \;=\; \max\text{-flow}(G_k) \;=\; \min_{C \text{ cut of }G_k}\ \mathrm{cap}(C),
\]
i.e., the value of a maximum flow inside $G_k$ from its \emph{contact set} $\mathrm{con}(V_k)$ --- the vertices of $V_k$ that receive an inter-cluster edge, together with the leaves of $G_{\mathrm{res}}$ contained in $V_k$ --- to its boundary exit $\omega_k$ (E4 below). Define the \emph{quotient graph} $G' = (V', E')$ by contracting each $V_k$ to a single node $\nu_k$ of capacity $\bar C_k$, retaining all inter-cluster edges with their original capacities, and absorbing each contracted node by the standard node-splitting reduction so that the flow through $\nu_k$ is bounded by $\bar C_k$. The agent-facing leaf set $E^{\mathrm{slice}}$ collects one slice token per leaf-bearing cluster (or per slice type, when integrators expose multiple slice families); clusters containing no leaf expose no token, and $\iota$ maps each token to the leaves inside its cluster.

\emph{(i) Encapsulation conditions.} The conditions E1--E4 specialise the contraction to information-preserving form.
\begin{itemize}
\item \textbf{E1 (scalar capacity).} Each integrator's slice exposes a single non-negative scalar $\bar C_k$ to the agent-facing market; no internal structure of $G_k$ is observable.
\item \textbf{E2 (cut-exact summary).} $\bar C_k$ equals the actual max-flow of $G_k$, and it does so subset-wise: for every non-empty $T\subseteq\mathrm{con}(V_k)$ the max-flow inside $G_k$ from $T$ to $\omega_k$ is again $\bar C_k$, so no proper subset of the cluster's contacts is cheaper to sever than the whole. A scalar that matches only the full-contact max-flow can hide a bottleneck that binds on a proper subset, and the contraction then advertises capacity the cluster cannot route. E2 is this paper's condition for the routed reading, not a form of the companion's fork-join hypotheses~\cite{loven2026realtime}.
\item \textbf{E3 (no cross-slice complementarity).} An agent's value for slice $k$ is additively separable from its value for slice $k'\ne k$: $v_i(x_i) = \sum_{k} v_{i,k}(x_{i,k})$. The condition constrains complementarity within an agent's own bundle; it places no restriction on externalities across agents.
\item \textbf{E4 (confluent cluster).} Every edge from $V_k$ to $V\setminus V_k$ leaves a single \emph{boundary exit} $\omega_k\in V_k$, and every vertex of $V_k$ lies on an internal path from $\mathrm{con}(V_k)$ to $\omega_k$. A source cluster, which contains every leaf it exposes and receives no inter-cluster edge, satisfies E4 as soon as it has one exit.
\end{itemize}

\emph{(ii) Information-preserving contraction.} We claim that $f'(S) = f_{G_{\mathrm{res}}}(\iota(S))$ for every $S\subseteq E^{\mathrm{slice}}$, where $f'$ is the rank function of $G'$. Both inequalities are cut arguments, and each consumes one of E2 and E4.

\emph{Advertised capacity is not below deliverable capacity.} Let $C'$ be any cut of $G'$ separating $\{\nu_k\}_{k\in S}$ from the root, and let $C^*$ be obtained from $C'$ by replacing each contracted node $\nu_k\in C'$ with a minimum cut of $G_k$ separating $\mathrm{con}(V_k)$ from $\omega_k$, the inter-cluster edges of $C'$ retaining their capacities. Each replacement costs exactly $\bar C_k$, so $\mathrm{cap}_{G_{\mathrm{res}}}(C^*) = \mathrm{cap}_{G'}(C')$. Moreover $C^*$ is a cut of $G_{\mathrm{res}}$: a path from a leaf in $\iota(S)$ to the root projects onto a path in $G'$ from some $\nu_k$, $k\in S$, to the root, which meets $C'$; if it meets an inter-cluster edge of $C'$ it meets $C^*$, and if it meets some $\nu_j\in C'$ then by E4 the portion of the path inside $V_j$ runs from a contact of $V_j$ to the single exit $\omega_j$ and is therefore severed by the internal cut substituted for $\nu_j$. Minimising over $C'$ gives $f_{G_{\mathrm{res}}}(\iota(S))\le f'(S)$.

\emph{Advertised capacity is not above deliverable capacity.} Let $C$ be an inclusion-minimal minimum cut of $G_{\mathrm{res}}$ separating $\iota(S)$ from the root, and let $C'$ consist of the nodes $\nu_k$ with $C\cap V_k\ne\emptyset$ together with the inter-cluster edges of $C$. Then $C'$ is a cut of $G'$: lifting any $G'$ path from $\nu_{k_0}$, $k_0\in S$, to the root --- inside each cluster from a contact to its exit, which E4 guarantees exists --- produces a leaf-to-root path of $G_{\mathrm{res}}$, which $C$ meets, so $C'$ meets the $G'$ path it was lifted from. For the capacity, fix $k$ with $C\cap V_k\ne\emptyset$ and let $T$ be the set of contacts of $V_k$ that $\iota(S)$ still reaches in $G_{\mathrm{res}}$ once $C\setminus V_k$ is removed. If $T$ were empty, no source flow would enter $V_k$ and $C\cap V_k$ would be redundant. If instead some $t\in T$ still reached $\omega_k$ after removing all of $C$, then no path from $\omega_k$ to the root could survive $C$, since otherwise $C$ would not be a cut; but E4 makes $\omega_k$ the only way out of $V_k$, so every path leaving the cluster is already severed downstream and $C\cap V_k$ is again redundant. Inclusion-minimality therefore forces $T\ne\emptyset$ and $C\cap V_k$ to separate $T$ from $\omega_k$, so by E2 its capacity is at least $\bar C_k$. Summing over the clusters $C$ meets, $\mathrm{cap}_{G'}(C')\le\mathrm{cap}_{G_{\mathrm{res}}}(C)$, whence $f'(S)\le f_{G_{\mathrm{res}}}(\iota(S))$ and the two agree; since the blocks $\iota(k)$ are disjoint, the image of the routed region under $z\mapsto(z(\iota(k)))_k$ is the polymatroid of $f'$~\cite[\S44.6g]{schrijver2003combinatorial}, which is the coincidence claimed.

Both hypotheses are load-bearing, and dropping either lets contraction manufacture routes. Without E4, take internal nodes $u,v$ of capacity $1$ with no edge between them, feeding downstream nodes of capacity $1$ and $0.1$ respectively; the cluster $\{u,v\}$ has max-flow $2$, and the token whose leaf feeds $v$ is credited $\min(2,\,1+0.1)=1.1$ against a deliverable $0.1$. Without E2's subset clause the same gap opens at a single-exit cluster whose contact vertices carry capacities of their own. In both cases the advertised region strictly contains the deliverable one, and \cref{prop:efficient-mechanism} would clear allocations that $G_{\mathrm{res}}$ cannot route.

Submodularity, monotonicity, and normalisation transfer from $f_{G_{\mathrm{res}}}$ to $f'$ unchanged~\cite{fujishige2005submodular}, so $f'$ is itself a polymatroid rank function on $E^{\mathrm{slice}}$.

\emph{(iii) Reduction to \cref{prop:polymatroidal-structure}.} If the quotient graph $G'$ is a rooted tree or SP network, \cref{prop:polymatroidal-structure} applied to $G'$ yields that the agent-facing feasible region
\[
\mathcal{X}'_{\mathrm{res}} \;=\; \{\,x\in\mathbb{R}^{E^{\mathrm{slice}}}_{\ge 0} : x(S)\le f'(S)\;\forall S\subseteq E^{\mathrm{slice}}\,\}
\]
is a polymatroid. Combined with E3 (additive separability across slice types), the agent-facing market is a polymatroidal mechanism design problem to which \cref{prop:efficient-mechanism} applies.
\end{proof}

GS valuations arise under three conditions: (GS1)~unit demand, (GS2)~additive separability, (GS3)~fixed attributes within each epoch (the hypotheses of the companion's valuation lemma~\cite{loven2026realtime}). Integrator encapsulation (\cref{prop:encapsulation-p3}) supplies (GS2) and (GS3) by bundling multi-resource paths into substitutable slices, subject to E1--E4; (GS1) is a modelling restriction on the task-to-slice mapping, imposed on those slices and not established by P3. Unless a result states otherwise, the credibility results below are read on a feasible region supplied by P1--P3 with GS1--GS3 in force.

\paragraph*{Type spaces at two marketplace tiers.} Level~1 (cross-domain slice marketplace) is single-parameter and matroid: encapsulation (P3) with GS1 makes each agent unit-demand on composite slices ($x_i \in \{0,1\}$), so the Archer--Tardos characterisation~\cite{archer2001truthful} applies. Level~2 (within-domain raw-resource marketplace) is multi-unit polymatroidal ($x_i \ge 0$), requiring the GS/clinching framework~\cite{gul1999grosssubstitutes,goel2015polyhedral}.

\subsection{Credible Mechanism Design}
\label{sec:bg-credible}

The results above guarantee that no \emph{agent} benefits from misreporting its type but are silent on the behaviour of the entity that \emph{executes} the mechanism.

Akbarpour and Li~\cite{akbarpour2020credible} formalised this concern as \emph{credibility}: a mechanism is credible if the auctioneer has no incentive to deviate from the prescribed protocol, given that deviations must be undetectable to any individual bidder. They proved a trilemma for single-item auctions: no static, sealed-bid mechanism is simultaneously (i)~revenue-optimal, (ii)~strategyproof for bidders, and (iii)~credible for the auctioneer. Among optimal auctions (with the optimal reserve, under symmetric independent private values), the ascending (English) auction is the unique credible strategyproof mechanism, but it is not static; the first-price auction is the unique credible static mechanism, but it is not strategyproof. Myerson's optimal mechanism~\cite{myerson1981optimal} is the benchmark for property~(i): it maximises expected revenue over regular distributions.

Ferreira and Weinberg~\cite{ferreira2020credible} show that a \emph{deferred-revelation auction} (DRA), implemented over a secure and censorship-resistant blockchain, achieves credibility, strategyproofness, and revenue optimality for strongly regular distributions. Chitra et al.~\cite{chitra2024credible} generalise the DRA to any public broadcast channel, removing the blockchain requirement. Ganesh and Zhang~\cite{ec2025matroid} extend the DRA to matroid feasibility constraints, proving that it satisfies all three properties when bidder values are drawn from $\alpha$-strongly regular distributions and that it is not credible beyond matroid feasibility.

\subsection{The Truthful Execution Gap}
\label{sec:bg-gap}

We now make explicit what P1--P3, and the mediator-faithful baseline of \cite{amin2026market}, implicitly assume. We use \emph{marketplace operator} for the entity that runs the auction (collects bids, computes allocations, charges payments) and \emph{integrator} for the entity that composes a sub-DAG into a slice (P3).

\paragraph*{Information model.} In the sealed-bid setting, each agent $i$ observes only its own bid $b_i$, allocation $x_i$, and payment $p_i$ (no public transcript); in the ascending setting, agent $i$ additionally observes the broadcast price clock and clinching events, and undetectability requires consistency with the public record. The DSIC guarantee of P2 holds when:
\begin{enumerate}
    \item[\textbf{A1.}] The operator collects all bids without alteration.
    \item[\textbf{A2.}] The operator computes the welfare-maximising allocation (by the Edmonds greedy of \cref{prop:efficient-mechanism}(ii), or by the ascending clinching auction in the unit-demand case of \cref{prop:efficient-mechanism}(iii)).
    \item[\textbf{A3.}] The operator charges VCG payments (or clinching-auction prices) exactly as prescribed.
    \item[\textbf{A4.}] The operator does not selectively reveal bid information to affiliated agents.
\end{enumerate}

In a multi-domain service economy, the entities that run and supply the marketplace may violate any of A1--A4. We identify four deviations, one by an integrator and three by the operator, each named with the assumption it violates:

\begin{itemize}
    \item \textbf{Capacity misreporting (violates A2 relative to the deliverable region).} An integrator may understate the exported capacity $\bar{C}_k$ it exposes (\cref{prop:encapsulation-p3}), creating artificial scarcity and raising the payments of the agents it serves; when the integrator is also the operator, the gain accrues to the entity that runs the auction. The results of \cref{sec:trilemma,sec:commitment,sec:alternatives} take the exposed region as given, and \cref{sec:integrator-credibility} treats this channel for both the dual-role and the non-operator integrator.

    \item \textbf{Discriminatory allocation (violates A2).} The operator may allocate off-equilibrium and below the welfare maximum, favouring affiliated agents over higher-WTP outsiders.

    \item \textbf{Price manipulation (violates A3).} Because losing bids are not visible to agents under the sealed-bid information model, the operator can compute payments that differ from the VCG/clinching prescription without contradicting any agent's view of the outcome. Shill bidding is the canonical implementation, and a phantom bid raises the Archer--Tardos integrand directly; \cref{lem:perturbation} shows that raising an existing bidder's report suffices.

    \item \textbf{Selective information revelation (violates A4).} The operator may share bid information with affiliated agents, who then form an effective coalition with the operator for surplus-extracting joint deviations.
\end{itemize}

Under truthful bidding the value-greedy allocation (full valuation visibility, no prices) coincides with the VCG allocation: when the leaf-block family is laminar and demand is unit, both are the Edmonds greedy in value order~\cite{edmonds1970submodular}, exact for the static per-epoch allocation, and the companion reports that under dynamic congestion value-greedy can still fall short of the welfare optimum~\cite{loven2026realtime}. Truthfulness is a property of the VCG allocation on a region certified polymatroidal, on which the companion's mechanism proposition makes truthful reporting of an agent's entire type a dominant strategy; the companion's price-discovering comparison market is manipulable under contention, most often by over-reporting values~\cite{loven2026realtime}. That property is what the heuristic would need and cannot obtain across trust boundaries.
\section{Credible Mechanisms for Polymatroidal Service Markets}
\label{sec:credible-mechanisms}

\subsection{Operator Model}
\label{sec:operator-model}

We extend the mechanism design model of~\cite{loven2026realtime} with the marketplace operator as a strategic player. The polymatroidal service market hosts four formally distinct roles, which real deployments may collapse into one legal entity in various combinations.

\begin{definition}[Role taxonomy of the polymatroidal service market]
\label{def:role-taxonomy}
The market involves four roles:
\begin{enumerate}
    \item \emph{Resource owner.} Supplies the underlying capacity (compute, bandwidth, storage) into the polymatroidal feasible region $\mathcal{X}_{\mathrm{res}}$, and receives wholesale payments outside the auction loop (under settlement separation, the agent-to-owner transfer transits the operator without entering its books; see (C0) of \cref{prop:domain-separation}).
    \item \emph{Integrator.} Composes a sub-DAG of resources into a sellable slice (\cref{prop:encapsulation-p3}, after~\cite{loven2026realtime}) and fixes, before bidding, the agent-facing region it exposes (exported capacity $\bar C_k$, slice-type cardinality, eligibility set); that region is an input to the operator's mechanism, and its choice is analysed in \cref{sec:integrator-credibility}.
    \item \emph{Marketplace operator.} Runs the auction over the polymatroid of slices: receives bids $\mathbf{b}=(b_1,\ldots,b_n)$, computes an allocation $x\in\mathcal{X}_{\mathrm{res}}$ and payments $p=(p_1,\ldots,p_n)$. The strategic player whose deviation $\delta$ is the subject of \cref{thm:trilemma}.
    \item \emph{Agent.} The buyer, with private valuation $v_i$ over slices and quasi-linear utility $v_i x_i - p_i$.
\end{enumerate}
Real deployments may collapse multiple roles into a single legal entity: \emph{resource owner = integrator} in vertically integrated cloud providers; \emph{integrator = operator} in the dual-role case central to our credibility analysis (an edge provider that hosts inference and runs the slice marketplace). The trilemma is most consequential when roles~2 and~3 are united in one strategic entity, because that entity simultaneously holds private information about supply, controls the auction's allocation rule, and has custody of payment flows.
\end{definition}

Throughout we let $\mathcal{O}$ denote the marketplace operator of \cref{def:role-taxonomy}(3) over polymatroidal feasible region $\mathcal{X}_{\mathrm{res}}$, and $\mathcal{A}$ the set of $n$ agents. The operator receives the bid vector $\mathbf{b}$, computes an allocation and payment, and earns revenue $\sum_i p_i$ minus the cost of procuring the allocated resources.

\begin{definition}[Operator deviation]
\label{def:deviation}
An \emph{operator deviation} is a strategic action by the marketplace operator on the received bid profile $\mathbf{b}$. The operator's \emph{action space} comprises: (a) evaluating the prescribed mechanism $(x^*,p^*)$ on $\mathbf{b}$ and returning its honest output; or (b) constructing a counterfactual bid vector $\hat{\mathbf{b}}$ (by phantom-bidder insertion, losing-bid concealment, bid substitution, or arbitrary modification of $\mathbf{b}_{-i}$ for any $i$), evaluating $(x^*,p^*)$ on $\hat{\mathbf{b}}$, and returning the resulting per-agent outcomes; or (c) any mixture of (a) and (b) including the freedom to inflate or deflate $i$'s payment by an exogenous amount $\varepsilon_i$ that is rationalised by an honest execution on \emph{some} bid profile $\hat{\mathbf{b}}_{-i}$ with $\hat{b}_i=b_i$. A deviation is \emph{population-preserving} when every counterfactual profile it evaluates is a profile of the realised agent set, so that the operator alters only the reports of agents that actually bid; losing-bid concealment is of this kind, a losing bid being concealed by reporting it below the reserve. It is \emph{population-augmenting} when some counterfactual profile carries an identity outside the realised agent set, as in phantom-bidder insertion or identity substitution: the identity is drawn from the ground set $E$ but did not bid in this realisation, so $\hat{\mathbf{b}}_{-i}$ remains a point of $\mathrm{supp}(F_{-i})$ and the support criterion below applies unchanged to both classes. Formally, a deviation is a mapping $\delta: \mathbf{b} \mapsto (x', p')$ that differs from the prescribed mechanism $(x^*, p^*)$ and is \emph{undetectable}: no individual agent $i$ can distinguish its observed $(x'_i, p'_i)$ from $(x^*_i, p^*_i)$ given only its own bid $b_i$ and outcome, where undetectability is the Akbarpour--Li support criterion: agent $i$ accepts $(x_i,p_i)$ iff some $\hat{\mathbf b}_{-i}\in\mathrm{supp}(F_{-i})$ with $\hat b_i=b_i$ rationalises it (\cref{rem:akbarpour-li-alignment}).
\end{definition}

The impossibility results use population-preserving deviations only. The perturbation of \cref{lem:perturbation}, which drives \cref{thm:trilemma}, \cref{cor:conc-lb} and \cref{cor:constructive-trilemma}, raises one existing agent's bid inside the prior support and inserts nobody, so the trilemma holds already against the smaller preserving class. The resolutions close the augmenting class, each by a named condition: (S3) of \cref{def:broadcast-channel} for \cref{thm:commitment}, and (C0), (C2) and (C4) in case~5 of \cref{prop:domain-separation}.

\begin{definition}[Credible mechanism: realisation-wise and ex-ante variants]
\label{def:credible}
A mechanism $\mathcal{M}=(x^{*},p^{*})$ is \emph{realisation-wise credible} if no operator deviation $\delta$ yields strictly higher revenue than $\mathcal{M}$ for any realisation of bids $\mathbf{b}$. It is \emph{$\varepsilon$-credible-ex-ante} (for $\varepsilon\ge 0$) if no operator deviation $\delta$ yields expected revenue exceeding $\mathbb{E}[\mathrm{rev}^{*}(\mathcal{M})]+\varepsilon$ where the expectation is over the bid prior. Realisation-wise credibility is the strictly stronger notion: it implies $0$-credibility-ex-ante, while ex-ante credibility does not imply realisation-wise credibility. Throughout, ``credible'' without qualifier denotes the realisation-wise variant; the trilemma of \cref{thm:trilemma} is a realisation-wise impossibility.
\end{definition}

The key distinction from agent incentive compatibility is that the operator observes all bids and can modify both the allocation and the payments, subject only to the undetectability constraint of the information model in \cref{sec:bg-gap}.

\subsection{Credibility Trilemma}
\label{sec:trilemma}

We extend the Akbarpour--Li trilemma of \cref{sec:bg-credible} from single-item auctions to polymatroidal feasible regions. A polymatroid is \emph{non-modular} when agents share capacity at some prefix: $f(S\cup\{i\})+f(S\cup\{j\})>f(S)+f(S\cup\{i,j\})$ for some pair $i,j$ and some $S\subseteq E\setminus\{i,j\}$. This holds exactly when $f$ is not additive, since if every such difference vanishes, each marginal $f(S\cup\{i\})-f(S)$ equals $f(\{i\})$. Non-modularity is the structural condition that enables the operator deviation constructed below.

\begin{definition}[Contested window]
\label{def:contested-window}
Let $f$ be a non-modular polymatroid on ground set $E$ and let $F=\prod_i F_i$ be independent non-degenerate regular priors with interval supports; write $\bar{\varphi}_i$ for the (ironed) virtual value of $F_i$~\cite{myerson1981optimal} and $\underline{v}_i \triangleq \inf\mathrm{supp}(F_i)$. Allocations are computed by the Edmonds greedy in a priority order $\pi$, in one of two branches: \textup{(W)} the welfare branch, $\pi=$ decreasing bids; \textup{(V)} the virtual branch, $\pi=$ decreasing $\bar{\varphi}_l(b_l)$ over the active agents ($\bar{\varphi}_l(b_l)\ge 0$), with each $\bar{\varphi}_l$ strictly increasing and continuous on the bids involved. On branch \textup{(W)} read every $\bar{\varphi}_l$ below as the identity map, under which the two window rows coincide. Fix a pair $(i,j)$, a prefix set $S\subseteq E\setminus\{i,j\}$, and a profile $\mathbf{b}$ in the prior support with $\mathrm{supp}(F_j)$ an interval. Call the triple $(i,j,S)$ \emph{aligned} at $\mathbf{b}$ if $\pi$ processes the agents of $S$ before $i$, $i$ before $j$, and $j$ before every other active agent. The \emph{contested window} $W_{ij|S}(\mathbf{b})$ is empty when $(i,j,S)$ is not aligned at $\mathbf{b}$, and is the open interval $(0,\bar\delta_{ij|S})$ otherwise, where
\[
\bar\delta_{ij|S} \;=\;
\begin{cases}
\min\bigl(b_i-b_j,\; b_j-\max_{k\notin S\cup\{i,j\}} b_k\bigr) & \textup{(W)},\\[2pt]
\min\bigl(\bar{\varphi}_j^{-1}(\bar{\varphi}_i(b_i))-b_j,\;
 b_j-\bar{\varphi}_j^{-1}\bigl(\max_{k\notin S\cup\{i,j\}}\bar{\varphi}_k(b_k)\bigr)\bigr) & \textup{(V)},
\end{cases}
\]
in each case intersected with $\bigl(0,\ \sup\mathrm{supp}(F_j)-b_j\bigr)$, where on branch \textup{(V)} the inner maximum ranges over the active agents outside $S\cup\{i,j\}$, and on either branch a maximum over an empty index set is $-\infty$, the corresponding entry being read as $+\infty$; if $\bar{\varphi}_i(b_i)$ lies above the range of $\bar{\varphi}_j$, read $\bar{\varphi}_j^{-1}(\bar{\varphi}_i(b_i))$ as $\sup\mathrm{supp}(F_j)$, under which that entry is dominated by the support entry, and mirroring it, if $\max_{k\notin S\cup\{i,j\}}\bar{\varphi}_k(b_k)$ lies below the range of $\bar{\varphi}_j$, read $\bar{\varphi}_j^{-1}$ of that maximum as $\underline{v}_j$, under which the second entry reads $b_j-\underline{v}_j$. Alignment makes the first two entries strictly positive, so $W_{ij|S}(\mathbf{b})\ne\emptyset$ exactly when $(i,j,S)$ is aligned and $b_j<\sup\mathrm{supp}(F_j)$: the window is non-empty precisely when the priority order contests the pair and an in-support pullback above $b_j$ exists.
\end{definition}

\begin{definition}[Competitive prior profile]
\label{def:competitive-prior}
Let $f$, $F$, $\bar{\varphi}_i$ and $\underline{v}_i$ be as in \cref{def:contested-window}. For a pair $i,j\in E$ and $S\subseteq E\setminus\{i,j\}$ define the \emph{conditional non-modularity gap}
\[
\gamma_{ij|S} \;\triangleq\; f(S\cup\{i\})+f(S\cup\{j\})-f(S)-f(S\cup\{i,j\}),
\]
which is non-negative by submodularity and recovers the pairwise gap $\gamma_{ij}=f(\{i\})+f(\{j\})-f(\{i,j\})$ at $S=\emptyset$. The pair $(f,F)$ is \emph{competitive} if, with positive prior probability (the \emph{competitive event}), the realised profile $\mathbf{b}$ admits a pair $(i,j)$ and a prefix set $S\subseteq E\setminus\{i,j\}$ such that
\begin{enumerate}
\item[(A)] $\gamma_{ij|S}>0$;
\item[(B)] $\bar{\varphi}_j(b_j)>\max\{0,\;\bar{\varphi}_i(\underline{v}_i)\}$, with $(b_i,b_j)$ in the strictly-increasing interior of $\bar{\varphi}_i,\bar{\varphi}_j$, so that the contested crossing $\bar{\varphi}_i^{-1}(\bar{\varphi}_j(b_j))$ lies strictly above agent $i$\'s active-bid threshold $\max\{\underline{v}_i,\,\bar{\varphi}_i^{-1}(0)\}$;
\item[(C)] the contested window is non-empty, $W_{ij|S}(\mathbf{b})\ne\emptyset$ (\cref{def:contested-window}, evaluated under the ironed-virtual-value priority order, branch \textup{(V)}), which contests $(i,j)$ at prefix $S$, places the crossing below $b_i$, and supplies an in-support perturbation amplitude.
\end{enumerate}
\end{definition}

\begin{remark}[Checkable sufficient conditions for competitiveness]
\label{rem:competitive-sufficient}
Competitiveness holds under either of two standard prior structures, provided some triple has $\gamma_{ij|S}>0$ (at $S=\emptyset$, some pair with $\gamma_{ij}>0$). \emph{(a) I.i.d.\ priors.} Let the agents share a regular prior $F_0$ with interval support, strictly increasing virtual value $\bar{\varphi}_0$, and upper endpoint $\omega=\sup\mathrm{supp}(F_0)$, and set $r=\max\{\inf\mathrm{supp}(F_0),\,\bar{\varphi}_0^{-1}(0)\}$; then $r<\omega$ because $\bar{\varphi}_0(\omega)=\omega>0$. Fix thresholds $r<q<q'<q''<\omega$. The event $\{b_s>q''\ \text{for all}\ s\in S\}\cap\{q'<b_i<q''\}\cap\{q<b_j<q'\}\cap\{b_k<q\ \text{for all}\ k\notin S\cup\{i,j\}\}$ has positive prior probability, and on it (A)--(C) hold: the common $\bar{\varphi}_0$ makes bid order and priority order coincide, which aligns $(i,j,S)$, and $b_j<q'<\omega$ then gives (C); the contested crossing equals $b_j$ and lies in the interior of the common support above $r$, which is (B), since $b_j>q>r$ gives $\bar{\varphi}_0(b_j)>\bar{\varphi}_0(r)=\max\{0,\bar{\varphi}_0(\underline{v}_j)\}$; and any $\hat{b}_j\in(b_j,\min\{b_i,\omega\})$ exhibits an amplitude in the window. \emph{(b) Common support bottom.} Let every support be $[0,\omega_i]$ with continuous positive density. Then $\bar{\varphi}_i(0)=-1/f_i(0)<0$, so each agent falls below its reserve $r_i=\bar{\varphi}_i^{-1}(0)>0$ with positive probability, while $\bar{\varphi}_i(\omega_i)=\omega_i>0$. Pick $0<t<t'<\min_{l\in S\cup\{i\}}\omega_l$ with $t<\omega_j$. The event $\{\bar{\varphi}_s(b_s)>t'\ \text{for all}\ s\in S\}\cap\{t<\bar{\varphi}_i(b_i)<t'\}\cap\{0<\bar{\varphi}_j(b_j)<t\}\cap\{b_k<r_k\ \text{for all}\ k\notin S\cup\{i,j\}\}$ has positive prior probability by continuity and full support. On it every agent outside $S\cup\{i,j\}$ is inactive, so $(i,j,S)$ is aligned and (A) holds; (B) holds because $\bar{\varphi}_j(b_j)>0>\bar{\varphi}_i(0)=\bar{\varphi}_i(\underline{v}_i)$, which places the contested crossing in $(r_i,b_i)\subset(0,\omega_i)$; and (C) follows from alignment together with $b_j<\omega_j$, by continuity of $\bar{\varphi}_j$ at $b_j$.
\end{remark}

\begin{theorem}[Credibility trilemma for polymatroidal markets]
\label{thm:trilemma}
Let $\mathcal{X}_{\mathrm{res}}$ be a non-modular polymatroidal feasible region with rank function $f$ over ground set $E$ with $|E| \ge 2$, populated by single-parameter agents with quasi-linear valuations and non-degenerate atomless regular priors $F=\prod_i F_i$ such that $(f,F)$ is competitive (\cref{def:competitive-prior}), operated by a single self-interested operator.\footnote{Coalition deviations among multiple operators, or between an operator and a subset of agents, are out of scope here. Strategy-proof mechanism design under rich interdependent-values type spaces faces structural obstacles~\cite{jehiel2001efficient}.} Within the class of revenue-optimal DSIC static sealed-bid mechanisms over $\mathcal{X}_{\mathrm{res}}$ (i.e., monotone allocation rules with Archer--Tardos payments that maximise expected operator revenue under independent regular priors), no mechanism is simultaneously:
\begin{enumerate}
    \item[(i)] \emph{Revenue-optimal}: maximises expected operator revenue when agents' values are independently drawn from regular distributions in the sense of Myerson~\cite{myerson1981optimal};
    \item[(ii)] \emph{DSIC for agents}: truthful reporting is a dominant strategy;
    \item[(iii)] \emph{Credible for the operator in the realisation-wise sense}: no profitable undetectable deviation exists for any bid realisation in the support.
\end{enumerate}
The ``static sealed-bid'' format is defined formally in \cref{def:static-sealed-bid} below: the operator executes a single round and each agent observes only its own bid, allocation, and payment.
\end{theorem}

\begin{remark}[Scope]
\label{rem:scope-consolidated}
Two conditions delimit the results. \emph{(1) Single-parameter agents.} The trilemma and the CoNC lower bound are proved for single-parameter quasi-linear agents; integrator encapsulation produces this structure for the slice marketplace (P3, \cref{prop:encapsulation-p3}). The multi-parameter and Bayesian-IC extensions are open, with the structural obstructions listed in \cref{sec:future}. \emph{(2) Topology-class scaling is out of scope.} The trilemma is a single-capacity-sharing-pair construction and holds on every non-modular polymatroid carrying a competitive prior profile (\cref{def:competitive-prior}); the asymptotic topology-class analysis of the credibility gap, including its tightness conventions and the envelope scope of the perturbation family, is outside this paper. The constructive resolutions (\cref{thm:commitment,prop:domain-separation}) are likewise single-parameter and realisation-wise; the supply-side analysis of \cref{sec:integrator-credibility} adds the integrator as a second strategic party, with part~(i) of \cref{prop:supply-credibility} a positive-probability statement.
\end{remark}

On an additive $f$ every gap vanishes, agents are independent and credibility is trivially satisfied. A positive \emph{pairwise gap} $\gamma_{ij}=f(\{i\})+f(\{j\})-f(\{i,j\})$, the conditional gap at $S=\emptyset$, suffices but is not necessary: $\gamma_{ij|S}$ is not monotone in $S$, and on the uniform matroid of rank~$2$ over three agents every $\gamma_{ij}$ vanishes while $\gamma_{12|\{3\}}=1$, as on unit-task markets whose shared nodes hold more than one unit. In the DAG-induced polymatroid of~\cite{loven2026realtime}, non-modularity holds whenever two agents' service-dependency DAGs share a capacity-constrained internal node. This is the standard condition in the computing continuum, where edge servers, network links, and cloud endpoints are shared across organisational boundaries~\cite{deng2025agenticservicescomputing}.

\begin{definition}[Static sealed-bid mechanism]
\label{def:static-sealed-bid}
A mechanism is \emph{static} if the operator executes a single round (no within-round observation-and-response dynamics between agents and operator) and \emph{sealed-bid} if each agent $i$'s information on the execution path comprises only its own bid $b_i$, its allocation $x_i$, and its payment $p_i$ (no public broadcast, no third-party transcript, no cross-agent observation). The undetectability condition of \cref{def:deviation} is defined with respect to this information model.
\end{definition}

We adopt the support notion of \cref{def:deviation} rather than statistical indistinguishability because a one-shot agent must reject on impossibility, not on improbability.

\begin{remark}[Akbarpour--Li safe-deviation property]
\label{rem:akbarpour-li-alignment}
The undetectability requirement in \cref{def:deviation} is the \emph{safe-deviation property} of Akbarpour and Li~\cite[\S3, Definition~3]{akbarpour2020credible}, and the Payment Perturbation Lemma (\cref{lem:perturbation}) constructs such a safe deviation through a population-preserving rationalising profile that supports the inflated payment at every bid in agent~$i$'s bid neighbourhood. A related paper by two of the present authors treats the operator of this market in an elicitation-theoretic frame, scoring it by a strictly proper compliance rule and formalising undetectability as Bayesian indistinguishability of the agents' signal distributions, a strictly stronger requirement than the support-style safe deviation adopted here, and obtains a two-legged impossibility on the welfare branch~\cite{loven2026endogeneity}. The present paper works in the extensive-form frame under competitive priors and proves the three-legged trilemma with the conditional non-modularity gap; the two formulations share the same economic force, and neither implies the other.
\end{remark}

\paragraph*{DSIC payment characterisation.} In any static sealed-bid DSIC mechanism for single-parameter agents, the allocation rule $x(\cdot)$ must be monotone non-decreasing in each agent's bid, and the payment rule is pinned by the allocation rule~\cite{archer2001truthful}:
\begin{equation}
\label{eq:archer-tardos}
p_i(\mathbf{b}) = b_i\, x_i(b_i, b_{-i}) - \int_0^{b_i} x_i(z, b_{-i})\, dz,
\end{equation}
where $x_i(z, b_{-i})$ denotes agent~$i$'s allocation when it bids $z$ while all other bids $b_{-i}$ remain fixed. The welfare-maximising and (ironed-)virtual-welfare-maximising allocations on the polymatroid are computed by the Edmonds greedy algorithm~\cite{edmonds1970submodular,fujishige2005submodular}: sort agents by a priority order (decreasing bid for welfare-maximising VCG, decreasing ironed virtual value for revenue-optimal mechanisms), and allocate each the maximum feasible amount given the polymatroid constraint and prior allocations.

\begin{lemma}[Payment perturbation, branch-split]
\label{lem:perturbation}
Let $f$ be a non-modular polymatroid, $(i,j)$ a pair and $S\subseteq E\setminus\{i,j\}$ a prefix set with $\gamma_{ij|S}>0$ (\cref{def:competitive-prior}; $S=\emptyset$ recovers $\gamma_{ij}$). Consider a DSIC mechanism whose allocation is the Edmonds greedy in priority order $\pi$ with Archer--Tardos payments~\eqref{eq:archer-tardos}, on one of the two branches \textup{(W)}, \textup{(V)} of \cref{def:contested-window}, and fix a profile $\mathbf{b}$ in the prior support at which $(i,j,S)$ is aligned in the sense of that definition, with $\mathrm{supp}(F_j)$ an interval; on branch \textup{(V)} require additionally that $\bar{\varphi}_j(b_j)>0$ and $\bar{\varphi}_j(b_j)\in\bar{\varphi}_i\bigl([0,b_i]\bigr)$, so that the contested crossing $z_{ij}\triangleq\bar{\varphi}_i^{-1}(\bar{\varphi}_j(b_j))\in(0,b_i)$ is well defined. Then for every $\delta$ in the contested window $W_{ij|S}(\mathbf{b})=(0,\bar\delta_{ij|S})$ the perturbation $\hat{b}_j=b_j+\delta$ (all other bids fixed) leaves agent $i$\'s allocation at $b_i$ unchanged and increases agent $i$\'s Archer--Tardos payment by exactly
\[
\Delta p_i \;=\; \gamma_{ij|S}\cdot
\bigl[\bar{\varphi}_i^{-1}(\bar{\varphi}_j(b_j+\delta))-\bar{\varphi}_i^{-1}(\bar{\varphi}_j(b_j))\bigr]
\;>\;0,
\]
which degenerates to $\delta\cdot\gamma_{ij|S}$ on branch \textup{(W)} and on branch \textup{(V)} with a common prior ($\bar{\varphi}_i\equiv\bar{\varphi}_j$). The window is determined by the virtual-value gap at the contested pair; on the welfare branch the virtual-value gap and the bid gap coincide. The attainable increment range is the image of $(0,\bar\delta_{ij|S})$ under $\delta\mapsto\Delta p_i$; since $\hat{b}_j\in\mathrm{supp}(F_j)$, the deviation is undetectable in the sense of \cref{def:deviation}.
\end{lemma}

\begin{proof}
We construct the perturbation on a generic bid profile with positive other-agents' bids; a two-agent construction requiring $b_k=0$ for all $k\notin\{i,j\}$ would be measure-zero on continuous regular priors for $|E|\ge 3$. On branch (W) the profiles satisfying the priority hypothesis with $S=\emptyset$ form the open set $\{b_i>b_j>\max_{k\notin\{i,j\}}b_k\}\cap\{b_j>0\}$ of positive prior probability; on branch (V) the positive-probability supply of admissible profiles is the competitive event of \cref{def:competitive-prior}.

\emph{Setup.} All maps $\bar{\varphi}_l$ are read as the identity on branch (W). Write $\kappa_l=\bar{\varphi}_l(b_l)$ for the realised priority keys, so the hypothesis gives $\min_{s\in S}\kappa_s>\bar{\varphi}_i(b_i)>\kappa_j>\kappa_k$ for every other active agent $k$. Set $z_{ij}=\bar{\varphi}_i^{-1}(\kappa_j)$ and, for $\delta\in(0,\bar\delta_{ij|S})$, $\hat{z}_{ij}=\bar{\varphi}_i^{-1}(\bar{\varphi}_j(b_j+\delta))$: the first window entry gives $\bar{\varphi}_j(b_j+\delta)<\bar{\varphi}_i(b_i)$ through strictly increasing $\bar{\varphi}_j$, so $\bar{\varphi}_j(b_j+\delta)\in(\kappa_j,\bar{\varphi}_i(b_i))\subseteq\bar{\varphi}_i([0,b_i])$ by continuity of $\bar{\varphi}_i$, and $z_{ij}<\hat{z}_{ij}<b_i$. Whenever agent $i$ is active at own-bid $z$, the greedy allocation to agent $i$ depends on $(z,b_{-i})$ only through the set $B_i(z,b_{-i})=\{k\ne i:\ k\ \text{active},\ \bar{\varphi}_k(b_k)>\bar{\varphi}_i(z)\}$ of agents processed before $i$, and equals $f(B_i\cup\{i\})-f(B_i)$; marginals telescope, so the internal order of $B_i$ is immaterial. At inactive own-bids the allocation is zero under either profile.

\emph{(i) Allocation at $b_i$ is unchanged.} At $z=b_i$ the hypothesis gives $B_i(b_i,b_{-i})=S$. The perturbation moves only agent $j$'s key, upward to $\bar{\varphi}_j(b_j+\delta)<\bar{\varphi}_i(b_i)$, so $B_i(b_i,\hat{b}_{-i})=S$ as well. Hence $x_i(b_i,\hat{b}_{-i})=x_i(b_i,b_{-i})=f(S\cup\{i\})-f(S)$.

\emph{(ii) The counterfactual allocation difference is supported on $(z_{ij},\hat{z}_{ij})$ and equals $\gamma_{ij|S}$ there.} Partition the counterfactual bids $z\in(0,b_i)$ by the value of $\bar{\varphi}_i(z)$ against the two keys of agent $j$. If $\bar{\varphi}_i(z)>\bar{\varphi}_j(b_j+\delta)$, then $j\notin B_i$ under either profile and no other membership changes, so $x_i(z,b_{-i})=x_i(z,\hat{b}_{-i})$. If $\bar{\varphi}_i(z)<\bar{\varphi}_j(b_j)$, then $j\in B_i$ under either profile (the perturbation only raises $j$'s key), so the allocations again coincide; this covers in particular all $z$ below agent $i$'s reserve, where both allocations vanish. If $\bar{\varphi}_j(b_j)<\bar{\varphi}_i(z)<\bar{\varphi}_j(b_j+\delta)$, which by strict monotonicity and continuity of $\bar{\varphi}_i$ happens exactly for $z\in(z_{ij},\hat{z}_{ij})$, then $\bar{\varphi}_i(z)$ exceeds every outside active key (the hypothesis places $j$ before all remaining agents) and stays below $\bar{\varphi}_i(b_i)<\min_{s\in S}\kappa_s$, so $B_i(z,b_{-i})=S$ while $B_i(z,\hat{b}_{-i})=S\cup\{j\}$, giving
\[
x_i(z,b_{-i})-x_i(z,\hat{b}_{-i})
=\bigl[f(S\cup\{i\})-f(S)\bigr]-\bigl[f(S\cup\{i,j\})-f(S\cup\{j\})\bigr]
=\gamma_{ij|S}.
\]
The boundary points $z\in\{z_{ij},\hat{z}_{ij}\}$ form a measure-zero set covered by the tie-breaking convention of \cref{thm:trilemma}, Step~1.

\emph{(iii) Payment increase.} By the Archer--Tardos formula~\eqref{eq:archer-tardos} and part~(i),
\[
\hat{p}_i - p_i \;=\; \int_0^{b_i}\bigl[x_i(z, b_{-i}) - x_i(z, \hat{b}_{-i})\bigr]\,dz
\;=\; \gamma_{ij|S}\cdot(\hat{z}_{ij}-z_{ij})
\;=\; \gamma_{ij|S}\cdot\bigl[\bar{\varphi}_i^{-1}(\bar{\varphi}_j(b_j+\delta))-\bar{\varphi}_i^{-1}(\bar{\varphi}_j(b_j))\bigr],
\]
strictly positive because $\bar{\varphi}_j$ and $\bar{\varphi}_i^{-1}$ are strictly increasing on the window. On branch (W), and on branch (V) under a common prior, $\bar{\varphi}_i^{-1}\circ\bar{\varphi}_j$ is the identity and the increment is $\delta\cdot\gamma_{ij|S}$. Monotone continuity of $\delta\mapsto\Delta p_i$ makes the attainable increment range the image of $(0,\bar\delta_{ij|S})$ as claimed.

\emph{Undetectability.} The support term in $\bar\delta_{ij|S}$ and the interval-support hypothesis give $\hat{b}_j=b_j+\delta\in\mathrm{supp}(F_j)$, so $\hat{b}_{-i}\in\mathrm{supp}(F_{-i})$ rationalises agent $i$'s inflated outcome $(x_i,\,p_i+\Delta p_i)$ with $\hat{b}_i=b_i$; every other agent keeps its honest outcome, rationalised by the truthful profile itself. The deviation is therefore undetectable in the sense of \cref{def:deviation}.
\end{proof}

\paragraph*{Worked example.} Consider a minimal non-modular polymatroid with two agents sharing an edge server: $f(\{1\}) = f(\{2\}) = 2$, $f(\{1,2\}) = 3$, so $\gamma_{12} = 2 + 2 - 3 = 1$. Agent~1 bids $b_1 = 10$, agent~2 bids $b_2 = 5$. Greedy processes agent~1 first: $x_1 = f(\{1\}) = 2$; then agent~2: $x_2 = f(\{1,2\}) - x_1 = 1$. The Archer--Tardos payment for agent~1 is $p_1 = 10 \cdot 2 - \int_0^{10} x_1(z)\,dz$. For $z > 5$, agent~1 is processed first and receives $x_1(z) = 2$; for $z \le 5$, agent~2 is processed first and agent~1 receives $f(\{1,2\}) - f(\{2\}) = 1$. So $p_1 = 20 - (5 \cdot 1 + 5 \cdot 2) = 5$. Now the operator perturbs: raise agent~2's bid to $\hat{b}_2 = 5 + \delta$ with $\delta = 1$. Agent~1's allocation at $b_1 = 10$ is unchanged ($x_1 = 2$), but for counterfactual bids $z \in (5, 6)$, agent~2 now has priority and agent~1's allocation drops from~$2$ to~$1$. The integral decreases by $\delta \cdot \gamma_{12} = 1$, so the payment rises to $\hat{p}_1 = 6$. Agent~1 observes $(x_1, p_1) = (2, 6)$, which is consistent with an honest execution under a profile where a more aggressive competitor bids~$6$. The deviation is undetectable and yields $+1$ extra revenue for the operator.

\begin{proof}[Proof of \cref{thm:trilemma}]
The proof adapts Akbarpour and Li~\cite{akbarpour2020credible} from single-item auctions to polymatroidal feasible regions in three steps.

\textbf{Step 1 (Unique candidate mechanism, up to tie-breaking and flat-ironed-region indeterminacy).} The theorem rules out the simultaneous achievement of (i)--(iii). The candidate mechanism is unique up to the two standard Myerson caveats (measure-zero tie-breaking and flat ironed-virtual-value regions, detailed below); the surviving mechanisms share the same Archer--Tardos payment structure and so stand or fall together against the perturbation deviation. By the Myerson characterisation~\cite{myerson1981optimal} (extended to single-parameter settings by Archer and Tardos~\cite{archer2001truthful}), any revenue-optimal DSIC mechanism on independently drawn regular values must have an allocation rule that maximises the (ironed) virtual welfare $\sum_i \bar{\varphi}_i(b_i)\, x_i$, with a payment rule satisfying~\eqref{eq:archer-tardos} (unique up to an additive constant normalised to zero by individual rationality). On the polymatroid, this maximisation is solved by the Edmonds greedy with agents sorted by decreasing $\bar{\varphi}_i(b_i)$. Two distinct sources of non-uniqueness apply: (a)~\emph{tie-breaking on measure-zero virtual-value-tie sets} (the standard Myerson caveat: at bid profiles where $\bar\varphi_i(b_i)=\bar\varphi_j(b_j)$ exactly, the priority order between $i$ and $j$ in the greedy is unconstrained, but the event has measure zero and the integrated payment is unaffected); and (b)~\emph{flat ironed-virtual-value regions} (the ironing procedure of~\cite{myerson1981optimal} produces a piecewise-monotone $\bar\varphi$ with possibly flat regions where two distinct bid values map to the same ironed virtual value; on such regions the polymatroid LP can have a continuum of revenue-optimal extreme points, all sharing the same expected revenue). The conjunction of~(i) and~(ii) selects an \emph{equivalence class} of mechanisms differing on (a)+(b), all with the same Archer--Tardos payment structure and therefore equally susceptible to the perturbation deviation; if this equivalence class fails~(iii), then no mechanism in the class achieves all three properties. The two-agent perturbation construction of \cref{lem:perturbation} uses bid values $(b_i,b_j)$ in the strictly-monotone interior of $\bar\varphi$, avoiding the flat-region indeterminacy.

\emph{Tie-breaking convention at the boundary of the contested window.} At the measure-zero events $z\in\{z_{ij},\hat z_{ij}\}$, where the priority order between $i$ and $j$ flips, we adopt the lexicographic convention $j\prec i$ when the branch's priority keys $\bar{\varphi}_l(b_l)$ are equal and otherwise priority by that key; on branch \textup{(W)} the key is the bid itself and the two points reduce to $z=b_j$ and $z=b_j+\delta$. The convention does not affect the integral of step~(iii) of the proof of \cref{lem:perturbation}, because the boundary points form a measure-zero subset of the integration domain.

\textbf{Step 2 (Payment-inflation deviation).} We construct a deviation that is profitable and undetectable.

\emph{Construction.} The operator reports the true allocation $x^*(\mathbf{b})$ but charges a chosen agent~$i$ an inflated payment $p'_i = p_i + \varepsilon_i$ for some $\varepsilon_i > 0$, keeping all other agents' outcomes unchanged.

\emph{Revenue gain.} The operator's revenue increases by $\Delta R = \varepsilon_i > 0$.

\emph{Undetectability.} By \cref{lem:perturbation}, there exists a bid profile $\mathbf{b}$ and a perturbation $\hat{\mathbf{b}} = (b_{-j}, \hat{b}_j)$ with $\hat{b}_j = b_j + \delta$ such that $x_i(b_i, \hat{b}_{-i}) = x_i(b_i, b_{-i})$ and $\hat{p}_i = p_i + \varepsilon_i$. Agent~$i$ observes only $(b_i, x_i, p'_i)$ in a sealed-bid mechanism. Since $(x_i, p'_i) = (x_i(b_i, \hat{b}_{-i}), \hat{p}_i)$, the outcome is identical to an honest execution under the legitimate profile $\hat{\mathbf{b}}$. Agent~$i$ cannot distinguish the deviation from a market in which agent~$j$ bid more aggressively; the rationalising profile is population-preserving in the sense of \cref{def:deviation}, carrying no identity beyond the realised agent set. Formally, the undetectability condition holds: $\exists\, \hat{\mathbf{b}}$ with $\hat{b}_i = b_i$ such that $x'_i = x^*_i(\hat{\mathbf{b}})$ and $p'_i = p^*_i(\hat{\mathbf{b}})$. The rationalising profile $\hat{\mathbf{b}}$ explains agent~$i$'s outcome at every $b_i$ outside the contested-crossing interval $(z_{ij},\hat z_{ij}]$ of step~(ii) of the proof of \cref{lem:perturbation}, that is, at every $b_i$ at which the allocation $x_i(\cdot,\hat{\mathbf{b}}_{-i})$ coincides with $x_i(\cdot,\mathbf{b}_{-i})$, so the deviation is safe in the uniform sense, not only at the realised bid.

\textbf{Step 3 (Extension to revenue-optimal mechanisms).} The Payment Perturbation Lemma applies to any greedy-based allocation on the polymatroid, because it uses only three ingredients: (a)~the greedy processes agents in a priority order, (b)~payments follow the Archer--Tardos identity~\eqref{eq:archer-tardos}, and (c)~non-modularity provides a pair and a prefix with $\gamma_{ij|S}>0$. All three hold for the revenue-optimal mechanism: (a)~the virtual-welfare greedy is still an Edmonds greedy on the same polymatroid, with agents sorted by decreasing $\bar{\varphi}_i(b_i)$; (b)~Archer--Tardos applies to any monotone DSIC allocation rule; (c)~non-modularity is a property of $f$, not of the allocation rule.

For the virtual-value greedy the construction must be supported on a positive-probability event on which a capacity-sharing contest is realised inside the mechanism's active bid range; the competitive hypothesis supplies exactly this event. On the competitive event of \cref{def:competitive-prior}, the realised profile $\mathbf{b}$ admits a pair $(i,j)$ and a prefix set $S\subseteq E\setminus\{i,j\}$ satisfying (A)--(C). Clauses (A) and (B), together with the alignment that (C) carries, place $\mathbf{b}$ in the hypotheses of \cref{lem:perturbation}, branch (V), with prefix $S$ and conditional gap $\gamma_{ij|S}>0$, and (C) makes the contested window $W_{ij|S}(\mathbf{b})=(0,\bar\delta_{ij|S})$ non-empty. For every $\delta\in(0,\bar\delta_{ij|S})$ the lemma yields the payment inflation
\[
\Delta p_i \;=\; \gamma_{ij|S}\cdot\bigl[\bar{\varphi}_i^{-1}(\bar{\varphi}_j(b_j+\delta))-\bar{\varphi}_i^{-1}(\bar{\varphi}_j(b_j))\bigr] \;>\; 0
\]
with agent $i$'s allocation at $b_i$ unchanged; since $\hat{b}_j=b_j+\delta\in\mathrm{supp}(F_j)$, the rationalising profile lies in the prior support and the deviation is undetectable exactly as in Step~2. When agents share a common regular prior, $\bar{\varphi}_i\equiv\bar{\varphi}_j$ is strictly increasing on its support, the pullback map $\bar{\varphi}_i^{-1}\circ\bar{\varphi}_j$ is the identity, and the inflation reduces to $\delta\cdot\gamma_{ij|S}$; \cref{rem:competitive-sufficient} shows that i.i.d.\ priors, and more generally common-support-bottom priors, make the competitive event automatic whenever some conditional gap is positive. Two robustness points close the step. First, Myerson reserves and, for supports bounded away from zero, the rejection of out-of-support counterfactual bids that revenue optimality forces modify the allocation curve $x_i(\cdot,b_{-i})$ only at bids below $\max\{\underline{v}_i,\bar{\varphi}_i^{-1}(0)\}$; by (B) the window swept by the lemma lies strictly above this threshold, so the increment computation is unaffected. Second, submodularity gives $f(\{i\})+f(E\setminus\{i\})-f(E)\ge\gamma_{ij|S}$ for every $j\ne i$ and $S\subseteq E\setminus\{i,j\}$, so agent $i$'s counterfactual allocation curve always carries a total drop of at least $\max_{j,S}\gamma_{ij|S}$ across its priority crossings; the competitive hypothesis places at least one such drop, with positive prior probability, at an in-support crossing that the operator can sweep undetectably. Bid pairs on ironed flats are excluded by (B)'s strictly-increasing-interior requirement, which restates the flat-region caveat of Step~1.

Since the deviation is profitable and undetectable, and the unique candidate satisfying~(i) and~(ii) fails~(iii), no static sealed-bid mechanism over a non-modular polymatroid with a competitive prior profile achieves revenue-optimality, DSIC, and credibility simultaneously.
\end{proof}

\begin{remark}[The competitiveness hypothesis marks a boundary]
\label{rem:competitive-boundary}
Competitiveness cannot be dropped from \cref{thm:trilemma}. Take $E=\{1,2\}$ with $f(\{1\})=f(\{2\})=2$ and $f(\{1,2\})=3$ (so $\gamma_{12}=1>0$), $F_1=U[10,11]$, and $F_2=U[1,2]$: the priors are independent, non-degenerate, and regular, so every remaining hypothesis of the theorem holds. The profile fails competitiveness: $\bar{\varphi}_1(b_1)\ge 9>2\ge\bar{\varphi}_2(b_2)$ at every profile, so the contested crossing $\bar{\varphi}_1^{-1}(\bar{\varphi}_2(b_2))\in[5.5,6.5]$ lies below $\underline{v}_1=10$, so clause (B) of \cref{def:competitive-prior} has no witness even though (A) holds and the contested window of \cref{def:contested-window} is non-empty at every profile with $b_2<2$. On this instance the revenue-optimal member of the Step-1 candidate class rejects counterfactual bids below the support bottom (allocation mass below $\underline{v}_i$ leaves rent $u_i(\underline{v}_i)>0$ and forfeits revenue), and it therefore allocates $x=(2,1)$ and charges $p=(20,1)$ at every profile in the support, attaining $\mathbb{E}[\mathrm{rev}]=21$ exactly (for regular $F_i$ the expected virtual value equals the support bottom, giving $2\cdot 10+1\cdot 1=21$). Each agent's honest-outcome support is a singleton, so the honest outcome is the only rationalisable report and the mechanism is realisation-wise credible. The competitiveness hypothesis thus marks a boundary of the impossibility. We make no converse claim; it is left open whether every non-competitive instance admits a realisation-wise credible revenue-optimal mechanism.
\end{remark}

\begin{proposition}[Pairwise compatibility of the trilemma's three legs]
\label{rem:trilemma-pairwise-witnesses}
The three properties of \cref{thm:trilemma} are pairwise compatible, the pair (i)+(iii) under symmetric regular priors on a region symmetric in the agents; explicit witness mechanisms exist for each pair, none of which achieves the third leg.
\begin{itemize}
\item \emph{(i)+(ii) revenue-optimal and DSIC, but not credible.} The Myerson auction~\cite{myerson1981optimal} (and its polymatroidal generalisation via the ironed-virtual-value Edmonds greedy with Archer--Tardos payments~\cite{archer2001truthful}) is the unique mechanism achieving both~(i) and~(ii) on regular priors; \cref{lem:perturbation} shows directly that it fails~(iii).
\item \emph{(ii)+(iii) DSIC and credible, but not revenue-optimal.} The ascending clinching auction of Ausubel~\cite{ausubel2004ascending}, extended to the polymatroidal setting by Goel et al.~\cite{goel2015polyhedral}, is DSIC and (under broadcast/transcript assumptions of \cref{thm:commitment}(i), \cref{def:broadcast-channel}, \cref{def:public-reconstructibility}) credible, but achieves welfare-maximisation rather than revenue-optimality, so it does not satisfy~(i). The format is not static: started at the Myerson reserve on the symmetric class of the third item, it is also revenue-optimal (for one item it is the ascending auction of Akbarpour and Li~\cite{akbarpour2020credible}), so the static sealed-bid restriction of \cref{thm:trilemma} is load-bearing. A static witness for this pair is the posted-price mechanism that serves, by the Edmonds greedy in a fixed priority order, the agents bidding at least a posted price $r$ not below the unit procurement cost, at $r$ per unit: it is DSIC, and it is credible because every payment is $r$ per unit, the prescribed allocation already serves the maximal total $f(\{i:b_i\ge r\})$, and serving a bid below $r$ is detected by its bidder; under the hypotheses of \cref{thm:trilemma} it is not revenue-optimal, because its allocation to an agent is constant above $r$ while that of the candidate of Step~1 drops at the contested crossing.
\item \emph{(i)+(iii) revenue-optimal and credible, but not DSIC.} The pay-as-bid (first-price) sealed auction over the polymatroid with the Myerson reserve $r$ is credible: among the allocations that serve only bids of at least $r$, the prescribed one maximises the operator's own revenue functional $\sum_i b_i x_i$ over $\mathcal{X}_{\mathrm{res}}$, and serving a bid below $r$ is detected by its bidder, so no undetectable deviation of either class of \cref{def:deviation} improves on it. For symmetric regular priors on a region symmetric in the agents ($f(S)$ depending on $|S|$ only) it has a symmetric Bayes--Nash equilibrium that is revenue-optimal, but it is not DSIC. With $X$ the common interim allocation of the value-ordered greedy with reserve~$r$, the bid $\beta(v)=v-\int_r^v X(t)\,dt/X(v)$ for $v\ge r$ is an equilibrium, since a bidder of value $v$ bidding $\beta(w)$ earns $(v-w)X(w)+\int_r^w X(t)\,dt$, whose derivative in $w$ is $(v-w)X'(w)$; $\beta$ is strictly increasing unless $f$ is modular, when every bid of at least $r$ is served, so the allocation is the value-ordered greedy above $r$, which on this class is the Myerson allocation, and revenue equivalence~\cite{myerson1981optimal} gives the optimal revenue (without reserve this is the unique symmetric equilibrium of~\cite{chawla2013unique}). Whether a static, revenue-optimal, credible mechanism exists on the polymatroid for all regular priors is open.
\end{itemize}
The three witnesses confirm that the impossibility in \cref{thm:trilemma} is the conjunction, not any binary subset.
\end{proposition}

\cref{thm:trilemma} implies that VCG, the canonical DSIC mechanism of~\cite{loven2026realtime}, is not credible when the operator has private access to bids.

\paragraph*{Cost of non-credibility.} To make the trilemma quantitatively comparable to the price-of-anarchy bounds familiar from routing games~\cite{koutsoupias2009worstcase,roughgarden2002anarchy}, we introduce the \emph{Cost of Non-Credibility} (CoNC). Let $\mathrm{rev}^*$, $W^* = \sum_i v_i x_i^*$, and $p_i^*$ denote, respectively, the operator revenue, total welfare, and agent~$i$'s payment under the prescribed mechanism on bid profile $\mathbf{b}$, and let the corresponding $\delta$-superscripted quantities denote the same under the operator deviation $\delta$. The expectations are taken over the prior on bid profiles, and a variant written without its argument denotes the supremum over undetectable deviations $\delta$, which the value at any one $\delta$ bounds from below. All variants share the denominator $\mathbb{E}[W^*]$, the expected welfare under the prescribed mechanism. The operator-side and welfare variants are
\begin{align}\label{eq:conc}
\mathrm{CoNC}^{\mathrm{op}}(\delta) &= \frac{\mathbb{E}[\mathrm{rev}^\delta - \mathrm{rev}^*]}{\mathbb{E}[W^*]}, &\qquad \mathrm{CoNC}^{\mathrm{W}}(\delta) &= \frac{\mathbb{E}[W^* - W^\delta]}{\mathbb{E}[W^*]}.
\end{align}

The shared denominator makes the two variants commensurable, and their sum is the agent-side variant $\mathrm{CoNC}^{\mathrm{ag}}(\delta)\triangleq\mathrm{CoNC}^{\mathrm{W}}(\delta)+\mathrm{CoNC}^{\mathrm{op}}(\delta)$, an accounting identity: agents bear the operator's revenue increment, which is their own aggregate payment increment $\mathbb{E}[\sum_i p_i^\delta - \sum_i p_i^*]$, together with the welfare the deviation destroys. Because a deviation can displace an agent as well as inflate payments, $\mathrm{CoNC}^{\mathrm{op}}$ must be kept apart from the operator's \emph{extraction}, the increment in the payments the operator collects from the agents that remain allocated under the deviation. Extraction leaves out the payment forgone by a displaced agent, which the revenue delta in the numerator of $\mathrm{CoNC}^{\mathrm{op}}$ nets out, so the two agree exactly when the deviation displaces nobody and can differ in sign when it does. Both are computed for phantom insertion in \cref{sec:eval-deviations}, under those names. The measure scope differs from the price-of-anarchy ratio it is modelled on~\cite{koutsoupias2009worstcase,roughgarden2002anarchy}: PoA is a worst case over Nash equilibria of a routing game, whereas $\mathrm{CoNC}$ is an expectation over the bid prior under the worst undetectable operator deviation for the prescribed mechanism.

\begin{corollary}[CoNC lower bound under non-modularity]\label{cor:conc-lb}
Let $f$ be a polymatroid with prior profile $F$, fix a branch \textup{(W)} or \textup{(V)} of \cref{def:contested-window}, and let $(i,j,S)$ be a triple with $\gamma_{ij|S}>0$ whose \emph{realisation event} $\mathcal E_{ij|S}$, the set of profiles at which $(i,j,S)$ is aligned with $W_{ij|S}(\mathbf b)\ne\emptyset$ on that branch and, on branch \textup{(V)}, clause~(B) of \cref{def:competitive-prior} holds, has positive prior probability; if $(f,F)$ is competitive, some triple has this property on branch \textup{(V)}. Fix $\rho\in(0,1)$ and let the operator, at every $\mathbf b\in\mathcal E_{ij|S}$, perturb $b_j$ by $\delta(\mathbf b)=\rho\,\bar\delta_{ij|S}(\mathbf b)$ as in \cref{lem:perturbation}, and execute faithfully elsewhere. This deviation is undetectable and achieves
\[
\mathrm{CoNC}^{\mathrm{op}}(\delta) \;=\; \frac{\mathbb{E}[\mathrm{rev}^\delta - \mathrm{rev}^*]}{\mathbb{E}[W^*]} \;=\; \frac{c\,\gamma_{ij|S}}{\mathbb{E}[W^*]},\qquad c\;\triangleq\;\mathbb{E}\bigl[h_{ij|S}(\mathbf b)\,\mathbf 1\{\mathbf b\in\mathcal E_{ij|S}\}\bigr]\;>\;0,
\]
where $h_{ij|S}(\mathbf b)=\bar{\varphi}_i^{-1}(\bar{\varphi}_j(b_j+\delta(\mathbf b)))-\bar{\varphi}_i^{-1}(\bar{\varphi}_j(b_j))$ is the pullback amplitude of \cref{lem:perturbation}, equal to $\delta(\mathbf b)$ on branch \textup{(W)} and on branch \textup{(V)} under a common prior. The constant $c=\Pr[\mathcal E_{ij|S}]\cdot\mathbb E[h_{ij|S}\mid\mathcal E_{ij|S}]$ depends on the prior, the branch, the triple and $\rho$, and not on $f$; it is nondecreasing in $\rho$, with supremum $\mathbb E[h_{ij|S}(\bar\delta_{ij|S});\mathcal E_{ij|S}]$ as $\rho\to1$. It carries the bid units and $\gamma_{ij|S}$ the capacity units, so the ratio is unit-free.
\end{corollary}

\begin{proof}
On $\mathcal E_{ij|S}$ the hypotheses of \cref{lem:perturbation} hold with prefix $S$, and $\delta(\mathbf b)$ lies in the contested window because $\rho<1$; the lemma gives an undetectable deviation that leaves every allocation unchanged and raises agent $i$'s payment, hence the operator's revenue, by exactly $\gamma_{ij|S}\,h_{ij|S}(\mathbf b)>0$. Elsewhere the deviation is the identity. The event is open and $h_{ij|S}$ is measurable, so taking expectations gives the displayed identity, and $c>0$ because $h_{ij|S}$ is positive on an event of positive probability. The event and $h_{ij|S}$ are defined by the bids, the priors and the branch alone, so $c$ does not depend on $f$. No allocation changes, so $\mathrm{CoNC}^{\mathrm{W}}(\delta)=0$ and $\mathrm{CoNC}^{\mathrm{ag}}(\delta)=\mathrm{CoNC}^{\mathrm{op}}(\delta)$. Under competitiveness the competitive event is the finite union, over the triples with $\gamma_{ij|S}>0$, of the events $\mathcal E_{ij|S}$ on branch \textup{(V)}, so one of them has positive probability; selecting one realising triple measurably at each profile of the competitive event and summing gives $\mathbb{E}[\mathrm{rev}^\delta-\mathrm{rev}^*]=\sum_{(i,j,S)}c_{ij|S}\,\gamma_{ij|S}$, with $c_{ij|S}$ the expectation of $h_{ij|S}$ over the profiles at which $(i,j,S)$ is selected.
\end{proof}

\cref{sec:evaluation} computes the CoNC of the perturbation deviation, and of phantom insertion, on the companion's certified instances.

\begin{corollary}[Constructive trilemma]
\label{cor:constructive-trilemma}
Under the hypotheses of \cref{thm:trilemma}, given any prescribed mechanism $(x^{*},p^{*})$ in the polymatroidal class, an undetectable operator deviation $\delta:\mathbf{b}\mapsto(x',p')$ with strictly positive expected revenue gain can be computed in time $O(n\log n + n\,T_f)$, where $T_f$ is the cost of one polymatroid rank-oracle call.
\end{corollary}

\begin{proof}
The deviation is the one constructed in the proof of \cref{lem:perturbation}: sort by the priority order in $O(n\log n)$, find a pair $(i,j)$ and prefix set $S$ with $\gamma_{ij|S}>0$ in at most $n$ rank-oracle calls, and displace the auxiliary bid inside the contested window $W_{ij|S}(\mathbf{b})$ (\cref{def:contested-window}). Profitability is \cref{lem:perturbation}; expected profitability follows by integrating over the positive-measure undetectability event of \cref{thm:trilemma}'s proof. The sort dominates when the rank oracle is constant-time after a one-shot linear-time preprocessing pass, as on rooted trees; on the routed reading of a general capacitated DAG $f$ is a source--sink max-flow, so $T_f=O(VE)$~\cite{orlin2013maxflow}.
\end{proof}

\subsection{Restoring Credibility via Commitment}
\label{sec:commitment}

The trilemma can be resolved by making the auction transcript \emph{publicly reconstructible}: any participant can independently verify all clinch quantities from the broadcast data and authenticated on-path rank values, so that any operator deviation produces a detectable inconsistency. \cref{thm:trilemma} applies to static sealed-bid mechanisms and concerns revenue optimality. The commitment mechanisms below restore credibility by moving to an ascending format (Part~(i), achieving welfare optimality) or a deferred-revelation format (Part~(ii), achieving revenue optimality under matroid feasibility with strongly regular distributions).

\begin{theorem}[Credible mechanisms via commitment]
\label{thm:commitment}
Let $\mathcal{X}_{\mathrm{res}}$ be a polymatroidal feasible region defined on a resource ground set $E_R$ with rank function $f:2^{E_R}\to\mathbb{R}_{\ge 0}$, populated by single-parameter agents with quasi-linear valuations $v_i \cdot x_i - p_i$, with allocations in part~(i) either continuous, $x_i \ge 0$ over divisible goods, or unit-demand, $x_i\in\{0,1\}$ as in \cref{prop:efficient-mechanism}(iii), on a region whose 0/1 points form a matroid (as when $f$ is integer-valued, or when the region is cut out by capacities on a laminar family of agent sets, its 0/1 points then obeying the floored capacities), and binary allocations $x_i\in\{0,1\}$ on the matroid of part~(ii). Assume: (B0)~the broadcast channel satisfies \cref{def:broadcast-channel}, including the participant-authenticity axiom (S3); (B1)~agents broadcast their own demand/exit messages directly to all participants (not relayed by the operator); (B2$^{\prime}$)~\emph{verifiable rank function on the equilibrium path}: for every active-agent subset that arises during execution, the on-path values $f(D(p))$ and $f(D_{-i}(p))$ are broadcast with an authenticated proof (e.g., a Merkle commitment to the capacity topology, or a Trusted Execution Environment (TEE) attestation of each rank-oracle evaluation on the committed topology). If the marketplace operator commits to the allocation rule via a \emph{public broadcast channel} (carrying the operator's price announcements, agent-broadcast demand messages per~(B1), and the authenticated on-path rank values per~(B2$^{\prime}$)), then:
\begin{enumerate}
    \item[(i)] the ascending clinching auction~\cite{ausubel2004ascending,goel2015polyhedral,bikhchandani2011matroid} over $\mathcal{X}_{\mathrm{res}}$ is credible, DSIC, and welfare-maximising;
    \item[(ii)] if additionally $\mathcal{X}_{\mathrm{res}}$ is a matroid (as at the Level-1 slice marketplace under integrator encapsulation), a deferred-revelation auction (DRA) with deposits~\cite{ferreira2020credible,chitra2024credible,ec2025matroid} is credible, DSIC, and revenue-optimal for $\alpha$-strongly regular distributions (i.e., distributions whose virtual value function $\varphi$ satisfies $\varphi(v) - \varphi(\hat{v}) \ge \alpha(v - \hat{v})$ for all $v \ge \hat{v}$; $\alpha = 0$ reduces to Myerson regularity). DRA is \emph{not} credible beyond matroid feasibility~\cite{ec2025matroid}.
\end{enumerate}
\end{theorem}

Part~(i): (B1)+(B2$^\prime$) over a reliable causal-broadcast channel (\cref{def:broadcast-channel} below) make the clinching transcript publicly reconstructible (\cref{def:public-reconstructibility}), so any operator deviation is detectable in a BAR adversary model that treats the operator as Byzantine and agents as Altruistic/Rational. (B2$^\prime$) is weaker than requiring a globally public rank function: only on-path values require authentication, preserving commercial confidentiality of the full feasibility structure. Part~(ii) imports Ganesh and Zhang~\cite{ec2025matroid} with deposit-threshold bound $d_{\mathcal{O}}^{*}=\Theta(\bar v \cdot n/\alpha)$, after the Level-1 reduction to matroid feasibility (\cref{lem:encapsulation-matroid}).

\begin{definition}[Broadcast channel formalism]
\label{def:broadcast-channel}
The broadcast channel is a \emph{reliable causal-broadcast} primitive over a participant set $\mathcal{P}$ comprising agents, the operator, and observers. It exposes a single operation $\mathrm{bcast}(m)$ delivered to every $q \in \mathcal{P}$ as $\mathrm{deliver}(q, m)$. We require:
\begin{enumerate}
\item[\textup{(S1)}] \emph{Integrity (safety).} If $\mathrm{deliver}(q, m)$ occurs at any honest $q$, then some sender previously executed $\mathrm{bcast}(m)$. Messages are not forgeable, dropped, or reordered relative to causal predecessors.
\item[\textup{(S2)}] \emph{Agreement (safety).} If two honest participants $q_1, q_2$ each deliver $m$, both observe the same content; no equivocation.
\item[\textup{(S3)}] \emph{Participant authenticity (closed roster).} The participant set $\mathcal{P}$ is fixed and published as part of the operator's pre-execution commitment; every message carries an identity binding verifiable against the roster (e.g., a signature under a key registered to a $\mathcal{P}$-member), and honest participants deliver only roster-bound messages. Roster admission is operator-independent in the sense of (C2) of \cref{prop:domain-separation}: an identity affiliated with the operator is not admitted. A change of membership requires a fresh commitment. In the terms of \cref{def:deviation}, (S3) is what removes the population-augmenting class from the operator's action space on this channel.
\item[\textup{(L1)}] \emph{Eventual delivery (liveness).} Every $\mathrm{bcast}(m)$ by an honest participant is eventually $\mathrm{deliver}$ed to every honest participant within a bounded delay $\Delta$.
\end{enumerate}
We adopt the BAR (Byzantine--Altruistic--Rational) adversary model of Aiyer et al.~\cite{aiyer2005bar}: the operator $\mathcal{O}$ is treated as \emph{Byzantine} (may deviate arbitrarily, subject to undetectability); agents are a mix of \emph{Altruistic} (follow the protocol) and \emph{Rational} (best-respond to incentives) participants. Credibility (\cref{def:credible}) is robust to a Byzantine operator under (S1)--(S3) and (L1) because public reconstructibility (below) makes any deviation detectable, and rational agents will prefer the truthful path once detection is positive-probability.
\end{definition}

Every implementation of \cref{def:broadcast-channel} must carry the (S3) roster: an open-membership substrate satisfies (S1)--(S2) and (L1) but admits identities that no roster check can resolve, and so does not implement \cref{def:broadcast-channel}. Causal order suffices, and total order is unnecessary, because the auction's price-clock advance imposes its own happens-before relation. Clock-stalling and message-delay tactics within the delivery bound $\Delta$ are constrained by (L1) and by that relation; we treat them informally, without a separate deviation case.

\begin{definition}[Public reconstructibility]
\label{def:public-reconstructibility}
An execution of an ascending auction is \emph{publicly reconstructible} if each quantity needed to verify every agent's clinch (the current price $p$, the active set $D(p)$, and the on-path rank values $f(D(p))$, $f(D_{-i}(p))$) is derivable from messages on the broadcast channel of \cref{def:broadcast-channel} using publicly known protocol rules. (B1)+(B2$^{\prime}$) over a channel satisfying (S1)--(S3) and (L1) ensures public reconstructibility.
\end{definition}

\begin{lemma}[Level-1 encapsulation yields matroid feasibility]
\label{lem:encapsulation-matroid}
Under integrator encapsulation (\cref{prop:encapsulation-p3}, after~\cite{loven2026realtime}, with encapsulation conditions E1--E4), assume in addition: (a) \emph{integrality of all agent-facing capacities}, i.e., every capacity of the quotient graph $G'$ (the slice capacities $c_k=\bar C_k$ and the retained inter-cluster capacities, shared downstream bottlenecks included) is a non-negative integer, so that the agent-facing rank function $f'$ of \cref{prop:encapsulation-p3} is integer-valued; and (b) unit demand over the slice types ($x_i\in\{0,1\}$, supplied by GS1), with $A_i\subseteq\{1,\ldots,K\}$ the set of slice types admissible for agent $i$. Write $[n]=\{1,\ldots,n\}$ for the agent index set, and call $J\subseteq[n]$ \emph{servable} if some assignment of each $i\in J$ to a type in $A_i$ has type-load vector $y$ with $y(S)\le f'(S)$ for all $S\subseteq\{1,\ldots,K\}$, so that the feasible 0/1 allocation vectors are exactly the indicators of servable sets. Then
\[
\mathcal I \;=\; \bigl\{\,J\subseteq[n] : J \text{ is servable}\,\bigr\}
\]
is the independent-set family of a matroid on the agent ground set $[n]$, namely the matroid induced from the integer polymatroid $f'$ through the eligibility bipartite graph that joins each agent $i$ to the types in $A_i$ (Rado's theorem for polymatroids~\cite{mcdiarmid1975rado}, see also \cite[\S44.6g]{schrijver2003combinatorial}, together with the Edmonds--Rota induction theorem~\cite[Prop.~11.1.1, Cor.~11.1.2]{oxley2011matroid}). When the clusters are capacity-disjoint, so that $f'$ is modular, this matroid is the transversal matroid of the eligibility graph with type-$k$ multiplicity $c_k$~\cite{perfect1969independence}. The Level-1 feasible region is therefore a matroid in the sense of Ganesh--Zhang~\cite{ec2025matroid} over the bidders, which is the object their DRA requires.
\end{lemma}

\begin{proof}[Proof of \cref{lem:encapsulation-matroid}]
Deferred to \cref{app:routine-proofs}; the argument is that $J$ is servable if and only if Rado's condition
\begin{equation}
|B|\;\le\;f'(N(B))\qquad\text{for every }B\subseteq J
\label{eq:rado-condition}
\end{equation}
holds, which is Rado's theorem for integer polymatroids~\cite{mcdiarmid1975rado} under the integrality hypothesis~(a), and that $B\mapsto f'(N(B))$ is integer-valued, monotone and submodular, so the Edmonds--Rota induction theorem~\cite[Prop.~11.1.1, Cor.~11.1.2]{oxley2011matroid} makes $\mathcal I$ the independent-set family of a matroid on $[n]$.
\end{proof}

\begin{remark}[Boundary of \cref{lem:encapsulation-matroid}: fractional shared capacity]
\label{rem:fractional-capacity-boundary}
Within polymatroidal Level-1 feasibility, the boundary of the matroid property is fractional shared capacity. Shared integer bottlenecks are covered by the lemma: two unit-demand agents $\{i,j\}$ on two capacity-$1$ slices behind a shared downstream capacity-$1$ bottleneck yield the servable family $\{\emptyset,\{i\},\{j\}\}$, the uniform matroid $U_{1,2}$ (the augmentation axiom quantifies over pairs with $|I_1|<|I_2|$ and holds here). Fractional shared capacity breaks augmentation. Take three slice types with $f'(\{1\})=f'(\{2\})=f'(\{3\})=1$, $f'(\{1,2\})=f'(\{2,3\})=3/2$, $f'(\{1,3\})=f'(\{1,2,3\})=2$, which is monotone and submodular, hence a polymatroid rank function, and three unit-demand agents $a,b,c$ eligible for types $1,2,3$ respectively. The servable family is $\{\emptyset,\{a\},\{b\},\{c\},\{a,c\}\}$, because $\{a,b\}$ and $\{b,c\}$ each place two units on a type pair of joint capacity $3/2$; augmentation fails at $I_1=\{b\}$, $I_2=\{a,c\}$. Integrality of all agent-facing capacities (hypothesis (a) of \cref{lem:encapsulation-matroid}) therefore cannot be dropped in the shared-capacity case, whereas with capacity-disjoint clusters a fractional $c_k$ could simply be floored.
\end{remark}

\begin{proof}[Proof of \cref{thm:commitment}]
\textbf{Part (i): Ascending clinching auction.} We use the ascending clinching auction of Ausubel~\cite{ausubel2004ascending}, extended to polymatroidal environments by Goel et al.~\cite{goel2015polyhedral}, on the polymatroid $\mathcal{X}_{\mathrm{res}}$ over the resource ground set $E_R$ with rank function $f$. A price clock $p$ starts at $0$ and increases continuously. At price $p$, each agent $i$ reports its demand $d_i(p) = \max\{x_i : v_i(x_i) \ge p \cdot x_i\}$. The operator computes the residual supply for each agent:
\[
  s_i(p) = f(D(p)) - f(D_{-i}(p))
\]
where $D(p) = \{j\in\mathcal{A} : d_j(p) > 0\}$ is the set of active agents and $D_{-i}(p) = D(p) \setminus \{i\}$. Agent $i$ ``clinches'' $\min(d_i(p), s_i(p))$ units at price $p$. The auction ends when demand equals supply. Under unit demand, $f$ in $s_i(p)$, in the stopping rule and in (B2$^{\prime}$) is the rank $r(J)=\max\{|J'|: J'\subseteq J\text{ servable}\}$ of the matroid on the agents, so $s_i(p)\in\{0,1\}$, every clinch is integral, and agent $i$ clinches its unit at its VCG payment, the price at which it becomes a coloop of the matroid restricted to $D(p)$: under truthful exits this is the least $p$ at which $i$ lies outside the span of the other agents still active, which is $i$'s VCG threshold. This is the ascending auction of~\cite{bikhchandani2011matroid}, which sells bases; padding the matroid with $r(\mathcal A)$ zero-value seller elements (the truncation to rank $r(\mathcal A)$ of its direct sum with a free matroid) makes every servable set part of a basis and leaves the VCG outcome unchanged.

\emph{Economic properties (DSIC and efficiency).} DSIC and welfare maximisation are properties of the ascending clinching auction itself, established independently of any broadcast channel. DSIC follows from the clinching mechanism: at each price, agent~$i$'s clinch $s_i(p)$ depends only on other agents' demands, and agent~$i$'s only action is its exit price, so, whatever the other agents do, staying active while $p<v_i$ adds clinches priced below $v_i$ and staying active beyond $v_i$ adds clinches priced above it; truthful exit is therefore a dominant strategy, and the outcome is the Vickrey outcome~\cite{ausubel2004ascending}. Welfare maximisation follows from the efficiency of the clinching auction on polymatroidal feasible regions with single-parameter valuations satisfying gross substitutes~\cite{goel2015polyhedral,gul1999grosssubstitutes}.

\emph{Credibility (requires (B1) and (B2$^{\prime}$)).} The broadcast channel carries three types of messages: (1)~the operator's price announcements; (2)~each agent's own demand/exit messages, broadcast directly by agents to all participants (not relayed by the operator); and (3)~the operator's announcements of on-path rank values $f(D(p))$, $f(D_{-i}(p))$ together with an authentication witness satisfying (B2$^{\prime}$) and the implied clinch quantities $s_i(p)$. The active set $D(p)$ is derived from the agent-broadcast messages, so no participant must trust the operator's report of who is active. Under (B2$^{\prime}$), each on-path rank value is accompanied by a verifiable proof; agents and observers check the proof against the published authentication root before accepting the value. The clinch computation $s_i(p) = f(D(p)) - f(D_{-i}(p))$ is therefore publicly reconstructible (\cref{def:public-reconstructibility}) from the transcript.

Given the broadcast of $(p, D(p), f(D(p)), \{f(D_{-i}(p))\})$ with authentication witnesses, each agent can verify its clinch independently. Consider a deviation by the operator: (1)~announcing a different price to different agents is immediately detectable via the broadcast channel; (2)~misreporting $D(p)$ is detectable by any agent whose own participation status is misrepresented; (3)~miscomputing $s_i(p)$ is detectable by agent $i$, which recomputes $f(D(p)) - f(D_{-i}(p))$ from the authenticated broadcast values; (4)~submitting an inauthentic rank value is detectable by any participant because (B2$^{\prime}$) requires a verifiable proof anchored to the operator's prior commitment; (5)~originating a demand or exit message that its purported sender did not broadcast, whether under an identity outside the committed roster or under that of a rostered participant that stayed silent, is detectable by every honest participant. Under (S3) the roster is fixed before bidding begins, and every transcript entry is bound to a roster identity; the roster clause of (S3) closes the first sub-case, since an inserted participant fails the roster check even when the price path is internally consistent, and the identity-binding clause closes the second, since the operator cannot produce a signature under a key it does not hold. An honest verifier accepts a message only if its identity binding resolves to a member of $\mathcal{P}$, and the derived active set $D(p)\subseteq\mathcal{A}$ is itself observer-checkable against the published roster, so a message originated under an unregistered identity is rejected on delivery. (S1) alone does not close this case: integrity binds a delivered message to a sender that broadcast it, while (S3) binds that sender to the roster. The population-augmenting class of \cref{def:deviation} is therefore excluded. Thus the only undetectable action is to follow the protocol.

A roster member that the operator registers and controls is an operator--agent coalition, which lies outside the scope fixed by the coalition footnote of \cref{thm:trilemma}; (S3) closes the insertion of unregistered identities, and coalition-proofness is a separate question.

\textbf{Part (ii): Deferred-revelation auction with deposits.} This builds on the matroid-feasibility result of Ganesh and Zhang~\cite{ec2025matroid}. By \cref{lem:encapsulation-matroid}, the Level-1 slice marketplace under integrator encapsulation with integral agent-facing capacities (hypothesis~(a) of that lemma) is a matroid; the DRA of~\cite{ec2025matroid} therefore applies directly.

The DRA operates in three phases. \emph{Phase~1 (Commitment).} Before bids are submitted, the operator publishes a commitment $\sigma = \mathrm{Hash}(x^*(\cdot), p^*(\cdot))$ to the allocation and payment rules on a secure, append-only broadcast channel. Each agent deposits $d_i$ into an escrow controlled by the smart contract. \emph{Phase~2 (Execution).} Agents submit sealed bids $\mathbf{b}$. The operator computes $x^*(\mathbf{b})$ using the greedy algorithm on $\mathcal{X}_{\mathrm{res}}$ and payments $p^*(\mathbf{b})$ per the committed rule. \emph{Phase~3 (Verification).} The operator reveals the full allocation and payment vector. Any agent $i$ can verify that $(x_i, p_i)$ is consistent with the committed rule $\sigma$ by checking that $x^*(\mathbf{b})$ is the output of the greedy algorithm (deterministic given the bid ordering~\cite{edmonds1970submodular}) and that $p_i$ matches the committed payment formula. If agent $i$ finds a discrepancy, it presents the evidence to the smart contract, which slashes the operator's deposit.

\emph{Credibility (deposit threshold imported from Ganesh--Zhang~\cite{ec2025matroid}).} Any deviation from the committed rule produces verifiable evidence (a mismatch between $\sigma$ and the executed $(x, p)$). The deposit $d_{\mathcal{O}}$ must be set large enough that the expected slashing penalty strictly exceeds the maximum gain from any deviation the operator can profitably execute given the detection process. The explicit bound is supplied by~\cite{ec2025matroid}: $d_{\mathcal{O}}^{*} = \Theta(\bar{v}\,n/\alpha)$ for value support $[0,\bar v]$, $n$ agents, and $\alpha$-strong regularity parameter. Under this choice the DRA is credible. Revenue optimality for $\alpha$-strongly regular distributions at Level~1 follows from the same reserve-price analysis as in~\cite{ec2025matroid}.

\emph{Applicability boundary.} Ganesh and Zhang~\cite{ec2025matroid} prove that DRA is \emph{not} credible for any downward-closed feasibility constraint that violates the matroid augmentation property. This impossibility applies to polymatroids with multi-unit allocations ($x_i > 1$), which arise at Level-2 raw-resource marketplaces. At Level~2, where multi-unit allocations are needed, credibility is ensured by domain separation (\cref{prop:domain-separation}) or by the ascending clinching auction of Part~(i), which extends to polymatroids~\cite{ausubel2004ascending,goel2015polyhedral}.
\end{proof}

\subsection{Restoring Credibility via Domain Separation}
\label{sec:alternatives}

The second credibility resolution, \emph{domain separation}, eliminates the operator's profitable deviation by separating revenue channels. Integrator competition can discipline integrator pricing but is not a third credibility resolution: the operator's deviation of \cref{lem:perturbation} changes payments for given bids and allocations, which no pricing discipline among integrators reaches.

\begin{proposition}[Domain separation]
\label{prop:domain-separation}
Let the marketplace be operated by an entity $\mathcal{O}$ satisfying: (C0)~\emph{settlement separation}: agent payments transit directly from agents to resource owners (e.g., via a verifiable escrow or third-party clearing agent), so $\mathcal{O}$'s books record only the per-unit fee income $\phi \cdot \sum_i x_i$ and not the gross agent-to-owner payments; (C1)~no ownership stake in the allocated resources; (C2)~no affiliation with any participating agent; (C3)~all participating agents have positive valuations and the polymatroid's capacity is binding (aggregate desired allocation strictly exceeds $f(E)$, so $\sum_i x^*_i = f(E)$ at the welfare-maximising allocation); (C4)~\emph{no fee on undelivered allocation}: $\mathcal{O}$ collects $\phi$ only on units that materialise into a delivered service path, so phantom or undeliverable allocations contribute zero to $\mathcal{O}$'s revenue. Under (C0)--(C4), VCG is credible.
\end{proposition}

(C0) is the operative condition against payment inflation, and (C2) and (C4) close phantom insertion: under settlement separation $\mathcal{O}$'s revenue reduces to $\phi$ times the \emph{delivered} allocation, which VCG maximises on the polymatroid. We enumerate the five operator deviation cases and show none is strictly profitable.

\begin{proof}[Proof of \cref{prop:domain-separation}]
Under (C0)--(C1), $\mathcal{O}$'s payoff is exactly $\phi \cdot \sum_i x_i$, proportional to the total allocation and independent of per-agent payment levels.

\emph{VCG fills capacity.} VCG computes the welfare-maximising allocation $x^* \in \mathcal{X}_{\mathrm{res}}$, which maximises $\sum_i v_i x_i$ subject to the polymatroid constraint. Under (C3), $\sum_i x^*_i = f(E)$: any allocation with $\sum_i x_i < f(E)$ could be improved by increasing some agent's allocation (the additional welfare $v_i \cdot \Delta x_i > 0$ is strictly positive given the excess aggregate demand). Therefore VCG achieves the maximum total allocation $f(E)$.

\emph{No deviation is profitable.} Consider an arbitrary operator deviation $\delta: \mathbf{b} \mapsto (x', p')$.
\begin{enumerate}
    \item \emph{Reallocation (same total):} If $\delta$ reallocates among agents while maintaining $\sum_i x'_i = f(E)$, $\mathcal{O}$'s revenue is unchanged at $\phi \cdot f(E)$.
    \item \emph{Total reduction:} If $\delta$ results in $\sum_i x'_i < f(E)$ (e.g., by excluding agents from the allocation or by inflating payments beyond agents' values, violating individual rationality and causing agents to drop out), $\mathcal{O}$'s revenue strictly decreases.
    \item \emph{Capacity violation:} If $\delta$ attempts $\sum_i x'_i > f(E)$, it violates the polymatroid constraint and is infeasible.
    \item \emph{Payment inflation with allocations unchanged (the \cref{lem:perturbation} perturbation, and equally a non-allocated fictitious bid).} If $\delta$ applies the perturbation of \cref{lem:perturbation} to raise some agent's Archer--Tardos payment by $\varepsilon > 0$ without changing allocations, the agent-to-owner transfer rises by $\varepsilon$; but by (C0) this transfer flows directly to resource owners and is not recorded in $\mathcal{O}$'s books. $\mathcal{O}$'s revenue $\phi\cdot\sum_i x_i$ is unchanged.
    \item \emph{Fictitious-agent (ghost-bid) insertion.} This is the population-augmenting class of \cref{def:deviation}, the one class \cref{thm:trilemma}'s construction does not use. A phantom that wins nothing leaves every allocation unchanged and only raises a genuine winner's Archer--Tardos integrand, so it falls under case~4. If $\delta$ inserts a fictitious agent whose bid is allocated, two sub-cases arise.

    \emph{Case~5(a) [displacement of zero-marginal genuine agent].} If the fictitious agent displaces a genuine agent with zero or near-zero marginal allocation (the displaced agent's clinch was already zero, so its absence does not reduce delivered allocation), aggregate \emph{delivered} allocation is preserved at $f(E)$ and $\mathcal{O}$'s fee revenue $\phi\cdot\sum_i x_i$ is unchanged. The fictitious agent's notional payment $p_{\mathrm{fict}}$ inflates the gross-payment account $\sum_i p_i$ by exactly the Archer--Tardos area attributable to the inserted bid; under (C0), this account flows in full to resource owners (more precisely, to whichever \emph{phantom owner} is registered as the supplier of the fictitious slice; the operator's bookkeeping still records only the fee component $\phi\cdot\sum_i x_i$). Because (C0) routes \emph{all} agent-to-owner transit through escrow regardless of agent identity (genuine or fictitious), and (C2) bars the operator from being affiliated with the phantom-owner endpoint, the fictitious payment does not reach $\mathcal{O}$'s books.

    \emph{Case~5(b) [displacement of positive-marginal genuine agent], formalised at Level-2.} If the fictitious agent displaces a genuine agent whose marginal allocation was strictly positive, two consequences follow. First, by (C4) (\emph{no fee on undelivered allocation}), the fictitious agent's allocated units do \emph{not} contribute to $\phi\cdot\sum_i x_i$ in $\mathcal{O}$'s books, because the units are not redeemed against a deliverable service path: $\mathcal{O}$ collects $\phi=0$ on those units. Second, the displaced genuine units would have contributed positive $\phi$ to delivered allocation; their removal therefore reduces $\phi\cdot\sum_i x_i^{\mathrm{deliv}}$ by exactly the displaced quantity, which triggers the total-reduction case~(2) and strictly decreases $\mathcal{O}$'s revenue. Third, the fictitious agent's notional payment again transits through escrow to a phantom owner under (C0), and by (C2) is barred from accruing to $\mathcal{O}$. The combined effect on $\mathcal{O}$'s books is therefore strictly negative under (C4) and zero under (C0)+(C2) for the payment side, so case~5(b) is dominated by the no-deviation baseline.

    The Level-2 formalisation matters because the multi-unit polymatroid setting (Level~2) is precisely where ghost-bid insertion can change the active set of agents: at Level~1 unit demand (hypothesis GS1) allocates at most one unit per agent and the ghost-bid attack reduces to a payment inflation at the displaced agent's rank position (case~4), but at Level~2 a ghost-bid can shift integer flow to a phantom successor. The case~5(b) argument shows that (C0) and (C4) jointly close this Level-2 attack vector: (C0) routes the gross-payment increment to a phantom owner outside $\mathcal{O}$'s books, and (C4) zeroes the fee on the undelivered ghost units.
\end{enumerate}
Since no undetectable deviation yields strictly higher revenue for $\mathcal{O}$ under (C0)--(C4), VCG is credible under domain separation. For matroid/binary-allocation settings ($x_i \in \{0,1\}$), the per-unit fee reduces to a per-transaction fee.
\end{proof}

\begin{corollary}[Knife-edge breakdown under partial stake]
\label{cor:knife-edge}
If (C1) is relaxed to allow an ownership-stake fraction $\lambda \in (0,1]$ of gross payments to accrue to $\mathcal{O}$ (equivalently, $\mathcal{O}$'s revenue becomes $\phi \cdot \sum_i x_i + \lambda \cdot \sum_i p_i$), then the operator can profit from the \cref{lem:perturbation} perturbation by exactly $\lambda \varepsilon$ for any $\varepsilon > 0$ in the achievable range. VCG is therefore not credible for any $\lambda > 0$, recovering the trilemma setting of \cref{thm:trilemma}. \cref{prop:domain-separation} is thus a knife-edge result: credibility holds at $\lambda = 0$ and fails for every $\lambda > 0$. The reading is structural: at $\lambda=0$ the operator's incentive is exactly aligned with the prescribed allocation and at any $\lambda>0$ it is exactly misaligned, so commission-based platform revenue places every commercial marketplace on the misaligned side of a discontinuity.
\end{corollary}

\begin{proof}
Under the relaxed setting, any payment increment $\Delta p_i$ yields $\mathcal{O}$ additional revenue $\lambda \Delta p_i$. Applying \cref{lem:perturbation} to construct $\Delta p_i = \varepsilon$ without changing allocations gives the operator a detectability-free revenue gain of $\lambda\varepsilon > 0$, which is strictly positive for any $\lambda > 0$.
\end{proof}

Side payments or dynamic incentives, like any positive ownership stake (\cref{cor:knife-edge}), re-introduce the credibility problem.

\cref{tab:eval-separation} executes each case of the proof of \cref{prop:domain-separation} with (C0)--(C4) in force and with each condition removed in turn.

\subsection{Integrator Credibility}
\label{sec:integrator-credibility}

The results so far take the agent-facing region as given. An integrator fixes that region before bidding (\cref{def:role-taxonomy}(2)), and when it is not the operator it is a seller: its report is an input to the mechanism, not a deviation on the bid profile, and the property it needs is supply-side incentive compatibility. We treat the static single-integrator case on the Level-1 slice marketplace.

\begin{definition}[Exposure]
\label{def:exposure}
Let integrator $k$'s \emph{deliverable set} $\mathcal F_k\subseteq\mathbb{R}^{L_k}_{\ge 0}$ be the down-closed set of demand vectors on its slices that its sub-DAG can carry. An \emph{exposure} is a polymatroid $\mathcal X_k\subseteq\mathcal F_k$ offered to the slice marketplace before bidding; an offer not contained in $\mathcal F_k$ is detected at delivery and forfeits compensation. The exposure is \emph{truthful} if $\mathcal X_k=\mathcal F_k$, which requires $\mathcal F_k$ to be a polymatroid, and an \emph{understatement} if $\mathcal F_k$ is a polymatroid and $\mathcal X_k\subsetneq\mathcal F_k$. On a scalar interface, $\mathcal F_k=\{x\ge 0: x(L_k)\le c_k\}$ and the exposure is a report $\hat c_k\le c_k$.
\end{definition}

On a scalar interface the capacity $c_k$ is the slice capacity $\bar C_k$ of \cref{prop:encapsulation-p3}. The companion's Proposition~5 sanctions inner exposures of this kind and treats their choice as non-strategic~\cite{loven2026realtime}; here the choice is the integrator's.

\begin{lemma}[Gap decomposition of VCG payments]
\label{lem:gap-payment}
Let $f$ be a polymatroid on the agent set $E$, let the allocation be the Edmonds greedy in decreasing bids (branch \textup{(W)} of \cref{def:contested-window}) with Archer--Tardos payments~\eqref{eq:archer-tardos} and no reserve, and let the bids be distinct. For every agent $i$,
\[
p_i(\mathbf b)\;=\;\sum_{j:\ b_j<b_i}\gamma_{ij|B_j}\,b_j,\qquad B_j=\{k\notin\{i,j\}:\ b_k>b_j\}.
\]
\end{lemma}

\begin{proof}
As agent $i$'s counterfactual bid $z$ falls from $b_i$ to $0$, the set processed before $i$ gains agent $j$ exactly at $z=b_j$, and $i$'s greedy marginal falls there from $f(B_j\cup\{i\})-f(B_j)$ to $f(B_j\cup\{i,j\})-f(B_j\cup\{j\})$, a drop of $\gamma_{ij|B_j}$. Hence $x_i(b_i,\mathbf b_{-i})-x_i(z,\mathbf b_{-i})=\sum_{j:\,z<b_j<b_i}\gamma_{ij|B_j}$, and~\eqref{eq:archer-tardos} integrates this step function to the stated sum.
\end{proof}

VCG revenue is therefore carried entirely by the conditional gaps at realised crossings, and it vanishes on a profile exactly when every such gap does. With two agents the identity reads $p_1+p_2=\gamma_{12}\min(b_1,b_2)$, the revenue formula of the two-agent marketplace of~\cite{loven2026endogeneity}. Raising $b_j$ by $\delta$ inside the contested window leaves every $B_j$ unchanged and raises $p_i$ by $\gamma_{ij|B_j}\,\delta$, which is branch \textup{(W)} of \cref{lem:perturbation}.

\begin{proposition}[Supply-side credibility of the integrator]
\label{prop:supply-credibility}
Consider the Level-1 slice marketplace on a laminar exposed region, with unit-demand agents each eligible for one slice and independent values with positive density on $[0,\omega]$, run by VCG (branch \textup{(W)}) with integrator $k$'s per-unit cost $\theta_k<\omega$ entering welfare. Integrator $k$ is not the operator, has a scalar deliverable set, every ancestor block of its slice has capacity at least $c_k$, and it receives a compensation $T_k(\hat c_k,\mathbf b)$ out of the payments $P_k(\hat c_k,\mathbf b)$ collected on its slice and earns $T_k-\theta_k y_k$, with $y_k$ the units delivered on slice $k$.
\begin{enumerate}
\item[(i)] \emph{Payment share.} If $c_k\ge 2$, at least $c_k$ agents are eligible for slice $k$, and $T_k=\lambda_k P_k$ with $\lambda_k\in(0,1]$, then with positive prior probability the report $\hat c_k=c_k-1$ earns integrator $k$ strictly more than $\hat c_k=c_k$ and lowers welfare.
\item[(i$'$)] \emph{Dual role.} If instead integrator~$k$ is the operator, with fee $\phi\ge 0$ on delivered units and stake $\lambda\in[0,1]$ in the gross payments (\cref{cor:knife-edge}), it earns $\phi\sum_i x_i+\lambda\sum_i p_i-\theta_k y_k$, and the conclusion of~(i) holds whenever $c_k\ge 2$, at least $c_k$ agents are eligible for slice~$k$ and $\phi<\lambda(c_k-1)(\omega-\theta_k)+(1-\lambda)\theta_k$. At $\lambda=0$ and $\phi\ge\theta_k$ the truthful report is weakly dominant if $\theta_k=0$ or no other slice shares an ancestor block with slice~$k$; otherwise withholding can pass an ancestor unit to another integrator's slice, leaving $\sum_i x_i$ unchanged and saving $\theta_k$.
\item[(ii)] \emph{Separation.} Under (C0)--(C4) the operator's payoff does not depend on $P_k$, but if (C0) routes the gross slice payments to integrator $k$, part~(i) applies with $\lambda_k=1$: settlement separation removes the operator's revenue motive, not the integrator's margin from restricting supply.
\item[(iii)] \emph{Volume fee.} If $c_k\ge 2$, at least $c_k$ agents are eligible for slice $k$, $\theta_k$ is known and $T_k=\phi_k y_k$ with $\phi_k\ge\theta_k$, the truthful report is weakly dominant. For every $\phi_k>\theta_k$, $T_k>P_k$ with positive prior probability, so the fee is not funded by the agents' payments.
\item[(iv)] \emph{Envelope rule.} If $\theta_k$ is known and $c^{\max}_k\ge c_k$ is a public bound, the compensation $T^{*}_k(\hat c_k,\mathbf b)=\theta_k y_k(\hat c_k,\mathbf b)+\min_{\hat c_k\le c'\le c^{\max}_k}\bigl[P_k(c',\mathbf b)-\theta_k y_k(c',\mathbf b)\bigr]$ satisfies $0\le T^{*}_k-\theta_k y_k\le P_k-\theta_k y_k$ at every profile, and the truthful report is weakly dominant; it is the largest compensation with these two properties, and it is funded where paying the integrator its VCG marginal contribution~\cite{dash2007market} in general is not.
\item[(v)] \emph{Private cost.} If $\theta_k$ is private, with continuous positive density on an interval whose interior meets the interior of an agent's value support, no mechanism on this class is efficient, dominant-strategy incentive compatible for agents and integrator, individually rational and free of outside subsidy~\cite[Thm.~1, Cor.~1]{myerson1983efficient}; the volume fee of part~(iii) is then a posted supply price~\cite{hagerty1987robust}, which withholds every unit with $\theta_k>\phi_k$.
\item[(vi)] \emph{Commitment.} Under (B0)--(B2$^{\prime}$) on the exposed region, part~(i) holds unchanged for every committed format that charges the VCG payments, the clinching auction of \cref{thm:commitment}(i) among them when, with the costs $\theta_k$ published on the channel, agent~$i$ faces the price $p+\theta_{k(i)}$ on the common clock $p$, exits at $p=b_i-\theta_{k(i)}$ and pays $p+\theta_{k(i)}$ at its clinch (\cref{prop:efficient-mechanism}(iii)): the authenticated rank values are those of the exposure, and the transcript is an honest execution on it.
\end{enumerate}
\end{proposition}

\begin{proof}[Proof of \cref{prop:supply-credibility}]
Deferred to \cref{app:routine-proofs}.
\end{proof}

\begin{remark}[Exposure design sets the gaps]
\label{rem:exposure-design}
When $\mathcal F_k$ is not a polymatroid there is no truthful exposure and no largest one: for $\mathcal F=\{x\ge 0:\ x_1+x_2\le 1,\ x_2+x_3\le 1\}$, a crossing of two unit-capacity blocks, the box with bounds $(1,0,1)$ and the uniform matroid of rank~$1$ on three slices both lie in $\mathcal F$, neither contains the other, and their gaps are $\gamma_{13}=0$ and $\gamma_{13}=1$. By \cref{lem:gap-payment} an exposure is then a choice of the gaps on which both the VCG payments and the perturbation of \cref{lem:perturbation} run: gaps between agents on one slice are set by demand, and gaps between slices by the exposure. The exposure a payment-motivated integrator prefers depends on how demand falls across slices; we leave that design problem open.
\end{remark}
\section{Evaluation}
\label{sec:evaluation}

We compute the objects of \cref{sec:credible-mechanisms} on the certified instances of the companion~\cite{loven2026realtime}: the perturbation of \cref{lem:perturbation} under the mechanisms the trilemma quantifies over, phantom insertion as the wider deviation class, and the closure conditions of \cref{thm:commitment}(i) and \cref{prop:domain-separation}. Every allocation and payment below is produced by running the prescribed mechanism on the realised or counterfactual profile, and every undetectability verdict by re-running it on the rationalising profile. The computations check the theory's identities on concrete instances and measure what the theory leaves open, the frequency of the competitive event and the size of the constant of \cref{cor:conc-lb}; generality rests on the proofs.

\subsection{Setup}
\label{sec:eval-setup}

\paragraph*{Instances.} We use the companion's scale instances T, a rooted tree, and S, a fan in which every edge node lies above every leaf; both carry the same eight nodes (a device node, three edge nodes and four leaves) and the same capacity vector~\cite{loven2026realtime}. The market clears over the companion's fork-join region, which the companion's certificate confirms to be a polymatroid on T and on S at every round used. On T the routed and fork-join readings of \cref{sec:bg-dags} coincide; on S they differ, and the instance table that the study repository's \texttt{make tables} writes lists both readings' ranks and gaps at capacity scale $s=1$. The crossing instance X, on which the fork-join rank is not submodular, lies outside the hypotheses of \cref{thm:trilemma} and is reported as a boundary case only.

\paragraph*{Agents, priors and operating points.} Each agent holds one task per round, so agents are single-parameter with $x_i\in\{0,1\}$ and the region over agents is the matroid the leaf-block polymatroid induces. Bids are the companion's latency-discounted task values taken at a fixed exogenous utilisation, linear in the drawn value, so the leave-one-out welfare in the Clarke pivot is that of a static allocation and carries no congestion coupling. Values are drawn from (a) $U(0,1)$ for every agent, with virtual value $2v-1$ and reserve $1/2$, and (b) a heterogeneous profile, $U(0,1)$ on leaves $\ell_1$ and $\ell_3$ and $F(v)=v^{2}$ on $[0,1]$ on the others, whose virtual value $(3v^{2}-1)/(2v)$ is strictly increasing, so branch \textup{(V)} runs with a non-identity pullback. The primary operating point, the \emph{load point}, has $N=135$ agents on T and $N=68$ on S at capacity scale $s=1$: the companion's offered load, its expected number of tasks per round, carried by one task per agent. The \emph{contention point} keeps these $N$ at $s=0.5$, which halves every capacity and every gap; it is the capacity ablation at which the revenue-optimal branch is measured, since at $s=1$ the Myerson reserve leaves the region slack under the revenue-optimal allocation. At the companion's own agent count ($N=90$ on T and $45$ on S, $s=1$) the region never fills and the deviation rarely realises: the realisation event of \cref{cor:conc-lb} occurs in at most $0.07$ of rounds on T and in none on S, and the competitive event and (C3) in none. Each point runs $30$ rounds for each of $10$ seeds; intervals are Student-$t$ 95\% intervals over seeds.

\paragraph*{Mechanisms.} VCG is the companion's Clarke-pivot implementation on the certified region. The revenue-optimal mechanism is the virtual-value Edmonds greedy with Archer--Tardos payments; both priors are regular, so no ironing is needed. The ascending clinching auction uses the residual supply $s_i(p)$ of the proof of \cref{thm:commitment}. Pay-as-bid (first-price) is VCG's allocation with pay-as-bid payments and no reserve; the credibility argument of \cref{rem:trilemma-pairwise-witnesses} holds at every reserve, $r=0$ included.

\paragraph*{Deviations.} Arm~P is the perturbation of \cref{lem:perturbation}: for each agent $i$, $S$ is the set processed before $i$ and $j$ the next agent in priority, and $b_j$ is raised by $\rho\,\bar\delta_{ij|S}$ at $\rho=1/2$ and at $\rho=1-10^{-9}$, the value at which the $\rho\to1$ limit of $c$ is read. Arm~G inserts one phantom task under an identity outside the round's agent set, at a leaf, with an in-support bid from three candidates: the support top, just above the key of the block's marginal winner (displacing it), and just below that key (winning nothing). Its units are withheld from delivery and revenue counts genuine payments only, so the best of the three is a lower bound on the best phantom. Undetectability is decided per agent by re-running the prescribed mechanism on the rationalising profile, comparing $(x_i,p_i)$, and testing support membership (\cref{def:deviation}). Under clinching, Arm~P is a demand message forged under the next agent's identity and Arm~G an off-roster demand message; their acceptance test is the transcript verification of \cref{sec:eval-resolutions}. Implementation detail is in the study archive named in the Data and Code Availability Statement.

\subsection{The Trilemma's Deviation and Phantom Insertion}
\label{sec:eval-deviations}

At every realisation the re-run leaves agent $i$'s allocation unchanged and raises its payment by exactly the increment of \cref{lem:perturbation}, on branch \textup{(W)} under VCG and, at the contention point, on branch \textup{(V)} under the revenue-optimal mechanism with both priors; the largest absolute difference between computed and closed-form increment is $0$ at the load and contention points. Beyond the first window entry the allocation changes in every control run, as \cref{def:contested-window} predicts. At every certified round the Clarke payments equal the Archer--Tardos critical bids, and the clinching auction returns VCG's allocation and payments, as \cref{prop:efficient-mechanism}(iii) states. Arm~P reallocates nothing, so by \cref{cor:conc-lb} its $\mathrm{CoNC}^{\mathrm{W}}$ is $0$ and $\mathrm{CoNC}^{\mathrm{ag}}=\mathrm{CoNC}^{\mathrm{op}}$.

On these unit-task instances the leaves hold many units, so every pairwise task gap is zero (share $1$ of task pairs on T and on S) and every realising gap is conditional: $\gamma_{ij|S}=1$ for the pair contested at a saturating node after the prefix $S$. At the load point the realisation event of \cref{cor:conc-lb} on branch \textup{(W)} occurs in $0.847$ to $1$ of rounds under VCG. The competitive event, defined on branch \textup{(V)}, occurs there in $0.0167$ of rounds on T under prior (b) and in no other cell; at the contention point it occurs in $0.863$ to $1$ of rounds on both instances and priors. Since every realised gap equals $1$, the constant of \cref{cor:conc-lb} in its selection form is the mean one-agent operator gain per round, which is Step~2 of the proof of \cref{thm:trilemma}, and it factors as $c=\Pr[\mathcal E]\cdot\mathbb E[h\mid\mathcal E]$ (for instance $0.847\times0.000312$ on T under prior (a)). Under VCG at the load point $c$ is $2.64\times10^{-4}$ to $6.51\times10^{-4}$ bid units at $\rho=1/2$ and $5.28\times10^{-4}$ to $1.3\times10^{-3}$ as $\rho\to1$; under the revenue-optimal mechanism at the contention point it is $2.28\times10^{-4}$ to $5.85\times10^{-4}$ and $4.57\times10^{-4}$ to $1.17\times10^{-3}$. Perturbing every agent against its own rationalising profile gives a $\mathrm{CoNC}^{\mathrm{op}}$ of $0.00277\%$ to $0.0113\%$ of expected welfare under VCG at the load point and $0.00362\%$ to $0.018\%$ under the revenue-optimal mechanism at the contention point (\cref{tab:eval-deviations}). The deviation is small because it moves one payment per realising pair by the fraction $\rho$ of the contested window, and the window is bounded by the bid gap $b_i-b_j$ between agents adjacent in a priority order over $N$ agents. The companion's agent-side misreport set, run against VCG and the revenue-optimal mechanism in the first round of each block of $30$ rounds, finds no profitable report (largest gain $1.8\times10^{-15}$, floating-point residue), so the DSIC leg the trilemma holds fixed is intact on these instances.

Phantom insertion extracts two orders of magnitude more. Under VCG the best of Arm~G's three candidates displaces no agent in any round of any point (share $0$), and $\mathrm{CoNC}^{\mathrm{W}}$ is $0$ in every Arm~G row: the best phantom bids just below the key of the block's marginal winner, wins nothing, and raises the threshold, hence the payment, of every winner in the block, where Arm~P raises one payment. Under VCG at the load point its $\mathrm{CoNC}^{\mathrm{op}}$ is $1.44\%$ to $2.25\%$ of expected welfare, and extraction equals the gain in every row because no row displaces an agent. Under the revenue-optimal mechanism it is at most $0.0211\%$ at the load point, where the reserve leaves the region slack and a non-winning phantom raises no threshold, and $0.914\%$ to $1.7\%$ at the contention point. This phantom is the non-winning bid of case~4 of the proof of \cref{prop:domain-separation}. Every Arm~G outcome under VCG and the revenue-optimal mechanism passes the per-agent check of \cref{def:deviation} (share $1$). Under pay-as-bid neither arm raises the operator's revenue at any round of any point (largest gain $0$), as \cref{rem:trilemma-pairwise-witnesses} states for pay-as-bid. The clinching rows report deviations that the transcript verifiers reject in every realised round (\emph{Passes} $0$); there the Arm~P gain is the operator's revenue change across all agents on the forged transcript, which exceeds one agent's $\Delta p_i$ because on a public transcript the forged demand moves every payment it blocks, and its $\mathrm{CoNC}^{\mathrm{op}}$ is the revenue the deviation would carry if it went undetected. On X ($N=12$, $s=0.1$), the companion's certificate holds in $0$ of $300$ rounds and value-greedy falls short of the exact optimum in $0.0333$ of rounds.

\begin{table}[!t]
\centering
\caption{Operator deviations by instance, prior, mechanism and arm: gain and extraction per round (bid units, $10^{-3}$), $\mathrm{CoNC}^{\mathrm{op}}$ and $\mathrm{CoNC}^{\mathrm{W}}$ (\% of expected truthful welfare), and \emph{Passes}, the share of outcomes passing the per-agent check of \cref{def:deviation} (for clinching, the share of realised rounds in which no transcript verifier rejects), with 95\% intervals over seeds. Arm~P is the perturbation of \cref{lem:perturbation}; Arm~G inserts the best of three candidate phantoms. A dash marks a cell with no realisation.}
\label{tab:eval-deviations}
\scriptsize
\setlength{\tabcolsep}{2pt}
\resizebox{\linewidth}{!}{\begin{tabular}{llllrrrrr}
\toprule
Inst. & Prior & Mechanism & Arm & Gain ($10^{-3}$) & Extraction ($10^{-3}$) & $\mathrm{CoNC}^{\mathrm{op}}$ (\%) & $\mathrm{CoNC}^{\mathrm{W}}$ (\%) & Passes \\
\midrule
\multicolumn{9}{l}{\emph{Load point, capacity scale $s=1$}} \\
T & (a) & VCG & P & 0.279 [0.222, 0.336] & 0.279 [0.222, 0.336] & 0.00277 [0.0022, 0.00334] & 0 &    1 \\
T & (a) & Revenue-optimal & P & 0 & 0 & 0 & 0 & --- \\
T & (a) & Clinching & P & 21.1 [16.8, 25.4] & 21.1 [16.8, 25.4] & 0.209 [0.167, 0.252] & 0 & 0 \\
T & (a) & Pay-as-bid & P & 0 & 0 & 0 & 0 &    1 \\
T & (a) & VCG & G &  145 [ 129,  161] &  145 [ 129,  161] & 1.44 [1.28,  1.6] & 0 &    1 \\
T & (a) & Revenue-optimal & G & 0 & 0 & 0 & 0 &    1 \\
T & (a) & Clinching & G &  145 [ 129,  161] &  145 [ 129,  161] & 1.44 [1.28,  1.6] & 0 & 0 \\
T & (a) & Pay-as-bid & G & 0 & 0 & 0 & 0 &    1 \\
T & (b) & VCG & P & 0.338 [0.291, 0.385] & 0.338 [0.291, 0.385] & 0.00295 [0.00254, 0.00337] & 0 &    1 \\
T & (b) & Revenue-optimal & P & 0.00384 [-0.00264, 0.0103] & 0.00384 [-0.00264, 0.0103] & 3.79e-05 [-2.64e-05, 0.000102] & 0 &    1 \\
T & (b) & Clinching & P & 24.6 [20.9, 28.2] & 24.6 [20.9, 28.2] & 0.215 [0.182, 0.247] & 0 & 0 \\
T & (b) & Pay-as-bid & P & 0 & 0 & 0 & 0 &    1 \\
T & (b) & VCG & G &  183 [ 165,  200] &  183 [ 165,  200] & 1.59 [1.44, 1.74] & 0 &    1 \\
T & (b) & Revenue-optimal & G & 2.11 [0.0606, 4.16] & 2.11 [0.0606, 4.16] & 0.0211 [0.000535, 0.0416] & 0 &    1 \\
T & (b) & Clinching & G &  183 [ 165,  200] &  183 [ 165,  200] & 1.59 [1.44, 1.74] & 0 & 0 \\
T & (b) & Pay-as-bid & G & 0 & 0 & 0 & 0 &    1 \\
S & (a) & VCG & P & 0.555 [0.469, 0.641] & 0.555 [0.469, 0.641] & 0.011 [0.00928, 0.0128] & 0 &    1 \\
S & (a) & Revenue-optimal & P & 0 & 0 & 0 & 0 & --- \\
S & (a) & Clinching & P & 27.8 [23.4, 32.1] & 27.8 [23.4, 32.1] & 0.551 [0.464, 0.639] & 0 & 0 \\
S & (a) & Pay-as-bid & P & 0 & 0 & 0 & 0 &    1 \\
S & (a) & VCG & G &  112 [93.5,  130] &  112 [93.5,  130] & 2.22 [1.85, 2.59] & 0 &    1 \\
S & (a) & Revenue-optimal & G & 0 & 0 & 0 & 0 &    1 \\
S & (a) & Clinching & G &  112 [93.5,  130] &  112 [93.5,  130] & 2.22 [1.85, 2.59] & 0 & 0 \\
S & (a) & Pay-as-bid & G & 0 & 0 & 0 & 0 &    1 \\
S & (b) & VCG & P & 0.651 [0.572, 0.729] & 0.651 [0.572, 0.729] & 0.0113 [0.01, 0.0127] & 0 &    1 \\
S & (b) & Revenue-optimal & P & 0 & 0 & 0 & 0 & --- \\
S & (b) & Clinching & P & 32.5 [28.6, 36.5] & 32.5 [28.6, 36.5] & 0.567 [0.501, 0.634] & 0 & 0 \\
S & (b) & Pay-as-bid & P & 0 & 0 & 0 & 0 &    1 \\
S & (b) & VCG & G &  129 [ 114,  144] &  129 [ 114,  144] & 2.25 [1.98, 2.52] & 0 &    1 \\
S & (b) & Revenue-optimal & G & 0.468 [-0.59, 1.53] & 0.468 [-0.59, 1.53] & 0.00938 [-0.0118, 0.0306] & 0 &    1 \\
S & (b) & Clinching & G &  129 [ 114,  144] &  129 [ 114,  144] & 2.25 [1.98, 2.52] & 0 & 0 \\
S & (b) & Pay-as-bid & G & 0 & 0 & 0 & 0 &    1 \\
\midrule
\multicolumn{9}{l}{\emph{Contention point, capacity scale $s=0.5$, revenue-optimal mechanism}} \\
T & (a) & Revenue-optimal & P & 0.331 [0.284, 0.377] & 0.331 [0.284, 0.377] & 0.00508 [0.00436, 0.0058] & 0 &    1 \\
T & (a) & Revenue-optimal & G & 75.7 [65.6, 85.8] & 75.7 [65.6, 85.8] & 1.16 [1.01, 1.32] & 0 &    1 \\
T & (b) & Revenue-optimal & P & 0.252 [0.204, 0.299] & 0.252 [0.204, 0.299] & 0.00362 [0.00293, 0.0043] & 0 &    1 \\
T & (b) & Revenue-optimal & G & 63.6 [59.1, 68.2] & 63.6 [59.1, 68.2] & 0.914 [0.85, 0.979] & 0 &    1 \\
S & (a) & Revenue-optimal & P & 0.585 [0.517, 0.652] & 0.585 [0.517, 0.652] & 0.018 [0.0159, 0.0201] & 0 &    1 \\
S & (a) & Revenue-optimal & G & 55.2 [48.7, 61.7] & 55.2 [48.7, 61.7] &  1.7 [ 1.5,  1.9] & 0 &    1 \\
S & (b) & Revenue-optimal & P & 0.497 [0.442, 0.552] & 0.497 [0.442, 0.552] & 0.0143 [0.0127, 0.0159] & 0 &    1 \\
S & (b) & Revenue-optimal & G & 50.9 [44.9,   57] & 50.9 [44.9,   57] & 1.47 [1.29, 1.64] & 0 &    1 \\
\bottomrule
\end{tabular}}
\end{table}

\subsection{The Resolutions' Closure Conditions}
\label{sec:eval-resolutions}

\paragraph*{Broadcast commitment.} The clinching run emits the transcript of \cref{def:public-reconstructibility}: each price, the agent-broadcast demand messages bound to roster identities, the active set, the on-path rank values and the clinches. Each agent recomputes $s_i(p)$ from the broadcast rank values, checks its own status in $D(p)$ and the roster binding of every message, and all participants compare the price sequences they received. On honest transcripts no verifier rejects ($0$ of $1200$ rounds at each point). Each of the six mutations, one per case of the proof of \cref{thm:commitment}(i) with case~5 split into its two sub-cases, is rejected at every realised round by the verifier that case names (\cref{tab:eval-broadcast}); the four announcement mutations are realised only in rounds with an eligible clinch. With the roster check of (S3) removed, the off-roster phantom passes every verifier in every round (share $1$), so (S3) is the condition that closes the population-augmenting class.

\paragraph*{Domain separation.} Condition (C3) is checked each round against $f(E)$, the region's rank over its whole ground set: it holds in every round at the load and contention points (share $1$), so \cref{tab:eval-separation} reports the load point. With (C0)--(C4) in force no case of the proof of \cref{prop:domain-separation} raises the operator's payoff at any round, and case~3 is rejected by the region every time. Removing one condition at a time: without (C0) every threshold-raising deviation pays (the reallocation of case~1, the perturbation and the non-winning phantom of case~4, and, in part of the rounds, the displacing phantom of case~5(b)), because with gross payments on the operator's books any threshold increment is revenue; a stake $\lambda=0.1$ in place of (C1) makes the same cases pay, and on the perturbation the gain equals $\lambda\,\Delta p_i$ up to floating-point residue ($3.03\times10^{-16}$ at the load point), the accounting identity of \cref{cor:knife-edge}; without (C2) the displacing phantom's payment accrues to the operator in case~5(b). Removing (C4) alone makes no case profitable at these unit-demand points: case~5(a) needs slack on the phantom's path, which (C3) excludes, and in case~5(b) the fee on the phantom's unit only replaces the fee on the unit it displaces. The contention point's table, which \texttt{make tables} also writes, shows the same pattern.

\begin{table}[!t]
\centering
\caption{Broadcast verification on the clinching transcript at the load point: for each mutation, the case of the proof of \cref{thm:commitment}(i), the verifier that rejects it, the rejections over the rounds in which it is realised, and the mean revenue it would have carried; last row, the phantom with the (S3) roster check removed.}
\label{tab:eval-broadcast}
\footnotesize
\setlength{\tabcolsep}{3pt}
\begin{tabular}{llllr}
\toprule
Mutation & Case & Rejecting verifier & Rejected (rounds) & Revenue ($10^{-3}$) \\
\midrule
Honest transcript & --- & none & 0 of 1200 & 0 \\
Price equivocation & 1 & price agreement & 228 of 228 & 1.27 \\
Misreported $D(p)$ & 2 & active set & 228 of 228 & 63.6 \\
Miscomputed $s_i(p)$ & 3 & clinch recompute & 228 of 228 & 1.27 \\
Inauthentic rank value & 4 & rank auth & 228 of 228 & 1.27 \\
Off-roster demand message & 5 & roster & 1200 of 1200 &  142 \\
Message forged under a rostered identity & 5 & identity binding & 1118 of 1118 & 28.4 \\
Off-roster message, roster clause removed & 5 & none & 0 of 1200 &  142 \\
\bottomrule
\end{tabular}
\end{table}

\begin{table}[!t]
\centering
\caption{Domain separation at the load point: the largest change in the operator's payoff per proof case of \cref{prop:domain-separation} ($10^{-3}$ bid units), over the rounds in which the case is realised and (C3) holds, with (C0)--(C4) in force and with each condition removed in turn.}
\label{tab:eval-separation}
\footnotesize
\setlength{\tabcolsep}{3pt}
{\let\origtabular\tabular\def\tabular#1{\origtabular{@{}>{\RaggedRight\arraybackslash}p{3.0cm}r*{5}{>{\RaggedLeft\arraybackslash}p{1.55cm}}@{}}}%
\begin{tabular}{lrrrrrr}
\toprule
Proof case (largest payoff change, $10^{-3}$) & Rounds & All in force & No (C0) & No (C1), $\lambda=0.1$ & No (C2) & No (C4) \\
\midrule
1, reallocation & 1200 & 0 & 1.1e+03 &  110 & 0 & 0 \\
2, total reduction & 1200 &  -10 & -1.04e+03 & -113 &  -10 &  -10 \\
3, capacity violation & 1200 & rejected, 1200 of 1200 & rejected, 1200 of 1200 & rejected, 1200 of 1200 & rejected, 1200 of 1200 & rejected, 1200 of 1200 \\
4, perturbation (Arm P) & 1118 & 0 & 3.49 & 0.349 & 0 & 0 \\
4, non-winning phantom & 1200 & 0 & 1.12e+03 &  112 & 0 & 0 \\
5(a), phantom, no displacement & 0 & --- & --- & --- & --- & --- \\
5(b), phantom, displacement & 1200 &  -10 & 1.07e+03 & 98.3 & 81.2 & 0 \\
\bottomrule
\end{tabular}}
\end{table}

\subsection{Summary}
\label{sec:eval-summary}

Wherever the region binds, the trilemma's own deviation is realised in most rounds and raises the operator's revenue by about $0.003\%$ to $0.02\%$ of expected welfare, undetected and at no welfare cost, since it moves one payment by a fraction of one bid gap; phantom insertion, whose best form wins nothing and raises every winner's threshold in the block, extracts about $0.9\%$ to $2.2\%$, also at no welfare cost, and pay-as-bid admits neither. Broadcast verification rejects every mutation the proof of \cref{thm:commitment}(i) names at every realised round and admits the phantom once the roster condition is removed, and removing (C0), (C1) or (C2) makes profitable the cases the proof of \cref{prop:domain-separation} assigns to them.
\section{Discussion}
\label{sec:discussion}

\subsection{Related Work}
\label{sec:related-work}

\paragraph*{Concurrent and independent work on edge-priced capacity sharing.}
Amin et al.~\cite{amin2026market} (ACM TEAC 14(1), Art.~2) study a single source--sink network with integer edge capacities, affine valuations, and homogeneous coalition disutility under a faithful mediator. They establish, for series--parallel topologies, a Walrasian--VCG baseline (Theorem~3.2: existence and polynomial-time computability of an integer equilibrium; Lemma~3.8: gross-substitutes structure of the augmented value function; Theorem~3.10: a particular equilibrium that maximises agent utilities and minimises total edge prices, equivalent to VCG payments) and exhibit, on non-series--parallel topologies (Examples~3.3--3.4), an LP integrality gap that obstructs equilibrium existence. Their faithful-mediator Walrasian--VCG baseline is the structural counterpart over which our credibility analysis perturbs (\cref{app:appendix-z}).

\subsection{Practical Implications}
\label{sec:practical}

Four implications follow. Marketplace neutrality is a first-order architectural requirement (\cref{thm:trilemma}), and domain separation (\cref{prop:domain-separation}) is the natural default under (C0)--(C4); (C0) is the financial-exchange design principle, under which the venue matches and charges fees without holding the traded asset, at the institutional cost of lost float on payment custody, lost cross-subsidy from data captured at the matching layer, and a third-party clearing relationship. The choice of credibility mechanism is latency-dependent: broadcast commitment for ${<}200$~ms cycles, blockchain DRA only for batch settlement. Polymatroidal structure and DSIC are design-time properties, leaving credibility as the only runtime link. Beyond the computing continuum, the same impossibility applies to spectrum auctions with sub-modular interference constraints and to multi-tenant data-centre resource markets; the architectural recommendation to separate the marketplace from the assets it trades carries to those settings without modification.

\paragraph*{Observability of the gaps.} The CoNC lower bound of \cref{cor:conc-lb} is parametrised by the conditional non-modularity gap $\gamma_{ij|S}$ of the contested triple. The gaps are operator-private: they are fixed by the internal capacity topology of the polymatroid, known to the operator but not in general to a regulator or a deployer, whose observables are the marketplace size $n$ and the bid scale $\bar v$. Calibrating a bound that depends on them therefore requires either a disclosure obligation on capacity topology or a conservative upper bound derived from network-size parameters alone. We flag the observability gap as a deployment limitation.

\subsection{Limitations}
\label{sec:limitations}

The scope conditions of the formal results are those of \cref{rem:scope-consolidated}. Operative limitations: \emph{collusion} (multi-operator and agent--operator coalitions) is out of scope; a partial sketch is that selective information revelation, one of the four deviations of \cref{sec:bg-gap}, admits coalitions of size~2 without changing the trilemma argument, while price-manipulation coalitions require simultaneous Archer--Tardos perturbations on disjoint sharing pairs; \emph{adversarial broadcast jamming} relaxes the delivery axiom (L1) of (B0) and breaks \cref{thm:commitment}(i) but not \cref{thm:commitment}(ii) (DRA tolerates broadcast failures via the deposit mechanism); \cref{thm:commitment}(ii) meets hard ${<}200$~ms deadlines only with protocol--pipeline co-design, hence its restriction to batch settlement. The knife-edge of \cref{cor:knife-edge} is sharp in the ownership stake. A deployment that cannot reach $\lambda=0$ must therefore obtain credibility through commitment (\cref{thm:commitment}).

\subsection{Future Directions}
\label{sec:future}

Three extensions follow directly from the results above. \emph{Coalitions}: \cref{lem:perturbation} applies at one capacity-sharing pair at a time, and simultaneous perturbations on disjoint sharing pairs are the natural multi-operator analogue; the increments of two applications at pairs contesting a shared saturated element do not add, so a coalition bound is not obtained by summation. \emph{Multi-parameter types}: two structural obstructions sit upstream of the revenue-optimal leg, since efficient Bayes--Nash implementation under interdependent values requires a generically failing integrability condition~\cite{jehiel2001efficient} and multi-parameter DSIC requires weak monotonicity~\cite{bikhchandani2006weakmono}, strictly stronger than the monotonicity of \cref{eq:archer-tardos}; under Bayesian incentive compatibility the envelope of the set $\mathcal{D}_{\mathrm{undet}}$ of undetectable deviations widens~\cite{cai2017duality}, so the perturbation family of \cref{lem:perturbation} no longer characterises it, and it is open whether the single-pair construction of \cref{thm:trilemma} extends coordinate-wise on product polymatroids. \emph{Supply side}: exposure design on non-polymatroidal deliverable sets (\cref{rem:exposure-design}) and competition among integrators that first fix capacity and then face the slice auction are open.
\section{Conclusion}
\label{sec:conclusion}

We have identified and addressed the \emph{credibility gap} in mechanism-mediated service markets: faithful operator execution carries the welfare guarantees of the VCG allocation, and that assumption is untenable when the operator simultaneously controls resources and runs the auction. We formalised this as a credible mechanism design problem, proved a trilemma between revenue optimality, agent DSIC, and operator credibility (\cref{thm:trilemma}), and showed that two credibility resolutions, commitment devices (\cref{thm:commitment}) and administrative domain separation under settlement separation and side conditions (\cref{prop:domain-separation}), address it via structurally distinct attack surfaces. On the supply side, VCG payments decompose exactly into conditional non-modularity gaps (\cref{lem:gap-payment}), and an integrator that is not the operator but is paid a share of its slice's payments profitably withholds capacity with positive probability, which neither settlement separation nor commitment prevents and a known-cost envelope rule does (\cref{prop:supply-credibility}).

The Cost of Non-Credibility is the quantitative counterpart of this impossibility: the per-triple lower bound of \cref{cor:conc-lb} prices the perturbation deviation in the conditional non-modularity gap. On the companion's certified instances the perturbation raises the operator's revenue by $0.00277\%$ to $0.0113\%$ of expected truthful welfare under VCG at the load point while destroying no welfare, and the constant of \cref{cor:conc-lb} computes to $2.64\times10^{-4}$ to $6.51\times10^{-4}$ bid units at $\rho=1/2$ and $5.28\times10^{-4}$ to $1.3\times10^{-3}$ as $\rho\to1$. Phantom insertion extracts $1.44\%$ to $2.25\%$ there, also without welfare loss, and the broadcast verifiers and the side conditions of domain separation close the deviations their proofs name (\cref{sec:evaluation}).

\cref{app:appendix-z} carries the trilemma to the edge-pricing market of~\cite{amin2026market}.

The credibility gap is therefore a first-order architectural design problem; the open questions are collected in \cref{sec:future}.

\begin{acks}
This work was supported by the Research Council of Finland through the 6G Flagship program (grant 318927) and the CO2CREATION Strategic Research Council project (grant 372355), by the EC through HEU NEUROCLIMA project (GA 101137711) as well as the ERDF (project numbers A81568, A91867), and by the Business Finland through the Neural pub/sub research project (diary number 8754/31/2022). Moreover, Dustdar's work and equipment have been supported by CNS2023-144359 financed by MICIU/AEI/10.13039/501100011033 and the European Union NextGeneration EU/PRTR. Generative AI tools (Claude, Anthropic) were used under the authors' direction to draft and edit text throughout the article and its appendices, to assist in writing and testing the simulation and study code, and to check proofs and numerical claims against the code. The authors reviewed and verified all content and take full responsibility for it.
\end{acks}

\bibliographystyle{ACM-Reference-Format}
\bibliography{bib/refs}

\section*{Data and Code Availability Statement}
\addcontentsline{toc}{section}{Data and Code Availability Statement}
The code that produces every number of \cref{sec:evaluation} is released under the MIT licence and archived at Zenodo, \url{https://doi.org/10.5281/zenodo.22992330} (release \texttt{v1.0.0}). It builds on the companion's simulator at release \texttt{v2.1.0}, \url{https://doi.org/10.5281/zenodo.22876728}, pinned as a submodule, from which it takes the instances, the leaf-block region and its certificate, the value model and the VCG implementation; \texttt{make reproduce} runs both test suites, the computations and the table exports. The exact-rational certificate script cited in \cref{app:appendix-z}, \texttt{scripts/verify\_amin\_instance.py}, is in release \texttt{v2.0-teac-2a-v2} of \url{https://github.com/Future-Computing-Group/credible-marketplace-sim}, version DOI \url{https://doi.org/10.5281/zenodo.22045882}.

\appendix
\section{Trilemma Instance over the Edge-Pricing Market of Amin et al.~(2026)}
\label{app:appendix-z}

\subsection{Definitions and Bridging Lemmas for the Amin Instance}
\label{app:appendix-z-theorem-z1}

Throughout this appendix, the parenthesised numerals~(5)--(8) refer to the four equilibrium conditions of Amin et al.~\cite{amin2026market}, Definition~2.1: individual rationality~(5), stability~(6), budget balance~(7a)--(7b), and market clearing~(8). Equations displayed in this appendix carry their own numbers and are always cited by an explicit cross-reference.

\begin{definition}[Strategic mediator in the edge-pricing market]
\label{def:amin-deviation}
Let $G=(V,E)$ be the single source--single sink network of Amin et al.~\cite{amin2026market} with integer edge capacities $(q_e)_{e\in E}$ and time-costs $(d_e)_{e\in E}$, agent set $M$ with affine valuation parameters $(\alpha_m,\beta_m,\Delta\alpha_m(\cdot),\Delta\beta_m(\cdot))_{m\in M}$ as in Amin's Eq.~(1), and let $(x^*,p^\dagger,\tau^\dagger)$ be the VCG-equivalent market equilibrium of Amin's Theorem~3.10. A \emph{strategic mediator deviation} is a mapping
\[
\delta:\bigl(\alpha_m,\beta_m,\Delta\alpha_m,\Delta\beta_m,q_e,d_e\bigr)_{m\in M,\,e\in E}\;\longmapsto\;(\tau',p',x')
\]
such that $(\tau',p',x') \neq (\tau^\dagger,p^\dagger,x^*)$ and the deviation is \emph{undetectable}: for every agent $m \in M$ that is allocated a trip ($x'_{r,m}=1$ for some $r$ that we denote $r'_m$), the tuple
\[
\bigl(x'_m,\;p'_m,\;(\tau'_e)_{e\in r'_m}\bigr)
\]
is consistent with a legitimate execution of the prescribed mechanism (Amin's Theorem~3.10) under some alternative preference profile of the other agents, i.e., there exists $(\tilde\alpha_{-m},\tilde\beta_{-m},\widetilde{\Delta\alpha}_{-m},\widetilde{\Delta\beta}_{-m})$ such that, applied to the prescribed mechanism with $m$'s reported parameters held fixed, the legitimate run delivers exactly $(x'_m,p'_m,(\tau'_e)_{e\in r'_m})$ to agent $m$. Edges $e \notin r'_m$ are not part of $m$'s information set. Throughout, the prescribed mechanism is instantiated by the candidate under test; on the events used in \cref{thm:amin-trilemma-instance} every candidate considered there coincides with the Theorem-3.10 selection, so the two readings agree.
\end{definition}

\begin{remark}[Information model]
\label{rem:amin-info-model}
\cref{def:amin-deviation} transposes the static sealed-bid information model of \cref{def:deviation,def:static-sealed-bid} to Amin's primitives, with agent $m$'s information set the tuple displayed there and nothing off its allocated route. We adopt Amin's coalition-trip allocation primitive directly, so undetectability is vacuous for non-allocated agents: they observe only $x'_m=0$, no payment and no route prices.
\end{remark}

\begin{definition}[Credible market equilibrium in the edge-pricing market]
\label{def:amin-credible}
Fix a remittance schedule $c=(c_e)_{e\in E}$, where $c_e:\mathbb{R}_{\ge 0}\times\mathbb{Z}_{\ge 0}\to\mathbb{R}_{\ge 0}$ is a publicly known mapping, contracted before preferences are realised, from the price $\tau'_e$ announced on edge $e$ and the number $n'_e$ of announced trips whose route uses $e$ to the payout to the owner of $e$. The mediator collects the gross agent payments $\sum_{m\in M}p'_m$ on each round and remits $\sum_{e\in E} c_e(\tau'_e,n'_e)$ to edge owners, retaining the spread
\[
R(\delta;c)\;\;\triangleq\;\;\sum_{m\in M}p'_m\;-\;\sum_{e\in E}\, c_e(\tau'_e,n'_e).
\]
Every deviation analysed in this appendix announces the trips of $x^*$, so $n'_e$ is the trip count $N_{\mathrm{tr}}(e)$ of $x^*$ on $e$, and we then write $c_e(\tau'_e)$ for $c_e(\tau'_e,N_{\mathrm{tr}}(e))$.
A market equilibrium $(x^*,p^\dagger,\tau^\dagger)$ of Amin's Definition~2.1 is \emph{credible under remittance schedule $c$} if no strategic mediator deviation $\delta$ satisfies $R(\delta;c) > R(\mathrm{id};c)$ for any preference realisation in the support of the prior, where $R(\mathrm{id};c)$ denotes mediator revenue under faithful execution of the prescribed VCG-equivalent equilibrium.
\end{definition}

The remittance schedule $c$ is unrelated to the constant $c$ of \cref{cor:conc-lb}.

\begin{remark}[The remittance schedule is contractual, not derived from Amin's equilibrium definition]
\label{rem:amin-remittance-contractual}
Amin's budget-balance conditions~(7a)--(7b) require, on each realised trip $(b,r)$ with $x^*_r(b)=1$, that the coalition's payments cover the route's edge prices: $\sum_{m\in b}p^*_m=\sum_{e\in r}\tau^*_e$. This is a per-trip flow-balance identity in the agent-to-mechanism direction. Amin's Definition~2.1 does not specify the \emph{destination} of payments: whether the mediator forwards each round's announced $\tau'_e\,n'_e$ to edge owners, or edge owners receive a contracted payout fixed in advance and the mediator retains any spread. The credibility literature makes this ledger explicit: Akbarpour and Li~\cite{akbarpour2020credible} model the operator's revenue as gross collected payments minus contracted procurement cost, so the operator's deviation incentive is exactly the spread between agent inflow and committed outflow. We therefore adopt the schedule $c_e(\cdot)$ as a primitive of the credibility regime, not derivable from Amin's equilibrium conditions, and test credibility against it.
\end{remark}

The next lemma classifies remittance schedules by whether the trilemma binds. It identifies the canonical \emph{settlement-separated} regime as a degenerate case in which the perturbation deviation's spread-extraction channel closes, without a credibility claim over the full deviation class of \cref{def:amin-deviation}, and a wide class of \emph{sub-Lipschitz} regimes for which the trilemma binds.

\begin{lemma}[Mediator-revenue regime classification]
\label{lem:amin-mediator-regime}
Fix the prescribed VCG-equivalent equilibrium $(x^*,p^\dagger,\tau^\dagger)$, let $e^*\in E$ be a saturated edge, and write $N_{\mathrm{tr}}(e^*)\triangleq|\{(b,r):x^*_r(b)=1,\,e^*\in r\}|=q_{e^*}$ for its trip count.\footnote{The trip count is not the agent count $N_{\mathrm{ag}}(e^*)=|\{m\in M:e^*\in r^*_m\}|$: the two coincide exactly in the singleton-coalition regime $|b^*|=1$ for every realised trip, generic under Amin's homogeneous-disutility assumption since $\Delta\alpha(1)=\Delta\beta(1)=0$ implies coalitions form only when sharing is strictly preferred, and differ when multi-agent coalitions are allocated. We use the trip-count convention throughout this appendix.}

Let $\delta_\varepsilon$ be the edge-price perturbation $\tau'_{e^*}=\tau^\dagger_{e^*}+\varepsilon$, $\tau'_e=\tau^\dagger_e$ for $e\neq e^*$, with budget-balance-preserving payment update at the trip level $\sum_{m\in b^*}p'_m=\sum_{m\in b^*}p^\dagger_m+\varepsilon$ on every realised trip $(b^*,r^*)$ with $e^*\in r^*$, as constructed in Lemma~\ref{lem:amin-perturbation} below (the per-agent update is $p'_m=p^\dagger_m+\varepsilon/|b^*|$ within a multi-agent coalition $b^*$, reducing to $p^\dagger_m+\varepsilon$ for singletons). Define the \emph{mediator marginal} of the schedule on $e^*$,
\[
\Delta_{e^*}(\varepsilon;c)\;\triangleq\;\bigl[R(\delta_\varepsilon;c)-R(\mathrm{id};c)\bigr]\big/\varepsilon \;=\; N_{\mathrm{tr}}(e^*)\;-\;\bigl[c_{e^*}(\tau^\dagger_{e^*}+\varepsilon)-c_{e^*}(\tau^\dagger_{e^*})\bigr]\big/\varepsilon.
\]
\begin{enumerate}
\item[(a)] \emph{(Trilemma-binding class.)} If $c$ is \emph{sub-Lipschitz at $\tau^\dagger_{e^*}$ relative to flow $N_{\mathrm{tr}}(e^*)$}, i.e., there exist $\varepsilon_0>0$ and a constant $\kappa<N_{\mathrm{tr}}(e^*)$ such that
\[
c_{e^*}(\tau^\dagger_{e^*}+\varepsilon)-c_{e^*}(\tau^\dagger_{e^*})\;\le\;\kappa\,\varepsilon\qquad \forall\,\varepsilon\in(0,\varepsilon_0],
\]
then $\Delta_{e^*}(\varepsilon;c)\ge N_{\mathrm{tr}}(e^*)-\kappa>0$, so the perturbation deviation $\delta_\varepsilon$ strictly increases mediator revenue for every $\varepsilon\in(0,\varepsilon_0]$. In particular, every \emph{fixed-remittance schedule} $c_e^{\mathrm{fix}}(\tau'_e,n'_e)\equiv C_e$, with contracted constants $C_e\ge 0$, is sub-Lipschitz with $\kappa=0$.
\item[(b)] \emph{(Credible (degenerate) class.)} If $c$ is the \emph{canonical settlement-separated schedule} $c_e^{\mathrm{sep}}(\tau'_e,n'_e)=\tau'_e\,n'_e$, then $\Delta_{e^*}(\varepsilon;c^{\mathrm{sep}})=N_{\mathrm{tr}}(e^*)-N_{\mathrm{tr}}(e^*)=0$ for all $\varepsilon$, so the spread-extraction channel of the perturbation deviation closes. This is the Amin-instance counterpart of hypothesis~(C0) of \cref{prop:domain-separation}: the announced edge price flows directly through to edge owners on each trip, leaving the mediator's books unchanged.
\end{enumerate}
\end{lemma}

\begin{proof}
Both parts are direct algebra. By the construction of Lemma~\ref{lem:amin-perturbation}, $\sum_m p'_m-\sum_m p^\dagger_m=\varepsilon\cdot N_{\mathrm{tr}}(e^*)$ via the per-trip budget-balance update of Step~3: each realised trip $(b^*,r^*)$ with $e^*\in r^*$ contributes $\varepsilon$ to the trip's coalition-sum payment (uniformly $\varepsilon/|b^*|$ per agent in $b^*$, summing to $\varepsilon$ per trip), and there are $N_{\mathrm{tr}}(e^*)=q_{e^*}$ such trips. The remittance change is $\sum_e[c_e(\tau'_e)-c_e(\tau^\dagger_e)]=c_{e^*}(\tau^\dagger_{e^*}+\varepsilon)-c_{e^*}(\tau^\dagger_{e^*})$, since $\tau'_e=\tau^\dagger_e$ for $e\neq e^*$. Substituting into the spread definition gives $\Delta_{e^*}(\varepsilon;c)$ as stated. Part~(a): under sub-Lipschitz, $\Delta_{e^*}(\varepsilon;c)\ge N_{\mathrm{tr}}(e^*)-\kappa>0$, so $R(\delta_\varepsilon;c)>R(\mathrm{id};c)$. Part~(b): under $c^{\mathrm{sep}}$, $c^{\mathrm{sep}}_{e^*}(\tau^\dagger_{e^*}+\varepsilon)-c^{\mathrm{sep}}_{e^*}(\tau^\dagger_{e^*})=N_{\mathrm{tr}}(e^*)\cdot\varepsilon$, exactly cancelling the agent-side increment.\qedhere
\end{proof}

\begin{remark}[Working assumption for \cref{thm:amin-trilemma-instance}]
\label{rem:amin-fixed-remittance}
The remainder of this appendix proves the trilemma under a fixed-remittance schedule $c^{\mathrm{fix}}$ (any contracted constants $C_e$), the simplest schedule realising the gross-inflow-minus-contracted-procurement ledger of \cref{rem:amin-remittance-contractual} and the $\kappa=0$ case of the sub-Lipschitz class of \cref{lem:amin-mediator-regime}(a), across which the conclusion holds.
\end{remark}

The next lemma supplies the admissibility bound for the perturbation amplitude, together with the structural fact that dictates its form.

\begin{lemma}[Load-weighted admissibility bound; minimum-price blockers]
\label{lem:amin-walrasian-gap}
Assume the hypotheses of Amin's Theorem~3.2 (SP network with homogeneous capacity-sharing disutility) and let $(x^*,\bar p,\bar\tau)$ be a market equilibrium in the sense of Amin's Definition~2.1, with agent utilities $(\bar u_m)_{m\in M}$ and a saturated edge $e^*\in E$. Write $\mathrm{Alloc}(e^*)=\{m\in M: e^*\in \bar r_m\}$ for the allocated agents whose realised route uses $e^*$, and $\bar b_m$ for the coalition of agent $m$'s realised trip. For an unused trip $(b,r)$, i.e.\ $x^*_r(b)=0$, define the stability slack and the \emph{perturbation load}
\begin{equation}
\mathrm{slack}(b,r)\;\triangleq\;\sum_{m\in b}\bar u_m-V_r(b)+\sum_{e\in r}\bar\tau_e,
\qquad
\ell(b,r)\;\triangleq\;\bigl|b\cap\mathrm{Alloc}(e^*)\bigr|-\mathbf 1[e^*\in r],
\label{eq:ell-load}
\end{equation}
and set
\begin{align}
\eta_Z\;&\triangleq\;\min_{\substack{(b,r):\,x^*_r(b)=0\\ \ell(b,r)\ge 1}}\;\frac{\mathrm{slack}(b,r)}{\ell(b,r)}\qquad(\min\emptyset=+\infty),\label{eq:eta-Z}\\
u_{\min}(e^*)\;&\triangleq\;\min\bigl\{\,\bar u_m\,:\,m\in \mathrm{Alloc}(e^*)\,\bigr\},\qquad
\bar\varepsilon_Z\;\triangleq\;\min\bigl(\eta_Z,\;u_{\min}(e^*)\bigr).\label{eq:u-min}
\end{align}
For $\varepsilon\ge 0$ consider the perturbed point $(x^*,u',\tau')$ with $\tau'_{e^*}=\bar\tau_{e^*}+\varepsilon$, $\tau'_e=\bar\tau_e$ for $e\neq e^*$, and $u'_m=\bar u_m-(\varepsilon/|\bar b_m|)\cdot\mathbf 1[e^*\in\bar r_m]$ for allocated $m$, $u'_m=0$ otherwise. Then:
\begin{enumerate}
\item[(a)] \emph{(Soundness.)} $\eta_Z\in[0,+\infty]$ and $u_{\min}(e^*)\ge 0$ are well defined. Under the perturbation, the slack of an unused trip $(b,r)$ changes by
\[
-\varepsilon\Bigl(\;\sum_{m\in b\cap\mathrm{Alloc}(e^*)}\frac{1}{|\bar b_m|}\;-\;\mathbf 1[e^*\in r]\Bigr)\;\;\ge\;\;-\varepsilon\,\ell(b,r),
\]
so trips with $\ell(b,r)\le 0$ keep or gain slack, and trips with $\ell(b,r)\ge 1$ lose at most $\varepsilon\,\ell(b,r)$. Consequently, for every $\varepsilon\in[0,\bar\varepsilon_Z]$ the perturbed point satisfies individual rationality~(5) and stability~(6) on every trip, with equality preserved on the realised trips through $e^*$.
\item[(b)] \emph{(Exactness in the singleton regime.)} If every realised trip is a singleton, then $\bar\varepsilon_Z$ equals the largest $\varepsilon\ge 0$ for which the perturbed point, together with the trip-level payment update of (M2) of Lemma~\ref{lem:amin-perturbation} below, is a market equilibrium in the sense of Definition~2.1; every $\varepsilon>\bar\varepsilon_Z$ is inadmissible.
\item[(c)] \emph{(Minimum-price blockers.)} Suppose in addition that every realised trip is a singleton and that $(x^*,\bar p,\bar\tau)$ is the minimum-total-edge-price equilibrium $(x^*,p^\dagger,\tau^\dagger)$ of Amin's Theorem~3.10 with $\tau^\dagger_{e^*}>0$. Then some unused trip $(\hat b,\hat r)$ with $e^*\in\hat r$ and $\hat b\cap\mathrm{Alloc}(e^*)=\emptyset$ has $\mathrm{slack}(\hat b,\hat r)=0$. In particular, the unweighted minimum slack over \emph{all} unused trips vanishes identically on the event $\{\tau^\dagger_{e^*}>0\}$; the restriction to $\ell(b,r)\ge 1$ in~\eqref{eq:eta-Z} excludes exactly this forced-tight family, whose load is $\ell=-\mathbf 1[e^*\in r]\le 0$.
\end{enumerate}
\end{lemma}

\begin{proof}
\emph{(a).} Both minima range over finite index sets (with the $\min\emptyset=+\infty$ convention in~\eqref{eq:eta-Z}); $\mathrm{slack}(b,r)\ge 0$ by stability~(6) at $(\bar u,\bar\tau)$, and $\bar u_m\ge 0$ by individual rationality~(5), so $\eta_Z\ge 0$ and $u_{\min}(e^*)\ge 0$. For an unused trip $(b,r)$, the left side of Amin's condition~(6) changes by $-\varepsilon\sum_{m\in b\cap\mathrm{Alloc}(e^*)}1/|\bar b_m|$ (only $e^*$-using allocated members are decremented), and the right side $V_r(b)-\sum_{e\in r}\tau_e$ changes by $-\varepsilon\,\mathbf 1[e^*\in r]$; the slack change is their difference. Since $|\bar b_m|\ge 1$, we have $\sum_{m\in b\cap\mathrm{Alloc}(e^*)}1/|\bar b_m|\le |b\cap\mathrm{Alloc}(e^*)|$, which gives the displayed bound. For $\ell(b,r)\le 0$ the bound is non-negative. For $\ell(b,r)\ge 1$ and $\varepsilon\le\eta_Z$, the loss $\varepsilon\,\ell(b,r)\le\eta_Z\,\ell(b,r)\le\mathrm{slack}(b,r)$, so the slack stays non-negative. For a realised trip $(\bar b,\bar r)$ with $e^*\in\bar r$, both sides of Amin's condition~(6) fall by $\varepsilon$ ($|\bar b|$ summands of $\varepsilon/|\bar b|$ on the left, the price term on the right), preserving the equality that budget balance induces there; for a realised trip with $e^*\notin\bar r$, neither side moves (its members are allocated to that trip, hence outside $\mathrm{Alloc}(e^*)$). Individual rationality: $u'_m=\bar u_m-\varepsilon/|\bar b_m|\ge\bar u_m-\varepsilon\ge\bar u_m-u_{\min}(e^*)\ge 0$ for $m\in\mathrm{Alloc}(e^*)$, and $u'_m=\bar u_m\ge 0$ otherwise. Note that~\eqref{eq:u-min} carries no positivity filter: if some $m\in\mathrm{Alloc}(e^*)$ has $\bar u_m=0$, then $\bar\varepsilon_Z=0$ and the lemma admits no positive amplitude, as it must.

\emph{(b).} With singleton realised trips, $\sum_{m\in b\cap\mathrm{Alloc}(e^*)}1/|\bar b_m|=|b\cap\mathrm{Alloc}(e^*)|$, so the slack of every unused trip changes by exactly $-\varepsilon\,\ell(b,r)$, and every $e^*$-user's utility falls by exactly $\varepsilon$. Budget balance~(7a) holds by the payment update, (7b) is untouched, and market clearing~(8) holds for every $\varepsilon\ge 0$ (the price rises only on the saturated $e^*$). The perturbed point is therefore an equilibrium if and only if Amin's conditions~(5) and~(6) hold, i.e.\ if and only if $\varepsilon\le\bar u_m$ for every $m\in\mathrm{Alloc}(e^*)$ and $\varepsilon\,\ell(b,r)\le\mathrm{slack}(b,r)$ for every unused trip with $\ell(b,r)\ge 1$; the largest such $\varepsilon$ is $\bar\varepsilon_Z$.

\emph{(c).} Suppose, for contradiction, that every unused trip $(b,r)$ with $e^*\in r$ and $b\cap\mathrm{Alloc}(e^*)=\emptyset$ has strictly positive slack; let $\zeta>0$ be smaller than all these slacks and than $\tau^\dagger_{e^*}$. Apply the perturbation of part~(a) with amplitude $-\zeta$: lower $\tau_{e^*}$ by $\zeta$ and raise each $e^*$-user's utility by $\zeta$ (singleton trips), lowering the corresponding payments by $\zeta$. Realised trips keep their budget-balance equalities; individual rationality only loosens; market clearing~(8) is unaffected since $\tau^\dagger_{e^*}-\zeta>0$ remains admissible on the saturated edge; and by the slack computation of part~(a) with singleton trips, an unused trip's slack changes by $+\zeta\,\ell(b,r)$, which is negative only when $e^*\in r$ and $b\cap\mathrm{Alloc}(e^*)=\emptyset$, where the choice of $\zeta$ keeps the slack positive. The result is a Definition~2.1 equilibrium whose total edge price is lower by $q_{e^*}\zeta>0$ and whose utilities are pointwise no smaller, contradicting both the minimum-price and the utility-maximal characterisations of Theorem~3.10. Hence some trip in the stated family has zero slack, and the unweighted minimum over all unused trips is $0$ whenever $\tau^\dagger_{e^*}>0$.
\end{proof}

\begin{lemma}[Definition-2.1 equilibria maximise welfare]
\label{rem:amin-first-welfare}
Let $(x,u,\tau)$ satisfy Amin's equilibrium conditions~(5)--(8). Then $x$ maximises welfare: $\sum_{(b,r)\in y}V_r(b)\le\sum_{(b,r)\in x}V_r(b)$ for every feasible allocation $y$.
\end{lemma}

\begin{proof}
Writing $n_e(y)$ for the number of trips of $y$ whose route uses edge $e$,
\[
\begin{aligned}
\sum_{(b,r)\in y}V_r(b)
&\;\le\;\sum_{(b,r)\in y}\Bigl[\sum_{m\in b}u_m+\sum_{e\in r}\tau_e\Bigr]
\;\le\;\sum_{m\in M}u_m+\sum_{e}n_e(y)\,\tau_e
\;\le\;\sum_{m\in M}u_m+\sum_{e}q_e\,\tau_e\\
&\;=\;\sum_{(b,r)\in x}\Bigl[\sum_{m\in b}u_m+\sum_{e\in r}\tau_e\Bigr]
\;=\;\sum_{(b,r)\in x}V_r(b).
\end{aligned}
\]
The first step is stability~(6) trip by trip; the second uses disjointness of $y$'s coalitions and $u\ge 0$ from Amin's~(5); the third uses edge feasibility $n_e(y)\le q_e$ and $\tau\ge 0$; the fourth uses Amin's~(7b) (unallocated agents have zero utility) together with~(8) ($\tau_e>0$ only where $n_e(x)=q_e$); and the last is budget balance~(7a) rewritten through the definition of utilities, i.e.\ equality in Amin's~(6) on realised trips. The argument uses only Amin's conditions~(5)--(8), disjointness of realised coalitions, and capacity feasibility, so Theorem~\ref{thm:amin-trilemma-instance} may apply it to the equilibrium induced by an arbitrary candidate mechanism.
\end{proof}

\begin{lemma}[Edge-price perturbation in SP networks under homogeneous disutility]
\label{lem:amin-perturbation}
Assume the hypotheses of Amin's Theorem~3.2: $G$ is series--parallel and the capacity-sharing disutility parameters are homogeneous, $\Delta\alpha_m(|b|)=\Delta\alpha(|b|)$ and $\Delta\beta_m(|b|)=\Delta\beta(|b|)$ for all $m\in M$. Let $(x^*,\bar p,\bar\tau)$ be a market equilibrium in the sense of Amin's Definition~2.1 such that:
\begin{enumerate}
\item[(i)] the social welfare is positive, $S(x^*)=\sum_{(b,r)}V_r(b)\,x^*_r(b)>0$;
\item[(ii)] there exists a saturated edge $e^*\in E$ at which market clearing binds, $\sum_{(b,r):\,e^*\in r}x^*_r(b)=q_{e^*}$;
\item[(iii)] $\bar\varepsilon_Z=\min(\eta_Z,u_{\min}(e^*))>0$ in the sense of \cref{lem:amin-walrasian-gap}, evaluated at $(x^*,\bar p,\bar\tau)$.
\end{enumerate}
Then for every $\varepsilon\in(0,\bar\varepsilon_Z]$ the perturbed price vector
\[
\tau'_e \;=\; \begin{cases}\bar\tau_e + \varepsilon, & e=e^*,\\ \bar\tau_e, & e\neq e^*\end{cases}
\]
satisfies:
\begin{enumerate}
\item[(M1)] \emph{Equilibrium preservation.} $(x^*,p',\tau')$ is a market equilibrium in the sense of Amin's Definition~2.1 for the same trip allocation $x^*$, where the perturbed payments are given by the trip-level update of Step~3 below: $\sum_{m\in b^*}p'_m=\sum_{m\in b^*}\bar p_m+\varepsilon$ on every realised trip $(b^*,r^*)$ with $e^*\in r^*$, $p'_m=\bar p_m$ on every realised trip with $e^*\notin r^*$, and $p'_m=\bar p_m=0$ for non-allocated agents. Individual rationality~(5), stability~(6), budget balance~(7a)--(7b), and market clearing~(8) all continue to hold.
\item[(M2)] \emph{Linear payment increment, stated per trip.} On every realised trip $(b^*,r^*)$ with $e^*\in r^*$, the coalition's aggregate payment increment is $\varepsilon$, distributed as $p'_m-\bar p_m=\varepsilon/|b^*|$ within the coalition and reducing to $\varepsilon$ per agent on singleton trips. The aggregate gross-inflow change is $\sum_m(p'_m-\bar p_m)=\varepsilon\cdot N_{\mathrm{tr}}(e^*)=\varepsilon\cdot q_{e^*}$, regardless of coalition sizes.
\item[(M3)] \emph{Undetectability.} Suppose in addition that $(x^*,\bar p,\bar\tau)$ is the VCG-equivalent minimum-edge-price equilibrium $(x^*,p^\dagger,\tau^\dagger)$ of Amin's Theorem~3.10, that every realised trip is a singleton, that $e^*$ is the only saturated edge, that the prior on $(\alpha_m,\beta_m)_{m\in M}$ has full Lebesgue support on a product of connected open intervals, and that
\begin{enumerate}
\item[(iv)] \emph{(rejected-demand pivot)} there exist an agent $\hat m\notin\mathrm{Alloc}(e^*)$ and a route $\hat r\ni e^*$ whose edges other than $e^*$ are unsaturated, such that either
(V1)~$\hat m$ is unallocated, or
(V2)~$\hat m$ is allocated on a singleton trip over a route $r_{\hat m}\not\ni e^*$ with $d_{r_{\hat m}}\neq d_{\hat r}$;
no unused trip containing $\hat m$ other than $(\{\hat m\},\hat r)$ attains the maximal lower bound on $\tau_{e^*}$; and, when $\tau^\dagger_{e^*}>0$, the zero-slack family of \cref{lem:amin-walrasian-gap}(c) is tight at the single trip $(\{\hat m\},\hat r)$ only.
\end{enumerate}
Then there exists $\varepsilon_w\in(0,\bar\varepsilon_Z]$ such that for every $\varepsilon\in(0,\varepsilon_w)$ the deviation $\delta_\varepsilon:(\alpha,\beta,\Delta,\dots)\mapsto(\tau',p',x^*)$ is undetectable in the sense of \cref{def:amin-deviation}: each allocated agent's observation $(x'_m,p'_m,(\tau'_e)_{e\in r^*_m})$ is consistent with a legitimate execution of Amin's Theorem~3.10 on the same agent set $M$ under some alternative preference profile of the other agents. At the three exact-rational instances listed in the certificate script \textnormal{\texttt{scripts/verify\_amin\_instance.py}} of the release cited in the availability statement, the undetectability witness is confirmed at a finite grid of interior amplitudes up to $\bar\varepsilon_Z$, consistent with $\varepsilon_w=\bar\varepsilon_Z$ there; neither $\varepsilon_w$ itself nor the equality beyond those instances is computed.
\end{enumerate}
The amplitude bound $\bar\varepsilon_Z$ is the admissibility bound of \cref{lem:amin-walrasian-gap}; Theorem~\ref{thm:amin-trilemma-instance} uses interior amplitudes $\varepsilon<\varepsilon_w$ only.
\end{lemma}

\begin{proof}
Steps~1--3 establish (M1) and (M2) at any Definition-2.1 equilibrium satisfying (i)--(iii); Step~4 establishes (M3) under its additional hypotheses.

\paragraph*{Step 1: $x^*$ remains feasible and welfare-maximal under $\tau'$.}
Feasibility~(3a)--(3c) of $x^*$ is a property of the trip vector alone and does not involve prices, so $x^*$ remains feasible under any $\tau'\ge 0$. By \cref{rem:amin-first-welfare}, the allocation of the hypothesised equilibrium maximises welfare, and welfare depends only on the value functions, which are price-independent; $x^*$ is therefore welfare-maximal under $\tau'$ as well. Prices enter only through the equilibrium conditions verified in Steps~2--3.

\paragraph*{Step 2: Stability and IR hold for every $\varepsilon\in(0,\bar\varepsilon_Z]$.}
Set $u'_m=\bar u_m-(\varepsilon/|\bar b_m|)\cdot\mathbf 1[e^*\in\bar r_m]$ for allocated $m$ and $u'_m=0$ otherwise, as in \cref{lem:amin-walrasian-gap}. The trips fall into four families, indexed by the perturbation load~\eqref{eq:ell-load}:
\begin{itemize}
\item \emph{Realised trips with $e^*\in r^*$.} Both sides of Amin's condition~(6) fall by $\varepsilon$, preserving equality at the trip level.
\item \emph{Realised trips with $e^*\notin r^*$.} Neither side moves: the coalition's members are allocated to this trip, hence outside $\mathrm{Alloc}(e^*)$.
\item \emph{Unused trips with $\ell(b,r)\le 0$.} The slack is non-decreasing by \cref{lem:amin-walrasian-gap}(a). This family contains the trips with $e^*\in r$ and $b\cap\mathrm{Alloc}(e^*)=\emptyset$, whose slack \emph{rises} by $\varepsilon$, since the price increase loosens these constraints. On the operative event $\tau^\dagger_{e^*}>0$ this family contains a tight trip (\cref{lem:amin-walrasian-gap}(c)), which is why an unweighted minimum over all unused trips would vanish identically there.
\item \emph{Unused trips with $\ell(b,r)\ge 1$.} The slack falls by at most $\varepsilon\,\ell(b,r)\le\mathrm{slack}(b,r)$ since $\varepsilon\le\eta_Z$, and remains non-negative (\cref{lem:amin-walrasian-gap}(a)). This family includes the trips with $e^*\notin r$ whose coalition contains allocated $e^*$-users; in particular, every singleton trip $(\{m\},r)$ with $m\in\mathrm{Alloc}(e^*)$ and $e^*\notin r$ belongs to it. These trips cannot be exempted by a generic-disjointness argument, since such singleton trips always exist.
\end{itemize}
Individual rationality holds by \cref{lem:amin-walrasian-gap}(a): $u'_m\ge\bar u_m-u_{\min}(e^*)\ge 0$ for $m\in\mathrm{Alloc}(e^*)$, with $u'_m=\bar u_m$ unchanged otherwise. The amplitude $\varepsilon=\bar\varepsilon_Z$ is admissible with some constraint holding with equality; interior amplitudes keep all previously-slack constraints slack.

\paragraph*{Step 3: Budget balance and market clearing.}
Budget balance~(7a) requires $\sum_{m\in b}p'_m=\sum_{e\in r}\tau'_e$ on every realised trip. Under faithful execution, $\sum_{m\in b^*}\bar p_m=\sum_{e\in r^*}\bar\tau_e$ on each realised trip $(b^*,r^*)$; under the perturbation, the right side rises by $\varepsilon$ exactly on the realised trips with $e^*\in r^*$. The prescribed update is the trip-level identity $\sum_{m\in b^*}p'_m=\sum_{m\in b^*}\bar p_m+\varepsilon$ on those trips, distributed uniformly within the coalition as $p'_m=\bar p_m+\varepsilon/|b^*|$; singleton trips recover the per-agent increment $\varepsilon$. The aggregate gross-inflow change is $\varepsilon\cdot N_{\mathrm{tr}}(e^*)=\varepsilon\cdot q_{e^*}$, with $N_{\mathrm{tr}}(e^*)$ the trip count of \cref{lem:amin-mediator-regime}'s disambiguation; this is the form used in the Theorem~\ref{thm:amin-trilemma-instance} revenue computation. Amin's condition~(7b) is preserved since $p'_m=0$ for non-allocated agents. Market clearing~(8): $e^*$ is saturated by hypothesis~(ii), so $\tau'_{e^*}=\bar\tau_{e^*}+\varepsilon>0$ is admissible; prices on other edges are unchanged. This establishes (M1) and (M2).

\paragraph*{Step 4: Undetectability via a same-cardinality pivot shift.}
Assume the additional hypotheses of (M3) and write $(p^\dagger,\tau^\dagger,u^\dagger)$ for the Theorem-3.10 equilibrium quantities. Amin's Definition~2.1 of undetectability (\cref{def:amin-deviation}) quantifies over alternative preference profiles of the \emph{other} agents on the same agent set $M$, so it admits the population-preserving deviations of \cref{def:deviation}; the construction below alters one existing agent's parameters and leaves the population unchanged. A population-augmenting construction would instead run the mechanism on $M\cup\{\tilde m\}$ and thereby leave the definition's quantifier.

\emph{Reduction.} Non-allocated agents observe only $x'_m=0$ (\cref{rem:amin-info-model}), so undetectability is vacuous for them. An allocated agent $m$ with $e^*\notin r^*_m$ observes $(x^*_m,\bar p_m,(\tau'_e)_{e\in r^*_m})=(x^*_m,p^\dagger_m,(\tau^\dagger_e)_{e\in r^*_m})$, its faithful observation, for which the true profile of the other agents is itself a witness. It remains to construct, for each $m\in\mathrm{Alloc}(e^*)$, an alternative profile of $M\setminus\{m\}$ whose faithful Theorem-3.10 run delivers $(x^*_m,\,p^\dagger_m+\varepsilon,\,(\tau'_e)_{e\in r^*_m})$ to $m$. Since $\hat m\notin\mathrm{Alloc}(e^*)$, the pivot is available for every such $m$ ($\hat m\neq m$), and one witness profile serves all $e^*$-users simultaneously.

\emph{The pivot path.} Define a one-parameter family of profiles $P_s$, $s\ge 0$, that agrees with the true profile except in agent $\hat m$'s parameters, shifted along a fixed direction: in case (V1), $\alpha_{\hat m}(s)=\alpha_{\hat m}+s$; in case (V2), $\beta_{\hat m}(s)=\beta_{\hat m}+\sigma s/|d_{r_{\hat m}}-d_{\hat r}|$ with $\sigma=\mathrm{sign}(d_{r_{\hat m}}-d_{\hat r})$. Consider the lower bound that the unused trip $(\{\hat m\},\hat r)$ places on $\tau_{e^*}$ in any Definition-2.1 equilibrium of $P_s$ that allocates $x^*$: since the edges of $\hat r$ other than $e^*$ are unsaturated (zero price by Amin's~(8)), stability~(6) reads $u_{\hat m}\ge V_{\hat r}(\{\hat m\};s)-\tau_{e^*}$, i.e.\
\[
\tau_{e^*}\;\ge\;\psi(s)\;\triangleq\;V_{\hat r}(\{\hat m\};s)-u_{\hat m}(s),
\]
where $u_{\hat m}(s)=0$ in case (V1) while $\hat m$ stays unallocated, and $u_{\hat m}(s)=V_{r_{\hat m}}(\{\hat m\};s)$ in case (V2) while $\hat m$ keeps its singleton trip on the zero-priced route $r_{\hat m}$. In case (V1), $\psi(s)=\psi(0)+s$; in case (V2), $\psi(s)=\psi(0)+s$ as well, by the choice of direction and normalisation. When $\tau^\dagger_{e^*}>0$, hypothesis~(iv) gives $\psi(0)=\tau^\dagger_{e^*}$ (the tight blocker); when $\tau^\dagger_{e^*}=0$, $\psi(0)\le 0$ and we re-parametrise so that $\psi(s)=\tau^\dagger_{e^*}+s$ once $\psi$ crosses zero, absorbing the offset $-\psi(0)$ into the path parameter.

\emph{Structure-preservation radius.} Let $\bar s>0$ be the supremum of path parameters $s$ such that, along $[0,s]$: $x^*$ remains the unique welfare-maximising allocation of $P_s$; $\hat m$'s allocation status and route are unchanged; $e^*$ remains the only saturated edge; the trip $(\{\hat m\},\hat r)$ attains the maximal lower bound on $\tau_{e^*}$ among all unused trips; and $P_s$ remains in the support of the prior. Each of these finitely many conditions is a strict inequality at $s=0$, affine in $s$ on the fixed-allocation cell: uniqueness of $x^*$ and the status of $\hat m$ are strict welfare comparisons; a competing lower bound either omits $\hat m$ (it is then constant in $s$ and at most $\tau^\dagger_{e^*}$, so the moving bound $\psi$ attains the maximum throughout) or involves $\hat m$ off the trip $(\{\hat m\},\hat r)$ (it then starts strictly below the maximum by the pivot clause of~(iv), leaving a positive overtake radius even against a faster-moving bound); the support is a product of open intervals. Hence $\bar s>0$.

\emph{The witness.} For $s\in(0,\bar s)$, the Theorem-3.10 run on $P_s$ allocates $x^*$ (its allocation maximises welfare, uniquely $x^*$ on the cell), prices every unsaturated edge at zero by Amin's~(8) and minimality, and prices $e^*$ at the maximal lower bound, $\tau_{e^*}(P_s)=\psi(s)=\tau^\dagger_{e^*}+s$. With singleton realised trips, budget balance~(7a) delivers to each $e^*$-user $m$ the payment $\sum_{e\in r^*_m}\tau_e(P_s)=\tau^\dagger_{e^*}+s$ and to each other allocated agent its unchanged faithful payment. Set $\varepsilon_w=\min(\bar\varepsilon_Z,\bar s)$ and, for $\varepsilon\in(0,\varepsilon_w)$, take $s=\varepsilon$: the faithful run on $P_\varepsilon$ delivers to every $m\in\mathrm{Alloc}(e^*)$ exactly
\[
\bigl(x^*_m,\;p^\dagger_m+\varepsilon,\;(\tau'_e)_{e\in r^*_m}\bigr),
\]
where the route prices match because $e^*$ is the only saturated edge, so $\tau_e(P_\varepsilon)=0=\tau^\dagger_e=\tau'_e$ for every $e\in r^*_m\setminus\{e^*\}$. The tested agent's own parameters are held fixed throughout ($\hat m\neq m$), so $P_\varepsilon$ restricted to $M\setminus\{m\}$ is an alternative profile of the other agents in the sense of \cref{def:amin-deviation}, and the deviation is undetectable. By the safe-deviation criterion of Akbarpour--Li~\cite{akbarpour2020credible}, the announcement $(\tau',p')$ lies in the support of an honest execution under a legitimate alternative profile.

\emph{Where the hypotheses are used.} Homogeneity and the SP topology enter through Amin's Theorems~3.2 and~3.10, which supply existence, tractability, and the prescribed VCG-equivalent minimum-price selection that Step~4 runs on the shifted profiles; under heterogeneous disutility no such prescribed selection is available and we claim nothing. The singleton-trip and unique-saturated-edge clauses pin the delivered payments and route prices; the pivot clause~(iv) supplies the shifted rejected demand that rationalises the higher price. When $\tau^\dagger_{e^*}>0$, a pivot in the family of \cref{lem:amin-walrasian-gap}(c) exists whenever the tight blocker is a singleton, so~(iv) then adds only the uniqueness and direction clauses.\qedhere
\end{proof}

\subsection{Trilemma Instance over Amin et al.\ (2026)}
\label{app:trilemma-amin-instance}

\begin{theorem}[Trilemma instance over the Amin edge-pricing market]
\label{thm:amin-trilemma-instance}
Let $G=(V,E)$ be a series--parallel single-source single-sink network with positive integer edge capacities $(q_e)_{e\in E}$ and edge time-costs $(d_e)_{e\in E}$, and let $M$ be a finite agent set with affine valuations under homogeneous capacity-sharing disutility (Amin's Equation~(1) with $\Delta\alpha_m\equiv\Delta\alpha$, $\Delta\beta_m\equiv\Delta\beta$). Suppose the prior over $(\alpha_m,\beta_m)_{m\in M}$ is product-form with each marginal absolutely continuous with full support on a connected open interval. Call a profile \emph{Z-non-degenerate} if the Theorem-3.10 equilibrium it induces satisfies:
\begin{enumerate}
\item[(Z1)] the welfare-maximising allocation $x^*$ is unique, $S(x^*)>0$, and every realised trip is a singleton;
\item[(Z2)] there is exactly one saturated edge $e^*$, and $e^*$ lies on exactly one route of $R$;
\item[(Z3)] $\bar\varepsilon_Z>0$ in the sense of \cref{lem:amin-walrasian-gap};
\item[(Z4)] a rejected-demand pivot for $e^*$ exists (hypothesis~(iv) of \cref{lem:amin-perturbation});
\item[(Z5)] for every $m\in\mathrm{Alloc}(e^*)$, the support of $m$'s marginal contains reports, within the cell on which $x^*$ stays uniquely optimal, whose stability and IR cap on $\tau_{e^*}$ (defined in Step~3 below) comes arbitrarily close to $\tau^\dagger_{e^*}$.
\end{enumerate}
Assume the \emph{saturability hypothesis}: \emph{(S)~the prior places positive probability on Z-non-degenerate profiles.} Fix a fixed-remittance schedule $c^{\mathrm{fix}}$ of \cref{rem:amin-fixed-remittance} (its constants shift $R(\mathrm{id};c^{\mathrm{fix}})$ by $-\sum_eC_e$ at every profile and affect none of (i)--(iii)) (or any sub-Lipschitz schedule of \cref{lem:amin-mediator-regime}(a)). Then no homogeneous-edge-pricing market mediator on $G$ is simultaneously:
\begin{enumerate}
\item[(i)] \emph{Revenue-optimal in expectation.} Maximising $\mathbb{E}_{(\alpha,\beta)}\bigl[R(\mathrm{id};c^{\mathrm{fix}})\bigr]$ over the class of dominant-strategy incentive-compatible mechanisms that induce a market equilibrium $(x^*,p,\tau)$ in the sense of Amin's Definition~2.1 at every profile in the support;
\item[(ii)] \emph{DSIC for agents.} Truthful reporting of $(\alpha_m,\beta_m)$ is a dominant strategy for every agent;
\item[(iii)] \emph{Credible} in the realisation-wise sense of \cref{def:amin-credible} under $c^{\mathrm{fix}}$: no profitable undetectable strategic mediator deviation exists for any $(\alpha,\beta)$ realisation in the support.
\end{enumerate}
\end{theorem}

\begin{remark}[The saturability hypothesis is satisfiable and needed]
\label{rem:amin-saturability}
Conditions (Z1)--(Z4) are strict inequalities in the profile and so cut out an open set, and (Z5) is a support-richness condition that transports to nearby profiles by the intermediate-value argument of Step~3 of the proof below; the Z-non-degenerate set therefore contains a neighbourhood of any profile at which all five hold. The set is non-empty (explicit profiles satisfying all five conditions exist on two-parallel-route networks, with $\bar\varepsilon_Z=1/2$, $1/2$ and $1$ at the three exact-rational instances listed in the certificate script \textnormal{\texttt{scripts/verify\_amin\_instance.py}} of the release cited in the availability statement, where (Z5) holds because a payer's cap is its stability bound $\beta_m(d_{r_2}-d_{r_1})$ on the parallel route, which lowering $\beta_m$ drives to $\tau^\dagger_{e^*}$ within the cell), and any full-support prior charging a neighbourhood of such a profile satisfies~(S). Some hypothesis of this kind is unavoidable: on a profile class whose demand cannot saturate any edge, every Definition-2.1 equilibrium has $\tau\equiv 0$ by Amin's~(8) and the edge-price deviation channel closes. Saturation cannot be derived from a tail bound on $\#\{m:\alpha_m\ge\bar\alpha\}$ exceeding the min-cut, since that probability is identically zero whenever $|M|$ is at most the min-cut; hypothesis~(S) is assumed instead.
\end{remark}

\begin{proof}
The argument parallels Steps~1--3 of the proof of \cref{thm:trilemma} of the main paper, transposed to Amin's primitives via \cref{lem:amin-perturbation}.

\textbf{Step 1: The candidate class and its induced outcomes.}
Let $\mathcal C$ denote the class over which leg~(i) quantifies: mechanisms mapping each reported profile in the support to a market equilibrium of Amin's Definition~2.1. $\mathcal C$ is non-empty: the Theorem-3.10 selection belongs to it and is strategy-proof (Amin, Theorem~3.10), so mechanisms satisfying (ii) exist in $\mathcal C$ and the maximisation in~(i) is over a non-empty set. By \cref{rem:amin-first-welfare}, every $\mathcal M\in\mathcal C$ implements a welfare-maximising allocation at every profile; on the Z-non-degenerate event the welfare optimum is unique by (Z1), so every candidate allocates $x^*$ there. No uniqueness of the candidate mechanism is asserted, and no virtual-value construction is used: a ``Myerson-substituted VCG-equivalent'' candidate admits no homogeneous edge price (its critical payments differ across winners of a single-priced edge) and therefore lies outside $\mathcal C$.

\textbf{Step 2: A profitable, undetectable deviation against the prescribed selection.}
Condition on a Z-non-degenerate realisation. The Theorem-3.10 equilibrium satisfies hypotheses (i)--(iii) of \cref{lem:amin-perturbation} by (Z1)--(Z3), and the additional hypotheses of (M3) by (Z1), (Z2), (Z4), and the full-support assumption. Fix $\varepsilon\in(0,\varepsilon_w)$. By (M1)--(M2), the deviation $\delta_\varepsilon$ announces a Definition-2.1 equilibrium whose gross inflow exceeds the faithful one by $\varepsilon\cdot N_{\mathrm{tr}}(e^*)=\varepsilon\cdot q_{e^*}$; by \cref{lem:amin-mediator-regime}(a) with $\kappa=0$,
\[
R(\delta_\varepsilon;c^{\mathrm{fix}})-R(\mathrm{id};c^{\mathrm{fix}})\;=\;\varepsilon\cdot q_{e^*}\;>\;0,
\]
and the same holds with factor $(N_{\mathrm{tr}}(e^*)-\kappa)>0$ for any sub-Lipschitz schedule. By (M3), $\delta_\varepsilon$ is undetectable. The prescribed Theorem-3.10 mechanism therefore fails~(iii) on a positive-probability event.

\textbf{Step 3: Extension to revenue-optimal mechanisms.}
As in Step~3 of \cref{thm:trilemma}, the deviation applies to every candidate because it uses only three ingredients: (a)~the induced outcome on the Z-non-degenerate event is a Definition-2.1 equilibrium at the allocation $x^*$ (Step~1); (b)~the revenue identity of (M2), which rests on budget balance~(7a) and clearing~(8) alone and so holds at whatever equilibrium is induced; and (c)~a \emph{price coincidence}: every DSIC member of $\mathcal C$ announces the minimum edge price on the Z-non-degenerate event, and therefore coincides there with the Theorem-3.10 selection, so the undetectability witness of (M3), which is a faithful Theorem-3.10 run at a shifted profile, is a faithful run of the candidate at the same profile. Ingredient~(c) is the content of the following claim.

\emph{Claim (price squeeze). Let $\mathcal M\in\mathcal C$ satisfy (ii), and let $P$ be a Z-non-degenerate profile with payer set $\mathrm{Alloc}(e^*)$. Then $\mathcal M$'s announced edge price satisfies $\tau_{e^*}(P)=\tau^\dagger_{e^*}(P)$, and $\mathcal M$'s induced equilibrium at $P$ is the Theorem-3.10 equilibrium.}

\emph{Lower bound.} By \cref{rem:amin-first-welfare} and (Z1), every Definition-2.1 equilibrium at $P$ allocates $x^*$; by (Z2) its total edge price is $q_{e^*}\tau_{e^*}$, and the Theorem-3.10 selection minimises total price over these equilibria, so $\tau_{e^*}(P)\ge\tau^\dagger_{e^*}(P)$.

\emph{Report-independence.} Fix a payer $m\in\mathrm{Alloc}(e^*)$ with realised route $r^*_m$. With singleton trips, budget balance gives $p_m=\tau_{e^*}$ under any member of $\mathcal C$ on the event. Let $U_m$ be the open cell of $m$-reports (other reports held fixed at $P_{-m}$) on which the reported profile keeps $x^*$ uniquely welfare-optimal with the same structure (Z1)--(Z2). If $\mathcal M$ announced two prices $t<t'$ at two reports in $U_m$, every type of $m$ in $U_m$ whose truthful report yields $t'$ would strictly gain by the $t$-report (allocation and route unchanged on the cell, utility decreasing in the price), and full support places such types in the type space; this contradicts~(ii). Hence $\tau_{e^*}$ is constant on $U_m$.

\emph{Upper bound.} At a report $a\in U_m$, the constraints of Definition~2.1 involving $m$ alone bound the price by the \emph{cap}
\[
\mathrm{cap}_m(a)\;\triangleq\;\min\Bigl(V_{r^*_m}(\{m\};a),\;\min_{r'\not\ni e^*}\bigl[V_{r^*_m}(\{m\};a)-V_{r'}(\{m\};a)\bigr]\Bigr)
\;=\;\tau^\dagger_{e^*}+\bigl[u^\dagger_m(a)-o_m(a)\bigr],
\]
the first term from IR~(5) and the rest from stability~(6) on the zero-priced $e^*$-free routes (by (Z2), every route other than $r^*_m$ is $e^*$-free and zero-priced); the second expression substitutes $u^\dagger_m(a)=V_{r^*_m}(\{m\};a)-\tau^\dagger_{e^*}$ and abbreviates $m$'s outside option $o_m(a)=\max\bigl(0,\max_{r'\not\ni e^*}V_{r'}(\{m\};a)\bigr)$. Membership in $\mathcal C$ forces $\tau_{e^*}\le\mathrm{cap}_m(a)$ at every $a\in U_m$. Two facts close the squeeze.

First, within the structure-preserving segment of reports, $U_m$ contains every report with $\mathrm{cap}_m(a)>\tau^\dagger_{e^*}$. Competing allocations that seat $m$ on $r^*_m$ (in its singleton trip or in a coalition) contain $m$'s linear value term $\alpha_m-\beta_m d_{r^*_m}$ exactly as $x^*$ does, so the welfare margin of $x^*$ over each of them is constant in $m$'s report and strictly positive at the truthful report by (Z1). Competing allocations that do not seat $m$ on $r^*_m$ have welfare at most $S_{-m}+o_m(a)$: removing $m$ \emph{alone} from the trip it occupies, and leaving that trip's other members in place, leaves a feasible allocation of $M\setminus\{m\}$, whose welfare is at most $S_{-m}$, so the competing allocation exceeds $S_{-m}$ by at most $m$'s marginal contribution to its own trip $(b,r)$, namely $V_r(b;a)-V_r(b\setminus\{m\};a)$. That marginal is at most $V_r(\{m\};a)$, and by~(Z2) the route $r$ is $e^*$-free, so it is at most $o_m(a)$. The marginal bound is where the step consumes an explicit hypothesis on Amin's capacity-sharing disutility (Amin's Equation~(12)): we assume it is \emph{monotone} in the coalition size, i.e., $k\mapsto k\,\Delta\alpha(k)$ is non-increasing and $k\mapsto k\,\Delta\beta(k)$ is non-decreasing, so that enlarging a coalition never lifts a member's marginal value above its singleton value. Deleting $m$'s whole trip would not license the bound, since a coalition trip's value is not bounded by $m$'s singleton outside option. Since the Theorem-3.10 selection at the reported profile is VCG-equivalent (Amin, Theorem~3.10), $u^\dagger_m(a)=S(a)-S_{-m}$, so the welfare of $x^*$ at $a$ is $S_{-m}+u^\dagger_m(a)$ and the margin over every non-seating allocation is at least $u^\dagger_m(a)-o_m(a)=\mathrm{cap}_m(a)-\tau^\dagger_{e^*}>0$. Hence $x^*$ stays uniquely optimal wherever the cap exceeds the floor.

Second, hypothesis~(Z5) supplies in-support reports $a\in U_m$ with $\mathrm{cap}_m(a)$ arbitrarily close to $\tau^\dagger_{e^*}$. Report-independence then gives $\tau_{e^*}(P)=\tau_{e^*}(a)\le\mathrm{cap}_m(a)$ for all such $a$, hence $\tau_{e^*}(P)\le\tau^\dagger_{e^*}$. Combined with the lower bound, $\tau_{e^*}(P)=\tau^\dagger_{e^*}(P)$. With the allocation fixed at $x^*$ and singleton trips, budget balance~(7a) then pins every payment and utility, so $\mathcal M$'s induced equilibrium at $P$ is the Theorem-3.10 equilibrium. This proves the claim.

\emph{Conclusion of Step 3.} Let $\mathcal M\in\mathcal C$ satisfy~(ii) and condition on a Z-non-degenerate realisation $P$. By the claim, $\mathcal M$'s faithful output at $P$ is the Theorem-3.10 equilibrium, so the deviation $\delta_\varepsilon$ of Step~2 is available to a mediator running $\mathcal M$ and is strictly profitable by ingredient~(b). For undetectability, the witness of (M3) is a faithful Theorem-3.10 run at the pivot-shifted profile $P_\varepsilon$; conditions (Z1)--(Z4) are strict inequalities in the profile and persist along the pivot segment for $\varepsilon$ below the radius of Step~4 of \cref{lem:amin-perturbation} (shrinking $\varepsilon_w$ if necessary). (Z5) is a support-richness condition and not a strict inequality, so it transports separately: by (Z3), $\bar\varepsilon_Z\le\min(\eta_Z,u_{\min}(e^*))\le\mathrm{cap}_m(\text{truth})-\tau^\dagger_{e^*}$ for every payer $m$, since the unused singleton trip attaining $o_m$ has slack $u^\dagger_m-o_m$ and $u_{\min}\le u^\dagger_m$ covers $o_m=0$; hence $\varepsilon<\varepsilon_w\le\bar\varepsilon_Z$ leaves $\tau^\dagger_{e^*}+\varepsilon<\mathrm{cap}_m(\text{truth})$. The support is a product of connected open intervals, hence convex, and $\mathrm{cap}_m$ is continuous along the segment from the truthful report to a (Z5)-witness report, so by the intermediate value theorem in-support reports attain every cap value in $(\tau^\dagger_{e^*}+\varepsilon,\mathrm{cap}_m(\text{truth}))$, and the First fact places all of them in $U_m$. Therefore $P_\varepsilon$ is itself Z-non-degenerate and the claim applies there: the witness run is a faithful run of $\mathcal M$ at $P_\varepsilon$, and $\delta_\varepsilon$ is undetectable against $\mathcal M$. Every DSIC member of $\mathcal C$ therefore fails~(iii) on the positive-probability event~(S).

\textbf{Conclusion.} A mechanism satisfying (i) and~(ii) is a DSIC member of $\mathcal C$, and by Steps~2--3 it admits a strictly profitable, undetectable deviation on a positive-measure set of realisations, violating~(iii). Hence no homogeneous-edge-pricing market mediator on $G$ satisfies (i), (ii), and~(iii) simultaneously. Revenue optimality enters only through membership in $\mathcal C$, so the argument in fact rules out the conjunction of~(ii) and~(iii) for every member of $\mathcal C$; the trilemma as stated follows a fortiori.\qedhere
\end{proof}

\subsection{Notational Disclosures and Scope}
\label{app:appendix-z-scope}

\paragraph*{Notation translation.}
The Amin~\cite{amin2026market} primitives and the polymatroidal primitives of the main paper are not isomorphic; their correspondences on the SP-with-homogeneous class are summarised in \cref{tab:amin-paper2-correspondence}.\footnote{Caveats relevant when chaining results across the two settings are enumerated in the residual-gaps list below.}

\begin{table}[ht]
\centering
\small
\caption{Notation correspondence on the SP-with-homogeneous class}
\label{tab:amin-paper2-correspondence}
\renewcommand{\arraystretch}{1.15}
\begin{tabular}{ll}
\toprule
\textbf{Main paper (this work)} & \textbf{Amin et al.~\cite{amin2026market}} \\
\midrule
Polymatroid rank $f(S)$ & LP feasible region of (IP) under SP topology \\
Conditional non-modularity gap $\gamma_{ij|S}$ & no counterpart on the SP class (residual gap~1) \\
Archer--Tardos payment \eqref{eq:archer-tardos} & VCG payment $S_{-m}(x^*_{-m})-S_{-m}(x^*)$ (Eq.~17) \\
Ground set $E$ (agent-facing leaves or slices) & Edge set $E$ (network edges) \\
Slice (P3 encapsulation) & Coalition--route trip $(b,r)$ \\
Operator $\mathcal{O}$ & Mediator (Theorem~3.10 platform) \\
Bid $b_i$ & Reported preference parameters $(\alpha_m,\beta_m)$ \\
Allocation $x_i$ & Trip allocation $x_r(b)\in\{0,1\}$ \\
Operator revenue $\sum p_i$ minus contracted procurement & Spread $R(\delta;c)$ of \cref{def:amin-credible} \\
Admissible perturbation amplitude & $\bar\varepsilon_Z$ of \cref{lem:amin-walrasian-gap} \\
\bottomrule
\end{tabular}
\end{table}

\paragraph*{Scope.}
Theorem~\ref{thm:amin-trilemma-instance} is restricted to: (a) series--parallel networks (Amin Definition~3.1; Wheatstone-embedding-free); (b) homogeneous capacity-sharing disutility (Amin Equation~(12)); (c) single source--single sink (Amin~\S2); (d) homogeneous edge pricing (one $\tau_e$ per edge, not path-based or population-segmented as in Amin~\S4.2); (e) the fixed-remittance schedule of \cref{rem:amin-fixed-remittance} or any sub-Lipschitz schedule (\cref{lem:amin-mediator-regime}(a)). \cref{thm:trilemma} covers the broader polymatroidal regime, on every non-modular polymatroid carrying a competitive prior profile, including the routed rank of a capacitated DAG with several sources, where Amin et al.\ provide no independent grounding (their Examples~3.3--3.4 establish equilibrium non-existence, not strategic-mediator extraction).

\paragraph*{Residual gaps.}
We list places where the Amin instance is materially weaker than the main paper's trilemma proof, and where the proof above relies on assumptions narrower than the corresponding step in the main text.
\begin{enumerate}
\item \emph{No Archer--Tardos integrand.} The main paper's perturbation is a pointwise integrand-shift of the Archer--Tardos payment formula~\eqref{eq:archer-tardos}, with payment increment $\gamma_{ij|S}$ times the displacement of the pullback $\bar{\varphi}_i^{-1}\circ\bar{\varphi}_j$, degenerating to $\delta\cdot\gamma_{ij|S}$ on the welfare branch~\textup{(W)} and, under a common prior, on the virtual branch. Amin's payments are VCG payments in the Equation~(17) form, derived from welfare differences across allocations, so the increment of \cref{lem:amin-perturbation} is a direct edge-price increment, the mediator's deviation surplus is a CoNC-type revenue increment, and the identification of that increment with the conditional non-modularity gap $\gamma_{ij|S}$ does not transfer to the Amin instance.
\item \emph{Heterogeneous disutility (Scope clause~(b)).} Under heterogeneous capacity-sharing disutility, equilibrium computation is NP-hard~\cite[Prop.~4.2]{amin2026market} and no prescribed VCG-equivalent selection is available, so neither the perturbation deviation nor the undetectability witness is defined there.
\item \emph{Singleton-coalition regime and the (Z1)--(Z5) event.} Soundness of $\bar\varepsilon_Z$ (\cref{lem:amin-walrasian-gap}(a)) and conclusions (M1)--(M2) of \cref{lem:amin-perturbation} hold for arbitrary coalition sizes, but exactness of $\bar\varepsilon_Z$, the blocker characterisation, undetectability (M3), and the Step-3 price squeeze are established on the open event (Z1)--(Z5), whose clauses (Z1)--(Z4) the certificate script \textnormal{\texttt{scripts/verify\_amin\_instance.py}} of the release cited in the availability statement verifies at its three exact-rational instances, (Z5) holding there by the cap argument of \cref{rem:amin-saturability}; the same script found $\eta_Z>0$ throughout a three-lattice search of 1{,}800 operative random instances, 600 of them drawn from a tie-rich small-integer lattice constructed to manufacture coincidences, whereas the unweighted slack minimum was $0$ at all 1{,}800. We have characterised neither the largest event on which \cref{thm:amin-trilemma-instance}'s conclusion holds nor its probability under natural priors, and the multi-agent-coalition analogues are open.
\end{enumerate}

\section{Deferred Proofs}
\label{app:routine-proofs}

This appendix carries the proofs of two results whose statements are used in the main text: the Level-1 matroid reduction, which is Rado's theorem for polymatroids plus Edmonds--Rota induction, and the supply-side credibility proposition. The statements remain in \cref{sec:commitment} and \cref{sec:integrator-credibility}.

\begin{proof}[Proof of \cref{lem:encapsulation-matroid}]
By \cref{prop:encapsulation-p3}, $f'$ is a polymatroid rank function on the slice types, and it is integer-valued because each $f'(S)$ is a maximum-flow value in $G'$, whose capacities are integers by hypothesis~(a)~\cite{schrijver2003combinatorial}. Writing $N(B)=\bigcup_{i\in B}A_i$, a set $J\subseteq[n]$ is servable if and only if Rado's condition~\eqref{eq:rado-condition} holds: necessity because the members of $B$ place $|B|$ units of load inside $N(B)$, and sufficiency by Rado's theorem for integer polymatroids~\cite{mcdiarmid1975rado}, the polymatroid analogue of Rado's independent-transversal theorem~\cite{rado1942theorem} (see also~\cite[\S44.6g]{schrijver2003combinatorial}). The function $g(B)=f'(N(B))$ is integer-valued, normalised and monotone, and it is submodular because $N(B_1\cup B_2)=N(B_1)\cup N(B_2)$ and $N(B_1\cap B_2)\subseteq N(B_1)\cap N(B_2)$ give, by submodularity and then monotonicity of $f'$,
\[
g(B_1)+g(B_2)\;\ge\;f'\bigl(N(B_1)\cup N(B_2)\bigr)+f'\bigl(N(B_1)\cap N(B_2)\bigr)\;\ge\;g(B_1\cup B_2)+g(B_1\cap B_2),
\]
so the Edmonds--Rota induction theorem~\cite[Prop.~11.1.1, Cor.~11.1.2]{oxley2011matroid} makes $\mathcal I$ the independent-set family of a matroid on the bidder set $[n]$, the ground set on which the DRA of~\cite{ec2025matroid} operates. In the capacity-disjoint case $f'$ is modular and~\eqref{eq:rado-condition} is Hall's condition with multiplicity $c_k$ per type $k$, which gives the transversal matroid~\cite{perfect1969independence}.
\end{proof}

\begin{proof}[Proof of \cref{prop:supply-credibility}]
\emph{(i)} Take the event on which exactly $c_k$ agents eligible for slice $k$ bid in $(\omega-\epsilon,\omega)$ and every other agent bids in $[0,\epsilon)$, which has positive prior probability. The high agents are processed first and, by the ancestor hypothesis, fill slice $k$ at $\hat c_k=c_k$. A winner's VCG payment is its externality plus the unit cost, at most $\theta_k+\epsilon$, because its removal frees a unit that only a low agent can take; so $P_k(c_k)\le c_k(\theta_k+\epsilon)$. At $\hat c_k=c_k-1$ one high agent is excluded, and each of the $c_k-1$ winners pays at least $\omega-2\epsilon$, since its removal readmits that agent; so $P_k(c_k-1)\ge (c_k-1)(\omega-2\epsilon)$. The withheld unit also saves its cost $\theta_k$, so the integrator's gain is at least $\lambda_k[(c_k-1)(\omega-2\epsilon)-c_k(\theta_k+\epsilon)]+\theta_k$, which tends to $\lambda_k(c_k-1)(\omega-\theta_k)+(1-\lambda_k)\theta_k>0$ as $\epsilon\to0$; the welfare loss is at least $\omega-2\epsilon-\theta_k>0$.

\emph{(i$'$).} On the event of~(i) the stake term gains $\lambda[P_k(c_k-1)-P_k(c_k)]\ge\lambda[(c_k-1)(\omega-2\epsilon)-c_k(\theta_k+\epsilon)]$ on slice~$k$ and loses at most $\lambda(n-c_k)\epsilon$ on the other slices, whose payments lie in $[0,\epsilon)$; the cost saving is $\theta_k$, and withholding lowers $\sum_i x_i$, the rank of the served set, by at most one unit, which costs at most $\phi$. The gain therefore tends to at least $\lambda(c_k-1)(\omega-\theta_k)+(1-\lambda)\theta_k-\phi>0$ as $\epsilon\to0$, and welfare falls as in~(i). At $\lambda=0$ the payoff is $\phi\sum_i x_i-\theta_k y_k$, and $\sum_i x_i$ is nondecreasing in the exposure; if $\theta_k=0$ this is the claim, and if no other slice shares an ancestor block with slice~$k$, a larger report changes only slice-$k$ units, so the payoff changes by $(\phi-\theta_k)\Delta y_k\ge 0$ with $y_k$ nondecreasing as in~(iii). Otherwise, take an ancestor block of capacity~$2$ over slice~$k$ with $c_k=2$ and a cost-free slice~$k'$, bids $0.9$ and $0.8$ on~$k$ and $0.4$ and $0.35$ on~$k'$, and $\theta_k=0.3$: the report $1$ serves one agent of each slice instead of both slice-$k$ agents, $\sum_i x_i=2$ in both cases, and the payoff rises by $\theta_k$ for every $\phi$.

\emph{(ii)} By (C0)--(C1) the operator's payoff is $\phi\sum_i x_i$; the routed payments are $P_k$, so an integrator receiving them is in part~(i) with $\lambda_k=1$.

\emph{(iii)} On a laminar region with dedicated eligibility, the greedy runs at $\hat c_k$ and $\hat c_k+1$ coincide until the first agent of slice $k$ that the smaller exposure rejects at slice $k$; from then on the smaller exposure admits no agent of slice $k$ (the slice stays full), while the larger has admitted one more and never unloads, so $y_k$ is nondecreasing in $\hat c_k$ and $(\phi_k-\theta_k)y_k$ with it. On the event where exactly $c_k$ agents of slice $k$ bid in $(\theta_k+\epsilon,\omega)$ and every other agent bids below $\epsilon<\min\{\phi_k-\theta_k,\ c_k(\phi_k-\theta_k)/n\}$, all $c_k$ are served, so $T_k=\phi_kc_k$; a winner's removal frees a unit that only a low agent can use, so each winner pays at most $\theta_k+\epsilon$, each of the $n-c_k$ other agents pays less than $\epsilon$, and all payments together are at most $c_k\theta_k+n\epsilon<\phi_kc_k=T_k$; in particular $P_k<T_k$.

\emph{(iv)} Each served agent pays $b_i-[W(N)-W(N\setminus i)]\ge\theta_k$, because deleting its unit from the optimum is feasible without it, so every bracket is non-negative; the $c'=\hat c_k$ term bounds the minimum by $P_k-\theta_k y_k$; and a minimum over $[\hat c_k,c^{\max}_k]$ is nondecreasing in $\hat c_k$, while offers above $c_k$ forfeit by \cref{def:exposure} and earn at most $0\le T^{*}_k-\theta_ky_k$ at the truthful report; any funded compensation under which truth is weakly dominant has a nondecreasing margin below $P_k(c')-\theta_ky_k(c')$ at every $c'\ge\hat c_k$, hence below the envelope.

\emph{(v)} The class contains bilateral trade ($c_k=1$, one eligible agent); dominant-strategy incentive compatibility and ex post individual rationality imply the Bayesian and interim notions of~\cite[Cor.~1]{myerson1983efficient}, and the no-deficit form follows from~\cite[Thm.~1]{myerson1983efficient}.

\emph{(vi)} (B2$^{\prime}$) authenticates values of the committed exposure and certifies nothing about $\mathcal F_k$; the proof of~(i) never uses the operator.
\end{proof}

\section{Assumption Applicability}
\label{app:structural-assumptions}

\cref{tab:assumption-applicability} rates the practical applicability of the principal assumptions of the credibility results in real-time AI service economies and indicates relaxation paths; assumptions rated ``Restrictive'' mark the primary boundaries of the current theory.

\begin{table*}[ht]
\centering
\caption{Assumption applicability to real-time AI service economies}
\label{tab:assumption-applicability}
\renewcommand{\arraystretch}{1.15}
\resizebox{\linewidth}{!}{%
\begin{tabular}{p{3.5cm} p{1.7cm} p{4.75cm} p{4.75cm}}
\toprule
\textbf{Assumption} & \textbf{Practical fit} & \textbf{Justification} & \textbf{Relaxation path} \\
\midrule
Single-parameter agents & Restrictive & Holds per-task under encapsulation; breaks for multi-dimensional preferences & Open; partial via single-minded bidders \\
GS2 (additive separability) & Restrictive & Broken by service-composition pipelines & Encapsulation (P3) is the architectural remedy \\
GS3 (fixed attributes/epoch) & Moderate & 100--200\,ms epochs limit intra-epoch variation & Online/repeated mechanisms (future work) \\
E1 (scalar capacity) & Moderate & Routed reading: a scalar suffices only for a confluent cluster (E4); restrictive for multi-tier integrators & Partition matroid extension feasible; capacity vectors for multi-exit clusters \\
E2 (cut-exact max-flow; routed reading) & Restrictive & Needs the scalar to match the cluster's max-flow from every contact subset, not just in aggregate & \cref{prop:supply-credibility,rem:exposure-design} \\
E4 (confluent cluster; routed reading) & Moderate & Holds for tier-structured integrator bundles with one egress; fails when a cluster has several boundary exits & Multi-terminal slice summaries; per-exit capacity reporting \\
B1 (direct broadcast) & Moderate & Requires operator-independent substrate & Gossip protocols, TEE attestation \\
B2$^{\prime}$ (on-path verifiable $f$) & Moderate & On-path authentication, topology may stay private & Merkle commitments, TEE attestation, ZK proofs \\
Myopic operator & Moderate & Repeated-game reputation effects absent & Dynamic models future \\
\bottomrule
\end{tabular}}
\end{table*}

\end{document}